\documentclass[11pt,letterpaper, twoside]{article}

\makeatletter

\renewcommand\normalsize{%
\@setfontsize\normalsize{11}{19}%
  \abovedisplayskip 1.0em plus 0.3em minus 0.5em
  \abovedisplayshortskip \z@ plus 0.3em
  \belowdisplayshortskip 0.6em plus 0.3em minus 0.3em
  \belowdisplayskip \abovedisplayskip
}
\normalsize

\makeatother

\AtBeginDocument{%
  \thickmuskip=3.5mu plus 2mu
  \thinmuskip=2.5mu
  \medmuskip=3mu plus 2mu minus 1.5mu
  \frenchspacing
  \sloppy
  \flushbottom
}

\usepackage{fancyhdr}

\newcommand{\runtitle}{Test-then-Punish: A Statistical Approach to Repeated Games}   

\usepackage{amssymb}
\usepackage{booktabs}
\usepackage{pifont}    
\usepackage{xcolor}    
\usepackage{tabularx}
\usepackage{makecell}
\newcommand{\cmark}{\textcolor{green!60!black}{\ding{51}}} 
\newcommand{\xmark}{\textcolor{red!70!black}{\ding{55}}}   
\usepackage{mathrsfs}

\usepackage{todonotes}
\usepackage{notation}
\usepackage[inline]{enumitem}
\usepackage{booktabs} 
\usepackage[ruled]{algorithm2e} 

\SetAlFnt{\small}
\SetAlCapFnt{\small}
\SetAlCapNameFnt{\small}
\SetAlCapHSkip{0pt}
\IncMargin{-\parindent}

\usepackage[authoryear,round]{natbib} 
\definecolor{outerS}{HTML}{F3E8D0} 
\definecolor{innerS}{HTML}{CBE7D6} 

\usepackage{amsmath, amsthm, amsfonts}
\usepackage{cleveref}
\usepackage{thmtools}

\theoremstyle{plain}
\newtheorem{theorem}{Theorem}

\newtheorem{lemma}{Lemma}
\newtheorem{corollary}{Corollary}

\theoremstyle{definition}
\newtheorem{definition}{Definition}

\newtheorem{condition}{Condition}

\theoremstyle{remark}
\newtheorem{remark}{Remark}

\usepackage[T1]{fontenc}
\usepackage[utf8]{inputenc} 
\usepackage{lmodern}

\crefname{theorem}{theorem}{theorems}
\Crefname{theorem}{Theorem}{Theorems}
\crefname{proposition}{proposition}{propositions}
\Crefname{proposition}{Proposition}{Propositions}
\crefname{lemma}{lemma}{lemmas}
\Crefname{lemma}{Lemma}{Lemmas}
\crefname{corollary}{corollary}{corollaries}
\Crefname{corollary}{Corollary}{Corollaries}
\crefname{definition}{definition}{definitions}
\Crefname{definition}{Definition}{Definitions}
\crefname{assumption}{assumption}{assumptions}
\Crefname{assumption}{Assumption}{Assumptions}
\crefname{example}{example}{examples}
\Crefname{example}{Example}{Examples}
\crefname{remark}{remark}{remarks}
\Crefname{remark}{Remark}{Remarks}

\usepackage{authblk}

\author[1, *]{Aymeric Capitaine}
\author[1, *]{Antoine Scheid}
\author[2]{Etienne Boursier}
\author[1]{Alain Durmus}
\author[3]{Michael I. Jordan}

\affil[1]{CMAP, CNRS, École polytechnique, France}
\affil[2]{Inria Saclay, Université Paris Saclay, LMO, France}
\affil[3]{Inria Paris, Ecole Normale Supérieure, PSL, France, University of California, Berkeley, USA}
\affil[*]{Equal Contribution}

\begin{document}


\title{Test-then-Punish: A Statistical Approach to Repeated Games}


\maketitle



\begin{abstract}
We study discounted infinitely repeated games in which players agree on a cooperative mixed action profile but, at each step, observe only the realized pure actions. This form of imperfect monitoring breaks classical trigger strategies, since deviations cannot be identified with certainty. To address this problem, we study how hypothesis testing can be used to sustain cooperation. First, we develop a framework that embeds statistical inference directly into strategic behavior. We introduce relaxed equilibrium notions that allow players to ignore vanishing-probability histories arising from rare but extreme realizations of the monitoring process. Within this framework, we formalize a generic test-then-punish strategy: players commit ex ante to a cooperative mixed action profile, continuously test whether observed play is consistent with this prescription, and permanently switch to punishment once sufficient statistical evidence of deviation accumulates. Under mild conditions on the testing procedure, this construction sustains any feasible and individually rational payoff for sufficiently patient players, yielding a Folk-theorem–type result under imperfect monitoring.
We then propose two explicit implementations of this strategy. The first relies on anytime-valid sequential tests, providing uniform control of Type I error over an infinite horizon and a finite expected detection time for payoff-relevant deviations. However, this strategy only accounts for stationary deviations and yields a Nash equilibrium. The second uses testing over batches with a fixed size, accommodating arbitrary deviations and achieving subgame-perfect Nash equilibrium, at the cost of losing global anytime guarantees on false punishments. Together, these results clarify a fundamental trade-off between statistical soundness and game-theoretic robustness, offering a principled foundation for data-driven cooperation in repeated strategic environments.
\end{abstract}
\noindent\rule{\textwidth}{0.4pt}
\vspace{1em}

\section{Introduction}\label{section:introduction}



Sustaining cooperation among strategic agents is central in many economic sectors, ranging from industrial organization and financial markets to environmental regulation and international trade. A common solution is given by punishment-based strategies: players agree on a cooperative course of action and continue to follow it as long as they observe that their opponents do the same. When a deviation is detected, they switch to a punishment phase. Provided that the threat of punishment is credible and sufficiently severe relative to the short-run gains from deviation, no player has an incentive to depart from the cooperative strategy.

Central to the literature on repeated games \citep{fudenberg1991game, mailath2006repeated}, the \emph{Folk Theorem} \citep{friedman1971non} formalizes the effectiveness of such punishment-based strategies. This theorem asserts that \textit{any} feasible and individually rational payoff profile of a given stage game can be sustained as a subgame-perfect Nash equilibrium \citep{moore1988subgame} when the game is repeated infinitely and players are sufficiently patient. This foundational result has been established both under perfect monitoring, where players are assumed to observe their opponents’ strategies  exactly, and under some imperfect monitoring setups, in which players observe only signals that are imperfectly correlated with those strategies.

The latter setting is of considerable practical interest, as economic actors are often obliged to operate with imperfect information structures. However, this strategy is notoriously difficult to analyze and
this has led to the development of a variety of approaches that primarily deliver existence results rather than providing implementable strategies. As a consequence, these approaches generally do not yield procedures that can directly exploit modern learning techniques and the abundance of data available in many real-world environments. This is all the more unfortunate given that real-world economic agents often make use of statistical analysis to implement punishment-based cooperative strategies, as illustrated by the following two examples.

A first concrete example is offered by financial auditing. Firms and auditors interact repeatedly over several periods, with firms expected to report financial statements truthfully and where deviations take the form of misreporting or manipulation. Because such behavior cannot typically be detected with certainty in a single period, compliance is assessed using statistical procedures. As formalized by the Public Company Accounting Oversight Board  \citep{as2305} in the United States, auditors form quantitative expectations, based on trends or regression-based models, and investigate outcomes that are statistically inconsistent with these benchmarks. When sufficient evidence accumulates, deviations trigger sanctions or legal liability. The prospect of delayed but severe punishment deters short-run deviations and sustains truthful reporting despite noisy monitoring.

Anti-doping enforcement in professional sports provides another example of cooperation sustained through explicit statistical testing in a repeated interaction. Athletes repeatedly choose whether to compete cleanly or not, while individual actions are unobservable and biological measurements are noisy. Under the Athlete Biological Passport, administered by the World Anti-Doping Agency, compliance is assessed using longitudinal statistical models that test whether an athlete’s biomarker profile is statistically consistent with natural variation. Sanctions are triggered when accumulated evidence renders clean behavior sufficiently unlikely \citep{wada}. Here again, the threat of punishment enforces cooperation.

The gap between the literature on repeated games with imperfect monitoring--which largely abstracts away from statistically grounded cooperative strategies--and real-world economic interactions, where data and learning techniques are pervasive, raises the following question: \emph{how can we formalize the use of hypothesis testing to sustain cooperation in repeated games}?

Additionally, statistical monitoring would be even more attractive if it could exploit prediction and inference methods. This leads us to a second, more practice-focused question: \emph{Is there a straightforward testing procedure that could realistically be employed in real-world instances of this problem}?

We answer the first question by developing a statistical framework for repeated games with imperfect public  monitoring, which reposes on two key concepts. First, due to the probabilistic nature of statistical inference, we introduce equilibrium notions that relax Nash and subgame--perfect Nash equilibria by allowing players to disregard rare-probability histories and tail events arising from extreme realizations of the monitoring process. Second, we formalize a generic \textit{test-then-punish} strategy: before play begins, players commit ex ante to a cooperative mixed action profile that implements a target feasible payoff. Then, observing only the realized pure actions, they continuously test the null hypothesis that each opponent is adhering to the prescribed mixed strategy. Cooperation persists as long as sequential tests fail to reject the null hypothesis. Once a test detects sufficient evidence of deviation, 
players switch permanently to a retaliatory punishment profile. Under mild conditions on the testing procedure---concretely met by the tests we design later---this template sustains any feasible and individually rational payoff for sufficiently patient players, yielding a Folk-Theorem–type result under imperfect public monitoring. We emphasize that our statistical approach allows us to bypass the standard decomposability and self-generation techniques typically used to analyze equilibria under imperfect monitoring \citep{abreu1990toward}. Instead, we build on familiar proof techniques from the perfect monitoring literature and address the stochasticity generated by random signals using probabilistic tools. The interest of such approach is to derive a simple and strategy, relying on a punishing rule with a timing that is explicitely controlled.

We address the second question by introducing two explicit testing strategies. The first one leverages anytime valid testing and e-processes \citep{ramdas2024hypothesis}. This framework enables the testing of a hypothesis across all game rounds simultaneously, maintaining validity regardless of the choice of stopping time. Assuming that players' strategies are stationary, we show that it is possible to construct a strategy that almost attains any feasible, rational payoff profile as a Nash equilibrium. Two statistical errors are critical in this context. On the one hand, a Type I error occurs when the null hypothesis is rejected despite being true, which corresponds to entering the punishment phase while players are still cooperating. Anytime-valid testing ensures that the probability of this event remains below a pre-specified threshold. We emphasize that our approach is novel in providing such a guarantee for games, contrary to the belief-based \citep{sekiguchi1997efficiency} or belief-free \citep{ely2005belief} methods typically employed in the literature which overlook this risk. On the other hand, a Type II error occurs when players fail to detect a deviation and do not trigger punishment. We establish a finite bound on the expected stopping time of cooperation when a player deviates, as a function of the magnitude of their deviation. Our result on an upper bound for the expected stopping time new in the testing literature and may be of independent interest.

The second strategy is designed to overcome the two limitations of the previous approach that only supports stationary deviations and solely provides Nash equilibrium (rather than subgame-perfect equilibrium). Instead of continuously testing a single null across all periods, we introduce a batch test-then-punish strategy that partitions the game into consecutive blocks of rounds. At the end of each batch, players run a statistical test on the empirical distribution of realized actions in the block to assess consistency with the prescribed cooperative mixed strategy. If the null is rejected in some batch, players switch permanently to the punishment profile. We show that this construction can approximate any feasible and individually rational payoff and yields a subgame-perfect Nash equilibrium. However, the batch procedure does not provide uniform, explicit control of Type I and Type II errors over the entire horizon.

Overall, there is a trade-off between the two testing procedures we propose. On the one hand, the anytime approach enjoys explicit Type I and Type II errors guarantees, making it an attractive option when such requirements are desired, typically for reasons of risk aversion or fairness. However, the anytime test-then-punish is weaker from a game-theoretic standpoint since it only prevents stationary deviations and does not enjoy subgame perfection. On the contrary, the batch test-then-punish handles arbitrary deviations, and is a subgame perfect Nash equilibrium, but  does not enjoy statistical guarantees regarding false positives. We summarize this discussion in \Cref{tab:properties}. Ultimately, the choice between the anytime and batch procedures comes down to whether players value statistical soundness or strong game-theoretic guarantees.

\begin{table}[t]
\centering
\renewcommand{\arraystretch}{1}
\setlength{\tabcolsep}{8pt}
\begin{tabular}{lcccc}
\toprule
 
& \makecell{Equilibrium\\type}
& \makecell{Non-stationary\\strategies}
& \makecell{Type I Error\\guarantee}
& \makecell{Type II Error\\guarantee} \\
\midrule
Anytime (\Cref{section:folk_anytime}) & Nash & \xmark & \cmark & \cmark \\
Batch (\Cref{section:folk_batch})   & subgame perfect Nash & \cmark & \xmark & \cmark \\
\bottomrule
\end{tabular}
\caption{Comparison of properties across methods.}
\label{tab:properties}
\end{table}

In summary, our main contributions include:
\begin{enumerate}
    \item \textbf{Statistical monitoring framework.} We formalize repeated strategic interaction when players prescribe mixed actions but only observe realized pure actions. To account for rare events in statistical monitoring, we extend Nash/subgame-perfect equilibria to account for histories that occur with small probability.
    \item \textbf{Generic test-then-punish strategy.} We propose a generic strategy where players commit ex ante to a cooperative mixed profile targeting a feasible payoff, test opponents’ action against this prescription, and switch to punishment upon detecting a deviation. Under mild conditions on the sequential test, this recovers a Folk-Theorem result under imperfect monitoring.
    \item \textbf{Explicit testing strategies.} Finally, we instantiate the template with two sequential tests: the first one controls false punishment over an infinite horizon while still detecting deviations in finite time. The second block-based variant handles any form of deviation and yields stronger high-probability sequential guarantees, at the cost of weaker global anytime Type I error control.
\end{enumerate}

\subsection*{Related work}

Our work stands at the intersection of distinct branches of literature: repeated games, algorithmic game theory, and hypothesis testing.

First, the Folk Theorem has been the subject of an extensive literature. The idea that repetition of a stage game expands the set of payoff profiles that can be sustained under a Nash equilibrium traces back at least to \citet{friedman1971non}, who assumes perfect public monitoring and time-averaged payoffs. The result is shown to hold under the same conditions with subgame perfect Nash equilibrium \citep{rubinstein1980strong}, and it is further extended to discounted payoffs \citep{fudenberg1986folk}. The Folk Theorem is also established under \textit{imperfect monitoring} \citep{aumann1995repeated}, where players do not observe the actual mixed strategies chosen by their opponents but only public signals correlated with them. By using decomposability and self-generation techniques \citep{abreu1990toward}, it is possible to show, albeit non-constructively, that any feasible rational payoff can be sustained under a subgame perfect Nash equilibrium provided pairwise identifiability holds \citep{fudenberg1994folk}. A further refinement arises with private signals \citep{mailath2002repeated}, where players no longer observe a common signal but instead receive different and potentially conflicting signals, which hinders immediate coordination. Several approaches have been proposed to address this difficulty, including allowing for communication among players \citep{compte1998communication}, assuming almost perfect private signals \citep{ely2002robust, horner2006private} or adopting belief-free strategies \citep{ely2005belief}. In addition, \citet{gossner1995folk} studies the sustainment of the punishment phase played as a mixed strategy while \citet{sugaya2022folk} establishes an existence Folk Theorem under private monitoring by using long review phases and implicit communication through actions to coordinate punishments. Recent advances have expanded the scope of repeated game analysis in several directions. These include uncertain repeated games, in which players may face different stage games over time \citep{LACLAU2012711}, imperfect monitoring in finite games \citep{horner25}, and environments with switching costs \citep{TSODIKOVICH2024137}. Folk-type results have also been established in more complex strategic settings, such as competing mechanisms \citep{ATTAR202179}, matching markets \citep{DEB20161}, stochastic games \citep{AIBA201458}, and network games \citep{LACLAU2012711}. Finally, recent work has connected the Folk theorem to algorithmic collusion \citep{CARTEA20261}, a rapidly growing literature that studies whether collusive behavior can emerge from interactions among algorithms \citep{calvano20}.

Our paper contributes to a growing literature at the intersection of economics and machine learning \citep{jordan2025collectivist}, aiming to move beyond traditional equilibrium analysis under uncertainty by replacing implicit belief-updating rules with implementable learning algorithms \citep{erev1998predicting, cesa2006prediction, roughgarden2010algorithmic}. In this line of research, players adapt their actions to observed outcomes using data-driven procedures. This approach has been successfully applied to the study of classical microeconomic problems, including dynamic pricing \citep{bertsimas2006dynamic, mueller2019low}, matching \citep{liu2020competing, johari2021matching}
the management of externalities \citep{scheid2024learning}, adverse selection and moral hazard \citep{bates2022principal, capitaine2024unravelling}, bidding in auctions \citep{morgenstern2016learning, feng2018learning},
social choice \citep{brandt2012computational}
among others. An illustrative example are principal-agent problems that recently benefited from learning techniques to enforce contracts \citep{fallah2024contract, collina2024repeated}, to infer the players' preferences \citep{kolumbus2024contracting, scheid2025online} or to play against unknown opponents \citep{mansour2020bayesian, kolumbus2022and, arunachaleswaran2025learning}. A key goal of this literature is to provide finite-time and finite-sample guarantees, such as explicit convergence rates to equilibria or to socially optimal outcomes \citep{syrgkanis2015fast}. Focusing on infinite horizon repeated games, our work advances this program by introducing explicit, test-based strategies to sustain cooperation.

Finally, our proposed deviation-detection mechanism crucially builds on the classical theory of statistical hypothesis testing \citep{fisher1934statistical, neyman1933ix, wald1992sequential, lehmann2005testing}.
From Bayesian formulations of evidence \citep{andraszewicz2015introduction} to general optimality-driven treatments of test construction \citep{lehmann2005testing} and the modern theory of multiple hypothesis testing \citep{benjamini1995controlling, shaffer1995multiple}, a substantial body of work has shaped contemporary hypothesis testing. Our technical approach relies more specifically on the recent theory of \emph{anytime-valid} sequential tests---most notably the notion of \emph{e-processes} \citep{shafer2005probability, ramdas2022testing, howard2021timea}. E-processes provide martingale-based measures of statistical evidence that can be updated online, without fixing a sample size in advance, while remaining valid under arbitrary stopping rules. In particular, an e-process (a nonnegative process with mean upper bounded by one) directly yields uniform Type I error control over an infinite horizon via Ville's inequality \citep{doob1939jean}. While anytime-valid testing is increasingly used in modern sequential decision problems, our work is, to the best of our knowledge, the first to leverage anytime-valid inference as a core enforcement primitive inside a repeated, game-theoretic test-and-punish construction.

\subsection*{Organization}

The paper is organized as follows. In \Cref{section:setting}, we introduce our framework and recall the Folk Theorem under perfect monitoring. In \Cref{section:folk_anytime}, we introduce the anytime test-then-punish strategy and show that, provided it is possible to control Type I and Type II errors, it leads to a Folk Theorem under imperfect monitoring. We then  explicitly design sequential tests satisfying these conditions thanks to e-processes. Finally, in \Cref{section:folk_batch}, we introduce the \textit{batch} test-then-punish and establish, under mild conditions, another version of the Folk Theorem under imperfect monitoring. Here again, we exhibit simple sequential tests which satisfy the aforementioned conditions.

\section{Setting and the Folk Theorem}\label{section:setting}


\paragraph{Notations}For any finite measurable space $\cA$, we denote by $\Delta(\cA)$ the set of probability measures over $\cA$. In particular, if $\cA=\{a_1,\ldots,a_K\}$, for any $w\in\Delta(\cA)$, we denote by $w[a_\ell]\in[0,1]$ the mass put by $w$ on $a_{\ell}\in\cA$. If $\cA=\cA^1\times\ldots\times\cA^N$ is the product of $N>0$ spaces, we write for any $i\in[N]$ $\cA^{-i}=\cA^1\times\ldots\times\cA^{i-1}\times\cA^{i+1}\times\ldots\times\cA^N$ so $\cA=\cA^i\times\cA^{-i}.$ Likewise, if  $\bw=w^1\otimes\ldots\otimes w^N$ is a product measure on $\cA$, for any $i\in[N]$ we denote by $w^{-i}=w^1\otimes\ldots\otimes w^{i-1}\otimes w^{i+1}\otimes\ldots\otimes w^N$ so $\bw=w^i\otimes w^{-i}$. Finally for any finite set $\cX$, we denote by $\mathsf{conv}(\mathcal{X})$ the convex hull of $\cX$.

\subsection{Framework}


We consider a repeated game with an infinite horizon and $N>0$ players. Each player $i\in[N]$ has access to a pure action set $\cA^i=\{a^i_1,\ldots,a^i_K\}$ such that $\Card(\cA^i) = K>0$. Their preferences are represented by utility functions $u^i \colon \cA \rightarrow[0,1]$, where $\cA=\cA^1\times\ldots\times\cA^N$. At any round $t\geq 0$, each player $i$ picks a  mixed strategy $w^i_t \in\Delta(\cA^i)$, so their expected payoff reads $u^i(w^i_t, w^{-i}_t)$ where, by a slight abuse of notation, we write $u^i (w^i_t , w^{-i}_t)=\sum_{a^i, a^{-i}} u ^i (a^i , a^{-i})\, w^i_{t}[a^i]\,w^{-i}_t[a^{-i}]$. At any time $t\geq0$, denoting by $\cH_t$ the $t$-ary Cartesian product of $\cA$, players have access to a public history $\hist_t\in\cH_t$ containing past realized pure actions, denoted by $\hist_t =\{(A^i_0,A^{-i}_0),\ldots,(A^i_{t-1},A^{-i}_{t-1})\}$, which they can use to inform their strategies. Observe our set of signals is the set of pure actions $\cA$, which corresponds to the canonical signal space \citep{horner2006private}. Formally, with $\cH=\cup_{t\geq0}\cH_t$, the sequence of moves $(w^i _0 , w^i _1,\ldots)$ from player $i$ is generated according to a (behavior) strategy
\begin{align*}
s^i \colon \cH \rightarrow \Delta(\cA^i) \eqsp,
\end{align*}
which maps histories to distributions over pure actions.  We denote by $\cS^i$ the set of such strategies for player $i$, and $\cS=\cS^1\times\ldots\times\cS^N$. Note that we focus on public strategies, which only depend on public histories $\hist_t \in\cH$ and  not private ones\footnote{For any player $i\in[N]$, we could indeed define  private histories $\tilde{\hist}_t = \{(w^i_0, A^i_0, A^{-i}_0),\ldots,(w^i_{t-1}, A^i_{t-1}, A^{-i}_{t-1})\}\in\tilde{\cH}_t$ and consider private strategies $\tilde{s}^i:\cup_{t\geq0}\tilde{\cH}\rightarrow \Delta(\cA^i)$. This however would break the recursive structure of the repeated game and significantly complicate the analysis \cite{mailath2006repeated}. Extending our study to this case is an interesting future research direction.}. Then, the repeated game unfolds as follows. Starting from a history $\hist_0 = \emptyset$, each player $i\in[N]$ picks a strategy $s^i\in\cS^i$ with $s^i (\hist_0)=w^i_0$, and for any $t\geq 0$
\begin{enumerate}
    \item They play a mixed strategy $w^i_t = s^i (\hist_t)$,
    \item A pure action $A^i _t \sim w^i _t$ is drawn from it for each player $i \in N$,
    \item Players receive a payoff $u^i (A^i_t, A^{-i}_t)$ and observe $A^{-i}_t \in\cA^{-i}$. The history for the next round is then updated as $\hist_{t+1}=\hist_t \cup \sigma(\{A^i_t, A^{-i}_t\})$.
\end{enumerate}
Observe that monitoring is imperfect because of step 3, where public histories only contain pure actions and not mixed strategies. For any strategy profile $\bs=(s^1,\ldots,s^N)\in\cS$, the strategies and the players' independent randomizations induce a unique probability measure $\P^{\bs}$ over the space of infinite action histories $\cA^{\mathbb{N}}$. We denote by $\E^{\bs}$ the expectation with respect to $\P^{\bs}$. Player $i$'s discounted payoff for a strategy profile $\bs=(s^1, \ldots, s^N)\in\cS$ reads
\begin{equation}\label{equation:definition_discounted_payoff}
U^i(\bs) = (1-\beta)\sum_{t=0}^{\infty}\beta^t \,\E^{\bs}\parentheseDeux{u^i (A^1 _t,\ldots, A^N _t)}\eqsp,
\end{equation}
where $0<\beta<1$ is a discount common factor. Finally, for any time $t\geq 0$ and history $\hist_t\in\cH$, we also define the \textit{continuation payoff} of player $i$ for strategies $\bs$ as
\begin{equation}
    \label{def:continuation_payoff}
    U^i (\bs\,;\, \hist_t ) = (1-\beta)\,\sum_{\ell=t}^\infty \beta^{\ell-t}\,\E^{\bs}\parentheseDeux{ u^i (A^1_{\ell},\ldots,A^N_{\ell})\,|\,\hist_t}\eqsp.
\end{equation}
We now introduce the two public equilibrium concepts used in this paper. 
\begin{definition}[$(\varepsilon, \mathfrak{S})-$NE]\label{definition:ne} For any $\mathfrak{S}=\mathfrak{S}^1\times\ldots\times\mathfrak{S}^N\subseteq\cS$ and $\varepsilon\geq0$, a strategy profile $\bs_\star =(s^1_\star,\ldots, s^N_\star)\in\mathcal{S}$ is a $(\varepsilon, \mathfrak{S})$-\emph{Nash equilibrium} if
    $$
U^i (s^i _\star, s^{-i}_\star)\geq U^i (s^i, s^{-i}_\star)-\varepsilon\qquad\text{for any}\quad i\in[N]\quad\text{and}\quad s^i\in\mathfrak{S}^i\eqsp.
$$
\end{definition}
\Cref{definition:ne} corresponds to the classic $\varepsilon$-Nash equilibrium concept, where the set of deviations available to players can be restricted to a subset $\mathfrak{S}\subset\cS$ of strategies (naturally, setting $\mathfrak{S}=\cS$ recovers the $\varepsilon$-Nash equilibrium). In what follows, we refer to a $(\varepsilon, \mathfrak{S})$-Nash equilibrium as  $(\varepsilon, \mathfrak{S})-$NE, and a $(\varepsilon, \cS)$-Nash equilibrium as a $\varepsilon$-NE. At a high level, this concept is motivated by the fact that incorporating all possible deviations into equilibrium analysis with statistical tests is information-theoretically infeasible--see \Cref{section:folk_anytime} for an extensive discussion.
We also introduce the following subgame perfect equilibrium concept.
\begin{definition}[$(\varepsilon, \delta)$-HP-SPNE]\label{definition:spne}
    For any $\delta\in(0,1)$ and $\varepsilon\geq0$, we call a strategy profile $\bs_{\star}=(s^1 _\star,\ldots, s^N _\star)\in\cS$ a $\delta$-\textit{high probability $\varepsilon$-subgame perfect Nash Equilibrium} if, for any $t\geq 0$,
$$\text{there exists}\quad \mathsf{\mathsf{H_t}}\subseteq\cH_t\quad\text{such that}\quad \P^{\bs_{\star}} (\hist_t\in\mathsf{H}_t)\geq 1-\delta\eqsp,$$
and for any history $\hist_t \in \mathsf{H}_t$,
\begin{align*}
U^i (s^i _\star ,s^{-i} _\star\,;\, \hist_t) \geq U^i (s^i , s^{-i}_\star\,;\,\hist_t)-\varepsilon\quad\text{for any }i\in[N]\quad\text{and}\quad s^i \in\cS^i\eqsp.
\end{align*}
\end{definition}

To the best of our knowledge, \Cref{definition:spne} has not been formalized in the repeated-games literature. It can be viewed as a probabilistic relaxation of subgame (public) perfection, in which approximate sequential rationality is required only on histories that arise with high probability under the equilibrium strategy. Note that if $\delta=0$, we recover the classic $\varepsilon$-Subgame Perfect Nash equilibrium concept. Here again, we must exclude degenerate subgames from the equilibrium analysis due to the statistical nature of our approach. In the rest of the text, this equilibrium concept is referred to as a $(\varepsilon, \delta)$-HP-SPNE. We  refer to a $(0,0)$-HP-SPNE as an SPNE.

\subsection{The Folk Theorem under perfect monitoring}
We now recall the classic Folk Theorem with mixed strategies, which we ultimately aim to extend to the case of imperfect monitoring. To ease the exposition, and in this section only, we assume that players enjoy perfect monitoring. Formally, in this section, the repeated game is assumed to unfold as follows. Starting from a history $\hist_0 = \emptyset$, player $i$ picks a strategy $s^i\in\cS^i$ with $s^i (\hist_0)=w^i_0$, and for any $t\geq 0$
\begin{enumerate}
    \item They play a mixed strategy $w^i_t = s^i (\hist_t)$,
    \item They receive a payoff $u^i (w^i_t, w^{-i}_t)$ and observe $w^{-i}_t \in\Delta(\cA^{-i})$. The public  history is updated for the next round as $\hist_{t+1}=\hist_t \cup\sigma(\{w^i_t , w^{-i}_t\})$.
\end{enumerate}
Note that for any $t\geq0$, the history $\hist_t = \{(w^i_0, w^{-i}_0),\ldots, (w^{i} _{t-1}, w^{-i}_{t-1})\}$ contains the \textit{actual} mixed strategies played so far, not realizations from it. This situation corresponds to perfect monitoring.

The \textit{Folk Theorem} states the following fact: in contrast to one-shot games, for which a given payoff profile can only be reached by rational agents if it corresponds to a Nash equilibrium, in a repeated game, \textit{any} feasible payoff profile $\bv=(v^1 ,\ldots, v^N)\in[0,1]^N$ above some threshold can be sustained as a (subgame perfect) Nash equilibrium (SPNE).

A standard approach to constructing an SPNE that achieves $\bv$ is to rely on a \textit{grim-trigger} strategy \citep{fudenberg1986folk}. Broadly speaking, players agree to stick to a strategy profile  $\bw_{\bv}=(w_{\bv} ^1 ,\ldots, w_{\bv} ^N)\in\Delta(\cA)^N$ achieving $\bv$, that is, $(u^1 (\bw_{\bv}),\ldots, u ^N (\bw_{\bv}))=\bv$. If one of them  deviates, their opponents switch to a punishment (typically a Nash equilibrium) action forever. If players care sufficiently about future, that is $\beta$ is close enough to one, the threat of punishment deters players from deviating from $\bw_{\bv}$. This ensures that $\bw_{\bv}$ is an SPNE inducing payoffs $\bv$.

We formally introduce the grim-trigger strategy and the Folk Theorem below. Fix some reference mixed strategy Nash equilibrium $\bb=(b^1 ,\ldots b^{N})\in\Delta(\cA)$ (which exists by the Nash theorem) and write $\underline{u}^i = u^i (\bb)$ the associated payoff for player $i$. For any history $\hist_t \in\cH$, the grim-trigger strategy profile $\bs_{\bv}=(s^1_{\bv},\ldots, s^N_{\bv})\in\cS$ is given, for any $i\in[N]$, by
\begin{equation}
    \label{def:grim}
    s_{\bv}^i (\hist_t) = \begin{cases}
        w_{\bv} ^i &\text{if }\hist_t = (\bw_{\bv},\ldots,\bw_{\bv})\eqsp,\\
        b^i &\text{otherwise}\eqsp.
    \end{cases}
\end{equation}In what follows, we write $\bar{u}_i =\max_{\bw}u^i (\bw)$ for player $i$'s maximum payoff, and $\cU^i  =\mathsf{conv}( u^i (\cA))$ the set of feasible instantaneous payoffs under mixed strategies for player $i$.
\begin{restatable}{theorem}{theoremfolk}\label{theorem:folk} Assume perfect monitoring. Fix a payoff profile $\bv=(v^1,\ldots, v^N)\in\cU ^1\times \ldots\times\cU ^N$ such that $v^i \geq \underline{u}^i$ for any $i\in[N]$. If $\beta \geq (\bar{u}^i - v^i)(\bar{u}^i - \underline{u}^i)^{-1}$ for any player $i$, then $\bs_{\bv}$ defined as in \eqref{def:grim} satisfies:
\begin{enumerate*}[label=(\roman*)]
    \item $
U^i (\bs_{\bv})=v^i$ for any $i\in[N]$;
    \item $\bs_{\bv}$ is an SPNE.
\end{enumerate*}
\end{restatable}

The proof of \Cref{theorem:folk}, which is well known \citep{fudenberg1986folk}, is given for completeness in  \Cref{section:proofperfect}. This theorem states that provided players are sufficiently patient (as captured by the lower bound on the discount factor $\beta$), the grim-trigger strategy \eqref{def:grim} allows to reach the agreed utility profile $\bv$ as an SPNE.

Note that in \Cref{theorem:folk}, punishment corresponds to a Nash equilibrium, as is standard in the literature \citep{fudenberg1986folk}. One could strengthen the result by considering punishments that deliver minmax payoffs, but this requires a more involved analysis—in particular, punishers must themselves be incentivized to continue punishing \citep{fudenberg1986folk}. For clarity and ease of exposition, we therefore restrict attention to Nash-equilibrium punishments in this paper.
  \section{Equilibrium Under Imperfect Monitoring via Anytime Testing}\label{section:folk_anytime}
The perfect monitoring setting presented in the previous section is unfortunately rarely encountered in real-life situations.
We thus instead ask whether similar results can be obtained in a setting where players do not observe actual mixed strategies, but only \emph{pure actions drawn from them}. \label{subsection:testthenpunish} In this imperfect information setting, applying the grim-trigger \eqref{def:grim} is impossible, since players do not directly observe the mixed strategies played by their opponents. In particular, it is no longer possible to assess with certainty whether a player $j$ deviated from $w^{j}_{\bv}$ (except when a pure action from outside the support of $w^{j}_{\bv}$ is drawn). However, we will show that it is possible to use \textit{statistical tests} to check whether such a deviation occurred with high probability. If a player gathers sufficient statistical evidence to support this hypothesis, they could trigger a punishment similarly to \eqref{def:grim}. This is the idea behind the \textit{test-then-punish} strategy that we present formally below.

Assume that players agree on  feasible payoffs $\bv=(v^1,\ldots, v^N)\in\cU^1\times\ldots\times\cU^N$ and profile of mixed strategies $\bw_{\bv}=(w^1_{\bv} ,\ldots,w^N_{\bv})\in\Delta(\cA)$ such that $v^i = u^i (\bw_{\bv})$ for any $i\in[N]$. We then define $\boldsymbol{\sigma}_{\bv}=(\sigma_{\bv}^1,\ldots, \sigma_{\bv}^N)\in\cS$ where for any $j\in[N]$, $\sigma^{j}_{\bv}(\hist_t)=w^{j}_{\bv}$ for all $t\geq0\text{ and }\hist_t\in\cH_t$. 
Then, each player $j$ can test as follows for $i\ne j$
\begin{equation}
    \label{eq:hoanytime}
\rmH^i _{0}:\;\eqsp s^{i}=\sigma^{i}_{\bv} \qquad\text{versus}\qquad \rmH^i _{1} \colon \;\,s^{i}\ne \sigma^i_{\bv}\eqsp.
\end{equation}
$\rmH^i _{0}$ being true means that player $i$ sticks to the cooperative strategy $w^{i}_{\bv}$. On the contrary, rejecting $\rmH^i _{0}$ amounts to concluding that the player $i$ does not play the agreed-upon strategy.

To decide whether $\rmH^i _{0}$ is true, player $j\ne i$ may rely on a sequence of tests $\bpsi^i = \{\psi^i _t\}_{t\geq0}$ where $\psi^i_0=0$ and for any  $t\geq 1, \psi^i_t \colon \cH_t\rightarrow \{0,1\}$. If $\psi^i _t=1$, then $\rmH^i_0$ has been rejected at time $t$. Here, we assume that all players $j\ne i$ use the same test for $\rmH^i_0$ (hence the absence of $j$ indexing in $\psi^i_t$) given that they all have access to a common public history $\hist_t \in\cH_t$ at time $t$. We can then define the \textit{test-then-punish} strategy $\bs_{\bv}\in\cS$ where for $j\in[N]$, $s^j _{\bv}$ is given by

\begin{equation}\label{def:testthenpunish}
s^j_{\bv}(\hist_t)=
\begin{cases}
w^j_{\bv}& \text{if } \prod\limits_{s=0}^t \prod\limits_{i=1}^{N}\,(1-\psi^i _s (\hist_{s}))=1\eqsp ,\\
b^{\,j}& \text{otherwise} \eqsp.
\end{cases}
\end{equation}
Informally, at time $t \geq 0$, if players have not detected any deviation, i.e., $\psi_t^i = 0$ for all $i \in [N]$, they continue playing the cooperative mixed profile $w_{\boldsymbol v}$. Once a deviation is detected, that is, $\psi_t^i = 1$ for some $i$, all players permanently switch to the punishing strategy in the following rounds.

To simplify notation, for any $i\in[N]$ we define the (random) rejection time associated with the sequence of tests $\bpsi^i$ as
\begin{equation}\label{def:rejectiontime}
\tau_{\boldsymbol{\psi}^i} =\inf\bigl\{t\geq 0 \colon \psi^i _t=1\bigr\}\in\N \cup\{\infty\} \eqsp,
\end{equation}
with the convention $\inf\varnothing =\infty$. In words, $\tau_{\boldsymbol{\psi}^i}$ is the first time $\rmH^i _0$ is rejected; that is, the first time player $i$ has been detected as deviating. Similarly, we define $\tau_{\bpsi}$ as
$$
\tau_{\bpsi}=\min_{i \in [N]} \tau_{\bpsi^i}\eqsp,
$$
which is the first time \textit{any} null hypothesis  $\rmH^1_0,\ldots, \rmH^N_0$ has been rejected, that is any player has been detected as deviating. It should be clear from \eqref{def:testthenpunish} that players play $\bw_{\bv}$ for any $t\leq \tau_{\bpsi}$, and turn to the punishment $\mathbf{b}$ for any $t>\tau_{\bpsi}$. 





Observe that the test-then-punish strategy \eqref{def:testthenpunish} mimics the grim trigger strategy under perfect monitoring \eqref{def:grim} with one notable difference: while the latter triggers a punishment as soon as a deviation is observed (as allowed by perfect monitoring), the former triggers a punishment as soon as one null hypothesis $\rmH^i_0$ is rejected (because of imperfect monitoring). If tests were perfect---that is, if $\rmH^i_0$ would be rejected if and only if player $i$ actually deviated, the two approaches would be equivalent. However, statistical tests are never perfect. They may suffer from Type I error (rejecting $\rmH_0$ while it is true) and Type II error (failing to reject $\rmH_0$ while it is false), because they rely on stochastic observations. In our case, this means that the test-then-punish strategy \eqref{def:testthenpunish} may either wrongfully trigger punishment while players are actually sticking to $\bw_{\bv}$, which corresponds to Type I error, or fail to punish a player who is deviating, which corresponds to Type II error. Both kinds of errors may prevent the test-then-punish procedure from achieving its goal of ensuring a target payoff profile $\bv$ under a (subgame perfect) Nash equilibrium. 

It is therefore natural to ask that the family of sequential tests $\bpsi=(\bpsi^1,\ldots, \bpsi^N)$ satisfies some conditions ensuring that Type I and Type II errors remain low, so that $\bs_{\bv}$ is effective.  In particular, regarding the Type I error,  we introduce the following requirement.
\begin{condition}[Low Type I error probability]    \label{def:lowlevel}There exists $\gamma \in (0,1)$ such that for any $i\in[N]$
\begin{equation}
\E^{\bs_{\bv}}[\2{\tau_{\bpsi^i}<\infty}]
= \P^{\bs_{\bv}}\parenthese{\tau_{\bpsi^i}<\infty}
\leq \gamma/N\eqsp.
\end{equation}
\end{condition}If \Cref{def:lowlevel} is true, then for any $i\in[N]$ the test sequence $\bpsi^i$ is such that the probability to wrongfully reject $\rmH^i_0$ under the cooperative strategy $\bw_{\bv}$ remains low. Equivalently, the probability to wrongfully detect a deviation from player  $i$ while they cooperate is low. Note that \Cref{def:lowlevel} immediately controls the probability of wrongfully entering the punishment phase, since a union bound gives
\begin{align*}
\P^{\bs_{\bv}}(\tau_{\bpsi}<\infty)=\P^{\bs_{\bv}}(\cup_{i\in[N]}\,\{\tau_{\bpsi^i}<\infty\})\leq\sum_{i\in[N]}\P^{\bs_{\bv}}(\tau_{\bpsi^i}<\infty)\leq\gamma\eqsp.
\end{align*}
This ensures that cooperation can be sustained in the long run with high probability. Note that \Cref{def:lowlevel} becomes more stringent with the number of players. Indeed, more players imply more hypothesis tests, which increases the probability of false alarms and therefore calls for tighter control of each individual test. Controlling the false discovery rate \citep{benjamini1995controlling}, particularly via anytime extensions \citep{wang2022false}, may provide a less conservative way to handle the growing number of hypotheses when a player faces multiple opponents, while still limiting false alarms.

Regarding Type II error, the family of sequential tests $\bpsi=(\bpsi^1,\ldots, \bpsi^N)$ should be such that if a player deviates, they ought to be detected in finite time with high probability. Note that it is extremely hard, if not hopeless, to detect deviations that are arbitrarily close to the collaborative strategy $\bw_{\bv}$. We thus focus on deviations that are ``far enough'' from $\bw_{\bv}$. To formalize this idea, we define for any $\varepsilon>0$ and $i\in[N]$ the total variation-ball of radius $\varepsilon$ centered in $w^i_{\bv}$
\begin{align*}
\mathsf{B}^i(\bw_{\bv}\, , \,\varepsilon) = \{s\in\cS^i \colon \: \lVert s(\hist_t)\,-\,w^{i}_{\bv}\rVert_1\leq 2 \varepsilon\text{ for any }t\geq0\text{ and }\hist_t\in\cH_t\}\eqsp.
\end{align*}
Notice that for any $\mathfrak{S}^i\subseteq\cS^i,\,\mathfrak{S}^i\setminus\mathsf{B}^i(\bw_{\bv},\varepsilon)$ contains all deviations in $\mathfrak{S}^i$ for player $i$ that are $\varepsilon$-significant, in the sense that they are at least $\varepsilon$ far from the cooperative strategy $w^{i}_{\bv}$ in total variation. This situation is depicted on \Cref{figure:deviation}.  Then, for any $\mathfrak{S}=\mathfrak{S}^1\times\ldots\times\mathfrak{S}^N \subseteq\cS$, we define the set of strategy profiles containing at least one $\varepsilon$-significant deviation
\begin{equation}\label{eq:spacedeviation}
\mathfrak{S}(\bw_{\bv},\varepsilon) = \defEns{\bs=(s^1,\ldots,s^N)\in\mathfrak{S},\:\text{there exists }i\in [N] \colon \, s^i \in\mathfrak{S}^i\setminus\mathsf{B}^i(\bw_{\bv}\,,\,\varepsilon)}\eqsp.
\end{equation}


\begin{figure}[ht]
\centering
\resizebox{0.35\textwidth}{!}{%
\begin{tikzpicture}
    \draw[thick, fill=outerS]
        plot [smooth cycle, tension=0.8]
        coordinates {
            (-4.2,0.6) (-3.5,2.2) (-1.5,2.8) (1.2,2.3)
            (3.4,1.4) (3.9,-0.4) (2.6,-2.4)
            (0,-2.9) (-3,-2.2)
        };
    \node at (3.0,2.2) {$\cS^i$};
    \draw[thick, fill=innerS]
        plot [smooth cycle, tension=0.8]
        coordinates {
            (-3,0.5) (-2.4,1.9) (-0.8,2.3) (1.2,1.9)
            (2.3,0.9) (2.1,-1.1)
            (0.8,-2.0) (-1.3,-2.1) (-2.6,-1.2)
        };
    \node at (2,1.8) {$\mathfrak{S}^i$};
    \filldraw[fill=white, draw=black, dashed]
        (0,0) circle (0.8);
    \fill (0,0) circle (2pt);
    \node[below right] at (0,0) {$w^i_{\bv}$};
    \draw[->] (0,0) -- (0.8,0) node[midway, above] {$\varepsilon$};
\end{tikzpicture}%
}
\caption{If one player $i\in[N]$ picks a strategy $s^i$ in the green subset, then $\bs=(s^i, s^{-i})\in\mathfrak{S}(\bw_{\bv},\varepsilon)$.}\label{figure:deviation}
\end{figure}
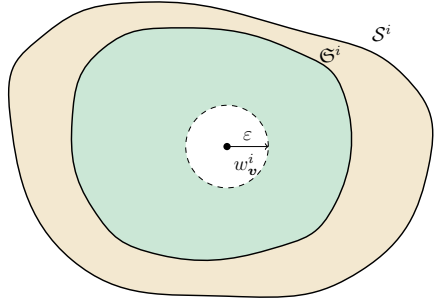

We are now ready to state our second requirement. Intuitively, the following condition states that any strategy profile  $\bs\in\mathfrak{S}$ containing a  $\varepsilon$-significant deviation triggers a punishment in finite time. 
\begin{condition}\label{def:highpower}
 For $\varepsilon>0$, there exists a sequence $\parentheseLigne{\zeta^{(\varepsilon)}_t}_{t\geq0}$ such that
 \begin{equation}
\text{for any }t\geq0\,,\quad \sup_{\bs\in\mathfrak{S}(\bw_{\bv}\,,\,\varepsilon)}\P^{\bs}(\tau_{\bpsi}\geq t)\leq \zeta^{(\varepsilon)}_t\quad\text{and}\quad \sum_{t=0}^\infty\,\zeta^{(\varepsilon)}_t < \infty\eqsp.
\end{equation}
\end{condition} \Cref{def:highpower} ensures that for any strategy profile  $\bs\in\mathfrak{S}(\bw_{\bv},\varepsilon)$ (such that at least one player deviates from $w^i_{\bv}$ by at least $\varepsilon$), the punishment time $\tau_{\bpsi}$ does not exceed $t$ with high probability. In particular, \Cref{def:highpower} implies that players enter the punishment phase in finite time, in expectation, when one of them $\varepsilon$-significantly deviates since
$$
\sup_{\bs\in\mathfrak{S}(\bw_{\bv},\varepsilon)}\E^{\bs}[\tau_{\bpsi}] = \sup_{\bs\in\mathfrak{S}(\bw_{\bv},\varepsilon)}\sum_{t\geq0}\P^{\bs}(\tau_{\bpsi}\geq t)\leq \sum_{t\geq 0}\,\sup_{\bs\in\mathfrak{S}(\bw_{\bv},\varepsilon)}\P^{\bs}(\tau_{\bpsi}\geq t)\leq \sum_{t=0}^\infty \zeta^{(\varepsilon)}_t < \infty\eqsp.
$$

We remark that  condition \eqref{def:highpower} leaves  ``small'' deviations in $\mathsf{B}^{i}(\bw_{\bv},\varepsilon)$ aside: such deviations are not required to trigger a punishment in finite time. In fact, this is not a matter of concern since such deviations cannot bring a significant utility gain, and are therefore covered by our approximate equilibrium notion. This is established formally in the following proposition. 
\begin{restatable}{proposition}{lemmasmalldeviation}\label{lemma:smalldeviation}Fix $\varepsilon>0$. For any $i\in[N]$, if $s^i \in\mathsf{B}^i(\bw_{\bv}\,,\,\varepsilon)$, then
$\;
U^i (s^i, s^{-i}_{\bv})\leq v^i + \varepsilon\eqsp.
$
\end{restatable}To summarize, \Cref{def:highpower} together with \Cref{lemma:smalldeviation} ensures that a player $i$ deviating to some $s^i\ne s^i_{\bv}$ either triggers a punishment in finite time (if $s^i \in\mathfrak{S}^i\setminus\mathsf{B}^i(\bw_{\bv},\varepsilon)$, by \Cref{def:highpower}) or cannot enjoy a large utility gain (if $s^i\in\mathfrak{S}^i\cap\mathsf{B}^i(\bw_{\bv},\varepsilon)$, by \Cref{lemma:smalldeviation}). This observation allows us to establish  a Folk Theorem based on the test-then-punish strategy \eqref{def:testthenpunish}, provided  \Cref{def:lowlevel} and \Cref{def:highpower} hold.
\begin{restatable}{theorem}{folkanytime}\label{theorem:folk_anytime_main}
   Let $\mathfrak{S}=\mathfrak{S}^1\times\ldots\times\mathfrak{S}^N \subseteq\cS$ and   consider a family of sequential tests $\boldsymbol{\psi}$ satisfying  \Cref{def:lowlevel} and \Cref{def:highpower}. Let $\bv\in\cU^1\times\ldots\times\cU^N$ be a feasible payoff such that for any $i\in[N]$
\begin{align*}
   v^i \geq\underline{u}^i = u^i (\bb) \eqsp,
\end{align*}
where $\bb=(b^1 ,\ldots b^{N})\in\Delta(\cA)$ is some reference mixed strategy Nash equilibrium, which exists by the Nash theorem. Denoting $\tau^{(\varepsilon)}=\sup\limits_{\bs\in\mathfrak{S}(\bw_{\bv}\, , \,\varepsilon)}\E^{\bs}[\tau_{\bpsi}]<\infty$, and $\bar{u}^i = \max_{a \in \cA} u^i(a)$, if for any $i\in[N]$
\begin{equation}\label{eq:patienceconditionanytime}     \bar{\beta}^i=\frac{\bar{u}^i - (1-\gamma)v^i}{\bar{u}^i - \underline{u}^i}\in[0,1]\quad\text{and}\quad
    \beta^{\tau^{(\varepsilon)}} \geq \bar{\beta}^i\eqsp,
\end{equation}
    then $\bs_{\bv}$ defined in \eqref{def:testthenpunish}  satisfies
\begin{enumerate}[label=(\roman*)]
\item  
    for any $i\in[N]$, $(1-\gamma)v^i \leq U^i (\bs_{\bv})\leq v^i\eqsp,$
\item $\bs_{\bv}$ is a $(\varepsilon+\gamma,\mathfrak{S})$-NE.
\end{enumerate}
\end{restatable}
It is informative to compare \Cref{theorem:folk_anytime_main} with its perfect monitoring counterpart, \Cref{theorem:folk}. In the latter, $U^i (\bs_{\bv})$ is \textit{exactly} equal to the target payoff $v^i$. In the former, it lies in an interval of length $\gamma$ around $v^i$. This difference is related to Type I error: even though all players stick to the collaborative strategy $w _{\bv}$, they bear the risk of wrongfully entering the punishment phase with probability at most $\gamma$. Note that a better test results in a lower error $\gamma$, and thereby decreases the utility gap.

Also, observe that in \Cref{theorem:folk}, $\bs_{\bv}$ is an exact SPNE. On the contrary, in \Cref{theorem:folk_anytime_main} it is a $(\varepsilon+\gamma,\mathfrak{S})$-NE. The $\varepsilon+\gamma$ approximation term stems from type I error, and the possibility for players to pick small deviations that remain undetected as discussed in \Cref{lemma:smalldeviation}. Moreover, \Cref{theorem:folk_anytime_main} only shows that $\bs_{\bv}$ is a Nash equilibrium and not that it is  subgame perfect.
As we shall see in \Cref{section:folk_batch}, it is possible to obtain subgame perfection by adopting a different testing strategy, at the expense of Type I and II error control. Finally, observe that under perfect monitoring (\Cref{theorem:folk}) the patience condition depends directly on the discount factor $\beta$, since deviations are detected immediately. In contrast, under the test-then-punish strategy of \Cref{theorem:folk_anytime_main}, a deviation may go undetected for up to $\tau^{(\varepsilon)}$ periods, so the relevant condition involves $\beta^{\tau^{(\varepsilon)}}$. Consequently, the longer it takes to detect deviations, the more patience is required. More powerful tests, which reduce $\tau^{(\varepsilon)}$, relax this requirement.

\subsection{Anytime testing via plug-in e-processes}\label{subsection:stationary}

To make \Cref{theorem:folk_anytime_main} fully actionable, we still need to specify a way for players to design  a family of sequential tests $\bpsi$ satisfying \Cref{def:lowlevel} and \Cref{def:highpower}. 
This can be achieved by making use of \textit{anytime testing} and e-processes \citep{ramdas2024hypothesis}. We provide a brief review of this framework in \Cref{appendix:statistical_background}, while we only provide here the main ideas and concepts of this framework.

Given a null hypothesis $\mathrm{H}_0$, anytime valid testing is built on the notion of an \textit{e-process} that are typically a supermartingale $(M_t)_{t\geq 0}$ for any probability distribution belonging to $\mathrm{H}_0$. Under this setting, Ville's inequality \citep[][see \Cref{appendix:statistical_background}]{doob1939jean},  $(M_t)_{t\geq 0}$ satisfies for any $\P_0 \in \mathrm{H}_0$, \begin{equation}\label{equation:ville}
\P_0 \,\left(\sup_{t\geq 0} M_t \geq 1/\gamma\right)\;\leq\;\gamma
\qquad \mbox{for any }\ \gamma \in(0,1) \eqsp.
\end{equation}
In other words, under the null hypothesis, the event $\cup_{t\geq0} \{M_t \geq 1/\gamma\}$ cannot happen with probability greater than $\gamma$. Thus, a natural test statistic for testing $\rmH_0$ at level $\gamma$ is given by $ \2{M_t \geq 1/\gamma}$.

Considering $(\scrH_t)_{t\geq 0}$, the natural filtration associated to $(\hist_t)_{t\geq 0}$, $\scrH_t = \sigma(\hist_s \, :\, s \leq t)$, we now specify how a player $j \in [N]$ can design a $(\hist_t)_{t \geq 0}$-supermartingale $\{E^i _t\}_{t\geq 0}$ to test $\rmH^i_0$ for an opponent $i \in [N]$. Since at any round $t$, the pure actions $A^1_t,\ldots,A^N_t$ are observed, players can count how many times each player $i\in[N]$ played action $a\in\cA^i$ up to step $t-1$ for  which we denote by $N_{t}^{i}(a)=\sum_{s=0}^{t-1} \1\{A_s^{i}=a\}$. It is then possible to define the  (smoothed\footnote{$1$ and $K$ play the role of smoothing as in a Laplace estimator.}) empirical frequencies of actions played by player $i$ as
\begin{equation}\label{eq:plugin-predictor}
\hat{w}^{i}_{t}(a)=\frac{N_{t}^{i}(a)+1}{t+K} \qquad\text{for any } a\in \cA\eqsp,
\end{equation}
and $\widehat w^{i}_{-1}(a) = 1/K$. Then, we can form a plug-in e-process \cite{ramdas2024hypothesis} to test $\rmH^i_0$ for player $i$ as follows. Set $E_{-1}^i=1$ and for any $t \geq 0$
\begin{equation}\label{eq:e-process}
E_t^i = \prod_{s=0}^t \frac{\widehat w^{i}_{s}\,(A_s^{i})}{w_{\bv}^{i}\,( A_s^i )} \eqsp.
\end{equation}
We formally establish that $(E^i_t)_{t\geq 0}$ is a $(\hist_t)_{t \geq 0}$-supermartingale and therefore an e-process under $\rmH^i_0$ in \Cref{appendix:statistical_background}, \Cref{theorem:anytime}. Intuitively, if player $i$ draws a pure action from $w^i_{\bv}$, then $\hat{w}^i_s (a) / w^i_{\bv}(a) \approx 1$ for large $s$ and any $a\in\cA^i$, so $E^i_t$ does not grow too much on average over time. On the contrary, if player $i$ deviates, then $E^i_t$ ought to increase and triggers the test the first time it exceeds $1/\gamma$.
If we combine all the above elements, we obtain a practical testing procedure for players to test $\rmH^i_0$: they can keep track of $i$'s action empirical frequencies \eqref{eq:plugin-predictor} to build the  e-process \eqref{eq:e-process} and check whether it exceeds a given threshold. For $\gamma\in(0,1)$, this leads to the sequential test $\bpsi^i = \{\psi^i_t\}_{t\geq0}$ for $\rmH^i_0$ defined for any $t\geq 0$ as
\begin{equation}
\label{definition:test_anytime}
\tpsi^i_t =\1\left\{ E_t^i \geq N/\gamma \right\} \eqsp.
\end{equation}
By Ville's inequality \eqref{equation:ville}, we can explicitly control the significance level of this test by tuning the threshold. This immediately implies that such a test satisfies \Cref{def:lowlevel}, as formally recorded in the following proposition.

\begin{restatable}{proposition}{anytimetestlowlevel}\label{proposition:anytime_low_level} Let $\gamma\in(0,1)$ and consider for any $i\in[N]$, the sequence of tests $(\tpsi_t^i)_{t \geq 1}$ defined in \eqref{definition:test_anytime}
and the resulting strategy $\bs_{\bv}$ defined in \eqref{def:testthenpunish}. Then,  \Cref{def:lowlevel} is satisfied.
\end{restatable}
We now show that the test sequence defined above also satisfies \Cref{def:highpower}, under a simplifying assumption. Before proceeding, we highlight a fundamental difficulty. For $\mathfrak{S}=\mathfrak{S}^1\times\ldots\times\mathfrak{S}^N \subseteq\cS$, the Type II error control required in \Cref{def:highpower} must hold uniformly over the set $\mathfrak{S}(\bw_{\bv},\varepsilon)$ defined in \eqref{eq:spacedeviation}. That is, any strategy $s^i \in\mathfrak{S}^i$ such that $s^i (\hist_t)$ deviates by at least $\varepsilon$ from $w^i_{\bv}$ at some time $t\geq 0$ should, with high probability, eventually trigger a punishment. From a statistical perspective, this amounts to testing a simple null hypothesis $\rmH_0^i$ against an arbitrarily large composite alternative $\rmH_1^i$. It is well known that, without additional structural restrictions on the alternative, uniformly powerful tests generally do not exist, and nontrivial control of the Type II error is in many cases information-theoretically impossible \citep{lehmann2005testing}. We therefore simplify the problem by restricting attention to a set of simple deviations as follows
\begin{equation}
\label{def:stationarydeviations}
\tilde{\mathfrak{S}}^i
=
\left\{\,
s \in \cS^i \;:\; \text{there exists } \tilde{w}^i \in \Delta(\cA^i) \text{ such that  for any } t \geq 0,\ \hist_t \in \cH_t, \,  s(\hist_t)=
\tilde{w}^i \, 
\right\}\eqsp.
\end{equation}

In plain words, $\tilde{\mathfrak{S}}=\tilde{\mathfrak{S}}^1\times\ldots\times\tilde{\mathfrak{S}}^N$ contains \textit{stationary} strategies, which stick to a single mixed strategy $\tilde{w}^i \in\Delta(\cA)$. This simplifies the statistical analysis by ensuring that under $\rmH^i_{1}$, any player $i\in[N]$ draws i.i.d.\ actions from  $\tilde{w}^{i}$ before getting caught. However, we stress that the problem remains difficult, since players do not know what mixed strategy $\tilde{w}^i$ has been chosen by player $i$. In other words, $\rmH^i_1$ is still a composite alternative. We show in the following proposition that \Cref{def:highpower} is satisfied when strategies lie in  $\tilde{\mathfrak{S}}(\bw_{\bv} , \varepsilon)$, as defined in \Cref{eq:spacedeviation}.

We now have statistical leverage to obtain a bound on the expected stopping time $\ttau^i$ defined as 
\begin{align}\label{equation:stopping_time_i}
\ttau^i_{\psi} = \inf\{t \geq 1 \colon \tpsi^i_t = 1\} \eqsp,
\end{align}
where $(\tpsi^i_t)_{t\geq 0}$ is defined in \eqref{definition:test_anytime} for a error level $\gamma \in \ooint{0,1}$.

\begin{restatable}{theorem}{probaboundanaytime}\label{theorem:probaboundanaytime}
Consider $\epsilon >0$ and  a level $\gamma \in \ooint{0,1}$. Consider the sequential tests  $\{\tpsi^i\}_{i \in [N]}$ defined in \eqref{definition:test_anytime}. Then, for $\tilde{\mathfrak{S}}=\tilde{\mathfrak{S}}^1\times\ldots\times\tilde{\mathfrak{S}}^N$, where 
$\tilde{\mathfrak{S}}^i$ are defined in \eqref{def:stationarydeviations}, and $\ttau_{\bpsi} = \min_{i\in[N]} \ttau_{\psi}^i$, Condition $2$ holds with  
\begin{align*}
    \zeta^{(\epsilon)}_t = C_1\Biggl\{ t \exp\left(-\frac{2 t \epsilon^2}{C_2 \log(t)^2}\right) + \exp\left(-t\frac{4 \epsilon^4
    }{C_2 \log(t)} \right) + \exp\left(-\frac{t \epsilon^4}{2 |\log \underline{\bw_{\bv}}|}\right) \Biggr\}+\1\{t\leq 10\log(1/\gamma)/\epsilon^2\} \eqsp,
\end{align*}
for any $t \geq 2, \zeta^{(\epsilon)}_0=1, \zeta^{(\epsilon)}_1=1$, where  $C_1, C_2 >0$ are universal constants and we set $\underline{\bw_{\bv}} = \min_{i\in [N]} \min_{a \in \cA^i} w^i_{\bv}[a]$. In particular, we have $\sum_{t \geq 0} \zeta^{(\epsilon)}_t <\infty$.
\end{restatable}

\begin{restatable}{corollary}{corogamestoppingtime}\label{corollary:corogamestoppingtime}
   Under the same setting than \Cref{theorem:probaboundanaytime}, for any  $\bs \in \tilde{\mathfrak{S}}(\bw_{\bv}, \epsilon)$, there exists a universal constant $C>0$ such that
    \begin{align*}
         \E^{\bs}[\ttau_{\bpsi}] \leq \ttau_{\varepsilon}  = \frac{10\log(1/\gamma)}{\epsilon^2} + C \frac{1+|\log \underline{\bw_{\bv}}|}{\epsilon^5} \eqsp.
    \end{align*}
\end{restatable}



As $\gamma$ decreases (i.e., as the test becomes more conservative), the upper bound on $ \E^{\bs}[\ttau_{\bpsi}]$ increases, illustrating the classical trade-off between Type I error control and detection power. By \Cref{proposition:anytime_low_level} and \Cref{theorem:probaboundanaytime}, sequential tests as defined in \eqref{definition:test_anytime} satisfy \Cref{def:lowlevel} and \Cref{def:highpower}. It is therefore possible to instantiate \Cref{theorem:folk_anytime_main} with the family of tests $\{(\tpsi^i)_{t \geq 1}\}_{i \in [N]}$ to obtain the following corollary.

\begin{restatable}{corollary}{anytimefolk}\label{corollary:folk_anytime}
Consider the same setting than \Cref{theorem:probaboundanaytime} with $\gamma = \varepsilon$, and the resulting strategy   $\bs_{\bv}$ defined in \eqref{def:testthenpunish} with tests specified in \eqref{definition:test_anytime}.  Moreover, let $\bv\in\cU^1\times\ldots\times\cU^N$ be a feasible payoff profile such that $v^i \geq \underline{u}^i$ for $i\in[N]$. If for any $i\in[N]$
\begin{equation}\label{eq:patiencecorollaryanytime}
    \bar{\beta}^i=\frac{\bar{u}^i - (1-\varepsilon)v^i}{\bar{u}^i - \underline{u}^i}\in[0,1]\quad\text{and}\quad
\beta^{\ttau_{\varepsilon} } \geq \bar{\beta}^i\eqsp,
\end{equation}
with $\ttau_{\varepsilon} $ defined in \Cref{corollary:corogamestoppingtime}. Then, it holds 
\begin{enumerate}[label=(\roman*)]
\item For any $i\in[N]$, $(1-\varepsilon)v^i \leq U^i (\bs_{\bv})\leq v^i\eqsp,$
\item $\bs_{\bv}$ is a $(2\varepsilon\,,\, \tilde{\mathfrak{S}})$-NE, where $\tilde{\mathfrak{S}}$ is defined as in \eqref{def:stationarydeviations}.
    \end{enumerate}
\end{restatable}
Observe that in \Cref{corollary:folk_anytime}, taking $\varepsilon\rightarrow 0$ allows us to nearly recover the perfect monitoring result of \Cref{theorem:folk}. In other words, when players use our test-then-punish strategy with sequential tests built as e-processes, they can cope with imperfect monitoring and sustain cooperation to obtain payoffs arbitrarily close to any feasible rational profile $\bv$. We note, however, that as $\varepsilon\rightarrow 0$, the patience condition \eqref{eq:patiencecorollaryanytime} requires that $\beta\rightarrow1$.

\section{Equilibrium Under Imperfect Monitoring via Batch Testing}\label{section:folk_batch}

In the previous section, we saw that anytime testing can be used to implement the test-then-punish strategy, and that it can certify an arbitrarily low Type I error. However, the approach has from two limitations. First, our analysis is limited to stationary deviations, \ie, belonging to $\tilde{\mathfrak{S}}$ defined in \eqref{def:stationarydeviations}, so as to obtain the guarantee on Type II error presented in  \Cref{def:highpower}. While plausible from a practical standpoint, one could argue that this assumption is restrictive since it rules out adversarial behaviors in which players employ sophisticated adaptive strategies. Second, $\bs_{\bv}$ is only a Nash equilibrium in \Cref{theorem:folk_anytime_main}, while we aspire to obtain a subgame perfect Nash equilibrium. 

We therefore ask whether it is possible to implement the test-then-punish strategy in full generality, by considering the set of \textit{all} deviations  $\mathfrak{S}=\cS$, and to ensure that it is a subgame perfect equilibrium. As we shall show, this more ambitious objective can indeed be achieved by switching to a different testing strategy. The price to pay, however, is the loss of explicit control over Type I error.

Broadly speaking, we propose to use a \textit{batch} testing procedure which can be described as follows. We first gather rounds in batches with a predetermined length; in each batch, players check whether any player $i\in[N]$ played on average the mixed action $w^{i}_{\bv}$ for this batch, based on the pure actions they observe; if this is the case, all players keep playing $\bw_{\bv}$ themselves; otherwise, they turn to the punishment $\boldsymbol{b}\in\Delta(\cA)$ for all of the subsequent batches. 

Formally, assume that players agree on  feasible payoffs $\bv\in\cU^1\times\ldots\times\cU^N$ and profile of mixed strategies $\bw_{\bv}=(w^1_{\bv} \ldots,w^N_{\bv})\in\Delta(\cA)$ such that $v^i = u^i (\bw_{\bv})$ for $i\in[N]$. Fix a batch length $L>0$, and define for any $k\geq 0$ the batch $\cB_k$ as
$$
\cB_k = \{Lk  ,\ldots, L(k+1)-1\}\eqsp,
$$
so $\N=\cup_{k\geq 0}\cB_k$. At the end of each batch $\cB_k$, players check by means of a statistical test whether player $i$'s actions are drawn i.i.d from $w^{i}_{\bv}$ on each of the rounds $t\in\cB_k$. Formally, they test for any $k\geq0$
\begin{equation}
    \label{eq:h0batch}
\rmH^i _{0,k}:\quad ( \,A^i_t\, )_{\,t\in\cB_k}\overset{\mathrm {i.i.d}}{\sim}w^i_{\bv}\qquad\text{versus}\qquad \rmH^i _{1,k}:\quad ( \,A^i_t\, )_{\,t\in\cB_k}\overset{\mathrm {i.i.d}}{\not\sim}w^i_{\bv} \eqsp.
\end{equation}
For $k \geq 0$, denote by $\batch_k = \{(A^i_t)_{t\in\cB_k}\, :\, i \in [N]\} \in \cA^L$. Then, a test for the hypothesis \eqref{eq:h0batch} is a map $\phi^i_k \colon  \cA^L \rightarrow\{0,1\}$  with $\phi^ i _k = 1$ meaning that $\rmH^i _{0,k}$ is rejected at batch~$k$.

We are now ready to introduce the batch test-then-punish strategy. For any $t\geq 0$, we write $k_t = \lfloor t / L \rfloor$ the index of the batch to which $t$, which implies in particular that $\cup_{k<k_{t}}\,\batch_k \subset \hist_t$. Then, we define the batch test-then-punish strategy as $\bbs_{\bv}=(\bars^1_{\bv},\ldots,\bars^N_{\bv})\in\cS$ as follows. For any $j\in[N]$ and $\hist_t \in\cH_t$, $\bar{s}_{\bv}(\hist_t)=w^j_{\bv}$ if $k_t =0$, and for any $k_t \geq 1$
\begin{equation}
\label{def:testthenpunsihbatch}
    \bars^j_{\bv} (\hist_t) = \begin{cases}
        w^j_{\bv} &\text{if}\quad \prod\limits_{k=1}^{k_{t}-1}\prod\limits_{i=1}^{N}\parenthese{1-\phi^i _{k}(\batch_k)}= 1\eqsp, \\
        b^j\eqsp &\text{otherwise}\eqsp.
    \end{cases}
\end{equation}
The batch test-then-punish strategy proceeds in the same way as its anytime counterpart. At batch $k\geq0$, if no deviation has been detected---i.e., $\phi^i_{l}=0$ for any $i\in[N]$ and any step $l\in\{1,\ldots,k-1\}$---players keep playing the cooperative strategy $\bw_{\bv}$. However, as soon as a deviation is detected, that is $\phi^i_k = 1$ for some $i$ and $k$, they switch to the punishment phase by playing $\boldsymbol{b}$ forever after. Just as previously, we can define rejection times associated to the family of test sequences $\bphi=(\bphi^1, \ldots,\bphi^N)$ as follows
\begin{equation}
\label{eq:stoppingtimebatch}\text{for any }i\in[N],\quad\kappa_{\bphi^i}=\inf\,\{k\geq0:\: \phi^i_k=1\}\qquad\text{and}\qquad
    \kappa_{\bphi}=\min_{i\in[N]}\kappa_{\bphi^i}\eqsp.
\end{equation}In \eqref{eq:stoppingtimebatch}, $\kappa_{\bphi}$ is the first global iterate $k$ such that one hypothesis among $\rmH^1_{0,k}\,\ldots\,\rmH^N _{0,k}$ is rejected.  In other words, it is the batch after which the punishment phase starts. The batch test–then–punish strategy is illustrated in \Cref{figure:batcharrow}. A deviation is detected at the end of batch $\kappa_{\bphi}$, after which players enter the punishment phase starting from batch $\kappa_{\bphi}+1$.
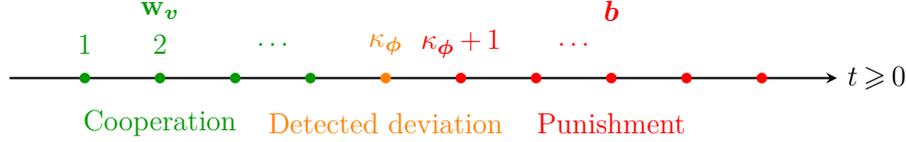
\begin{figure}[ht]
\centering
\begin{tikzpicture}[>=stealth, thick]
    \draw[->] (0,0) -- (11,0) node[right] {$t \geq 0$};

    \def\kindex{5}

    \foreach \i in {1,...,\numexpr\kindex-1\relax} {
        \fill[green!60!black] (\i,0) circle (2pt);
    }

    \fill[orange] (\kindex,0) circle (2pt);

    \foreach \i in {\numexpr\kindex+1\relax,...,10} {
        \fill[red] (\i,0) circle (2pt);
    }

    \node[green!60!black] at (1,0.45) {$1$};
    \node[green!60!black] at (2,0.45) {$2$};
    \node[green!60!black] at (3.5,0.45) {$\ldots$};

    \node[orange] at (\kindex,0.45) {$\kappa_{\bphi}$};
    \node[red] at (\kindex+1,0.45) {$\kappa_{\bphi}+1$};
    \node[red] at (\kindex+2.5,0.45) {$\ldots$};

    \node[green!60!black] at ({(\kindex-1)/2},0.9) {$\bw_{\bv}$};
    \node[red] at ({(\kindex+11)/2},0.9) {$\boldsymbol{b}$};

    \node[green!60!black] at ({(\kindex-1)/2},-0.6) {Cooperation};
    \node[orange] at (\kindex,-0.6) {Detected deviation};
    \node[red] at ({(\kindex+11)/2},-0.6) {Punishment};

\end{tikzpicture}
\caption{Timeline of the batch test--then--punish strategy. Dots represent ends of batches.}
\label{figure:batcharrow}
\end{figure}

Just as in the anytime case, one can ask if there are conditions that support the design sequences of batch tests $\bphi=(\bphi^1, \ldots, \bphi^N)$ so that \eqref{def:testthenpunsihbatch} is effective. We introduce two such conditions. While they resemble in spirit  those introduced for anytime testing (\Cref{def:lowlevel} and \Cref{def:highpower}), they are slightly less demanding, and more easily satisfied for arbitrary deviations. 

First, the probability to wrongfully detect a deviation \textit{on a given batch} should be small for most continuation games. Recall that for any $t\geq0$, we define $k_t = \lfloor t/L\rfloor$ to be the batch to which $t$ belongs.
\begin{condition}\label{eq:batchtypeierror}There exist  $p_L\in(0,1)$ and $q_L\in(0,1)$ such that for any $t\geq0$, there exists a set of histories $\mathsf{H}_t\subset\cH_t$ with $\P^{\bbs_{\bv}}(\hist_t\in\mathsf{H}_t)\geq 1 -q_L$ satisfying, for any $\hist_t\in\mathsf{H_t}$
\begin{align*}
\P^{\bbs_{\bv}}(\kappa_{\bphi}\geq k+1\,|\, \kappa_{\bphi}\geq k\,,\,\hist_t)\geq 1-p_L\qquad\text{for any}\quad k\geq k_t\eqsp.
\end{align*}
\end{condition}Informally, \Cref{eq:batchtypeierror} ensures that at any point of time, starting from an history that is likely under $\bbs_{\bv}$ (with probability at least $1-q_L$), when players are still playing according to the cooperative strategy $\bw_{\bv}$  (given that $\kappa_{\bphi}\geq k$), the event $\{\kappa_{\bphi}=k\}$, which corresponds to wrongfully triggering the punishment phase is not likely (with probability at most $p_L$). The probabilities $p_L$ and $q_L$ should be typically be thought of as small. Note that this assumption is weaker than \Cref{def:lowlevel}, since we only require the Type I error to be low \textit{within each batch}. We let $p_L$ and $q_L$ depend on the batch size $L$, since we expect this probability to decrease as batch size increases.

Second, for any player $i\in[N]$, a deviation from $w^i _{\bv}$ that remains undetected on a batch should not bring too much utility gain.
\begin{condition}\label{eq:utilitygain}
For $\beta \in\ooint{0,1}$,
there exists $\Delta_L$ such that for any $k\geq0$, player $i\in[N]$  sequence of actions $(a^i_{t})_{t\in\cB_k} \in (\cA^i)^L$, 
$$
\2{\phi^i_k = 0}\sum_{t\in\cB_k}\beta^t\,u^i (a^i _t, w^{-i}_{\bv})\leq \sum_{t\in\cB_k}\beta^t (v^i+\Delta_L) \eqsp,$$where $u^i (a^i, w^{-i}_{\bv})=\sum_{a^{-i}}u^i(a^i ,a^{-i})\,w^{-i}_{\bv}[a^{-i}].$
\end{condition}Broadly speaking, \Cref{eq:utilitygain} ensures that if no deviation is detected on a batch (as captured by $\2{\phi^i _k = 0}$), then no sequence of actions $(a^i_t)_{t\in\cB_k}$ for player $i$ can result in a discounted payoff greater than $v^i + \Delta_L$ on this batch when opponents stick to $w^{-i}_{\bv}$.

Recall that \Cref{def:highpower} required the probability of not triggering a punishment under deviation to be bounded. This requirement was precisely why, in the anytime framework, we restricted attention to stationary deviations $\tilde{\mathfrak{S}}$: establishing such a guarantee for arbitrary deviations was not feasible. In contrast, \Cref{eq:utilitygain} merely requires that the utility gain from deviating be bounded. As we shall see later, this condition holds for arbitrary deviations under the class of tests considered in \Cref{subsection:batchconcentration}. This is the key advantage of the batch approach, and it explains why it can accommodate a substantially richer class of strategies than anytime testing. As a side remark, note that if players were to observe realized payoffs rather than actions, they could test condition \Cref{eq:utilitygain} directly. This observation suggests a natural extension of our framework to settings with payoff-based monitoring rather than action-based monitoring, which we leave for future work.

We show below that, when Conditions \ref{eq:batchtypeierror} and \ref{eq:utilitygain} are satisfied, a Folk Theorem can be recovered using batch testing.

\begin{restatable}{theorem}{folkbatchtesting}\label{theorem:batch_folk}
    Let $\varepsilon>0$, $\beta \in\ooint{0,1}$ and consider $\bbs_{\bv}$ defined in \eqref{def:testthenpunsihbatch} with $\bphi$ satisfying  \Cref{eq:batchtypeierror}  and \Cref{eq:utilitygain}. Let $\bv\in\cU^1\times\ldots\times\cU^N$ be a feasible payoff such that $v^i \geq\underline{u}^i$ for $i\in[N]$. Suppose in addition that for any  $i\in[N]$,
    \begin{equation}
        \label{eq:patiencebatch}
    \beta^L \geq \frac{\bar{u}^i - v^i - \Delta_L}{\bar{u}^i - \underline{u}^i}\eqsp.
    \end{equation}
    Then, it holds
\begin{enumerate}
    \item For any $i\in[N],\,$ $\beta^{L}\parenthese{1-\frac{p_L}{1-\beta^L}}\,v^i \leq U^i (\bbs_{\bv})\leq v^i\eqsp,$
    \item $\bbs_{\bv}$ is a $\parenthese{\beta^{-L}(1-\beta^{2L}+\Delta_L) + \frac{p_L}{1-\beta^L}\, , \,q_L}$-HP-SPNE.
\end{enumerate}
\end{restatable}

\Cref{theorem:batch_folk} shows that batch testing, provided \Cref{eq:batchtypeierror} and \Cref{eq:utilitygain} hold, yields an approximate Folk Theorem. This result exhibits both similarities to and differences from its anytime counterpart, \Cref{theorem:folk_anytime_main}.

Regarding similarities, both results ensure that $U^i(\bbs_{\bv})$ remains close to $v^i$. In the anytime case, the approximation becomes sharper as the Type I error bound $\gamma$ becomes smaller; similarly, in the batch case, it improves as $p_L$ decreases. Moreover, in both cases the strategy profile constitutes an approximate equilibrium. In the anytime setting, the approximation term shrunk with $\gamma$ and $\varepsilon$, the magnitude of the deviation; similarly, in the batch setting, it vanishes as $p_L$ and $\Delta_L$, the maximum utility gain on each batch, decrease. Finally, the patience condition \eqref{eq:patiencebatch} involves $\beta^L$, just as \eqref{eq:patienceconditionanytime} involves $\beta^{\ttau_{\varepsilon}  }$. The intuition is identical: because tests are conducted only at the end of batches, players may enjoy up to $L$ rounds of deviation before detection.

However, \Cref{theorem:batch_folk} and \Cref{theorem:folk_anytime_main} differ in two important respects. First, the batch construction of $\bbs_{\bv}$ is an equilibrium with respect to \textit{all} possible deviations, whereas the anytime result provides guarantees only against stationary deviations. This is because \Cref{eq:utilitygain} is more easily met for arbitrary deviations than \Cref{def:lowlevel}. Second, the batch equilibrium is subgame perfect (except for a small fraction $q_L$ of continuation games), while in the anytime setting $\bs_{\bv}$ is a Nash equilibrium only.  The reason behind this difference is that for any step $t\geq0$, the anytime procedure keeps track of the empirical frequencies $\hat{w}^i_t (a)$ using \textit{all} observed actions up to time $t$. Thus, a history $\hist_t$ could be constructed at any round $t$ so as to strongly support $\rmH^i_0$ for a given player $i$, thereby allowing that player to deviate without being detected quickly.  On the contrary, the batch testing procedure  forgets past actions and only focuses on the current batch. This forgetful property, commonly referred to as \textit{bounded recall} in the repeated games literature \citep{mailath2006repeated}, is precisely why the construction yields subgame-perfect equilibria. For any continuation game starting at $t\geq 0$, the history  $h_t$ influences only the test statistic for the \textit{current} batch and does not accumulate as evidence into later batches. All in all, the fact that the equilibrium featured in \Cref{theorem:batch_folk} accommodates any deviation and is subgame perfect  makes this result substantially stronger than \Cref{theorem:folk_anytime_main} from a game-theoretic perspective.

\subsection{Batch testing via concentration}\label{subsection:batchconcentration}

In order to implement the batch test-then-punish strategy,  it remains to determine sequences of tests $\bphi=(\bphi^1, \ldots, \bphi^N)$ which satisfy \Cref{eq:batchtypeierror} and \Cref{eq:utilitygain}. We propose a very simple procedure: on each batch, players collect the pure actions played by player $i$, compute their empirical frequencies, and check whether they match with $w^{i}_{\bv}$. Formally, writing $\cA^i=\{a^i(1),\ldots,a^i(K)\}$, for any $k\geq 0$ and $i\in[N]$, we let
\begin{equation}\label{def:empiricalfrequencies}
\hat{w}^{i}_k = L^{-1}\parenthese{\sum_{t\in\cB_k}\1\{A^{i}_t = a^i(1)\},\ldots, \sum_{t\in\cB_k}\1\{A^{i}_t = a^i(K)\}}^\top\eqsp,    
\end{equation}
be the vector of empirical frequencies of pure actions played by $i$ on batch number $k$. Then, for a given threshold $\delta \in(0,1)$, we test $\rmH^i_{0,k}$ with the test statistic
\begin{equation}
    \label{def:test_batch}
    \tphi^i_k = \1\{\lVert \hat{w}^{i}_k - w^{i}_{\bv}\rVert_1 \geq \delta\} \eqsp.
\end{equation}The intuition behind \eqref{def:test_batch} is simple. If $\rmH^i _{0,k}$ is true, then the empirical distribution $\hat{w}^{i}_k$ concentrates around $w_{\bv}^{i}$ with high probability. The contrapositive immediately tells us that if $\lVert \hat{w}^{i}_k - w_{\bv}^{i}\rVert_1$ is high, we can suspect that the null $\rmH^i _{0,k}$ is false. We show in the next proposition that our test \eqref{def:test_batch} indeed fulfills the desiderata in \Cref{eq:batchtypeierror} and \Cref{eq:utilitygain}.

\begin{restatable}{proposition}{requirementsbatch}\label{proposition:conditionsbatch}
    Let $\beta,\delta \in\ooint{0,1}$. Consider for any  $i\in[N]$, the test defined in \eqref{def:test_batch}
and the resulting strategy $\bbs_{\bv}$ defined in \eqref{def:testthenpunsihbatch}. Then, 
    \begin{enumerate}
        \item \Cref{eq:batchtypeierror} holds with $p_L=q_L =2KN\exp(-2L\delta^2 / K^2)\eqsp,$
        \item \Cref{eq:utilitygain} holds with $\Delta_L =\delta +3(1-\beta^L)\eqsp.$
    \end{enumerate}
\end{restatable}
The first point of \Cref{proposition:conditionsbatch} follows directly from a concentration argument. In the second point, $\Delta_L$ features two terms. The first one stems from the possibility for player $i$ to slightly deviate from $w^{i}_{\bv}$ within a $\delta$ range, without triggering the test of player $i$. The second one follows from the ability of player $i$ to re-order pure actions within a batch without changing $\hat{w}^{i}_k$,  but in a way that most valuable actions are played at the beginning of the batch to enjoy a better discounting.  

It is worth noting two trade-offs in $L$ and $\delta$ appear in \Cref{proposition:conditionsbatch}. On the one hand, a larger \(L\) ensures that \(\hat{w}^{i}_k\) concentrates more tightly around \(w^i_{\bv}\) under \(\rmH^i_{0,k}\), which is statistically desirable ($p_L$ decreases with $L$). On the other hand, a large $L$ gives the opponent more room to deviate within a batch and potentially obtain a sizeable payoff gain, since detection can only occur at the end of that batch ($\Delta_L$ increases with $L$). A similar tension arises for the threshold $\delta$. A larger $\delta$ reduces the risk of false positives, i.e., punishing based on fluctuations in $\lVert \hat{w}^{i}-\w^{i}_{\bv}\rVert_1$ that are merely stochastic ($p_L$ decreases as $\delta$ increases). However, setting $\delta$ too high may encourage the opponent to deviate from \(w^{i}_{\bv}\), as such deviations become less likely to be detected ($\Delta_L$ increases with $\delta$).

That said, it is possible to strike a balance by carefully tuning $L$ and $\delta$, so both $p_L$ and $\Delta_L$ are small. This observation, in conjunction with \Cref{theorem:batch_folk}, lead to the following corollary.

\begin{restatable}{corollary}{corollaryfolkbatch}\label{corollary:folk_batch}
Let $\bv=(v^1, \ldots,v^N)\in\cU^1 \times\ldots\times\cU^N$ such that $v^i \geq \underline{u}^i$ for any $i\in[N]$, $\varepsilon>0$ and $\beta \in \ooint{0,1}$. Consider for any  $i\in[N]$, the test defined in \eqref{def:test_batch} with level  $\delta = \epsilon/16$ and batch size  \begin{align*}
 L=\left\lceil\frac{258K^2}{\varepsilon^2}\ln\parenthese{\frac{126KN}{\varepsilon^2}}\right\rceil\eqsp,
\end{align*}
and the resulting strategy $\bbs_{\bv}$ defined in \eqref{def:testthenpunsihbatch}. 
If $1-\varepsilon/16 \leq \beta^L \leq 1 -\varepsilon/32$, then
\begin{enumerate}[label=(\roman*)]
    \item $\parenthese{1 -\varepsilon}v^i \leq U^i (\bbs_{\bv})\leq v^i\eqsp,$
    \item $\bbs_{\bv}$ is a $(\varepsilon\,,\,\varepsilon)$-HP-SPNE.
\end{enumerate}
\end{restatable}


This corollary shows that, provided $\beta$ is sufficiently large, the batch test-then-punish strategy is an approximate HP-SPNE, and it yields a payoff that remains very close to the target $\bv$.  In particular, it is possible to drive $\varepsilon$ to zero, provided $\beta$ is sufficiently close to one. In this case, we almost recover the perfect monitoring result in \Cref{theorem:folk}, in spite of players only observing pure actions drawn from strategies. We once again stress that this result proves that $\bbs_{\bv}$ is a subgame perfect equilibrium, and that strategies are not assumed to be stationary. Finally, we note that the patience condition on $\beta$ in \Cref{corollary:folk_batch} is slightly more demanding than in the original \Cref{theorem:batch_folk}. Indeed, we require that $\beta^L\geq1-\varepsilon/16$, which is stronger than  $\beta^L\geq (\bar{u}^i - \underline{u}^i-\varepsilon/16)(\bar{u}^i - \underline{u}^i)^{-1}$, as well as an upper bound $\beta^L \leq 1 -\varepsilon/32$. The latter is necessary to bound the terms in $\beta^L$ appearing in the approximation terms of \Cref{theorem:batch_folk}. Although the condition on $\beta$ may appear restrictive, it can be read as follows: fixing $\beta \in (0,1)$, the construction yields an $(\varepsilon,\varepsilon)$-HP-SPNE for every $\varepsilon>0$ such that $1-\varepsilon/16 \leq \beta^L \leq 1 -\varepsilon/32$ holds (with an the appropriate choice of $L=L(\varepsilon)$); choosing the minimal admissible $\varepsilon$ then gives the tightest guarantee.

While \Cref{theorem:batch_folk} is a stronger result than \Cref{theorem:folk_anytime_main} from a game-theoretic point of view, the batch testing approach does \emph{not} inherit all the convenient properties of anytime testing. In particular, anytime testing allows us to explicitly control the Type~I and Type II errors (see \Cref{proposition:anytime_low_level} and \Cref{theorem:probaboundanaytime}). On the contrary, the batch setting does not offer any guarantee against Type I error. Actually, under the batch test-then-punish strategy, a wrongful punishment happens in finite time with probability one when the cooperative strategies are not degenerate.
\begin{restatable}{proposition}{nocontroltypei}\label{lemma:nocontrolbatch}
  Let $\beta,\delta \in\ooint{0,1}$. Consider for any  $i\in[N]$, the test defined in \eqref{def:test_batch}
and the resulting strategy $\bbs_{\bv}$ defined in \eqref{def:testthenpunsihbatch} for some feasible payoff $\bv$ such that $v^i\geq\underline{u}^i$ for $i\in[N]$. Suppose that 
\begin{equation}
\text{ $\min\limits_{j\in[N]}\defEns{\lVert w^j_{\bv} \rVert_0 = \sum\limits_{\ell=1}^K \2{w^j_{\bv, \ell}\ne 0}}>1$ and $\delta < 1- 1/\mathop{\min}\limits_{j\in[N]}\lVert w^j_{\bv} \rVert_0$ } \eqsp.
\end{equation}
Then  $\P^{\bbs_{\bv}}(\kappa_{\bphi} < \infty)=1$.
\end{restatable}

We emphasize that \Cref{lemma:nocontrolbatch} does not contradict \Cref{corollary:folk_batch}; rather, it only means that the punishment phase occurs in finite time with probability one. However, the punishment occurs late enough on average so it does not significantly impact the utilities, and thus still allowing players to cooperate as highlighted by \Cref{theorem:batch_folk}. Nevertheless, this feature may be undesirable from a fairness perspective or in environments where players are risk averse, since it implies that punishment is eventually unavoidable. In such settings, implementing the anytime test–then–punish mechanism may be preferable. This comparison highlights that the two approaches involve distinct trade-offs, and the choice between them should ultimately depend on the institutional context and on players’ attitudes toward risk and fairness.

\section{Conclusion}

The folk theorem has long stood as the canonical way to study cooperation in repeated games, relying on punishments triggered by deviations. This paper proposes a new framework for analyzing how cooperation can be sustained
in multiplayer repeated games through statistical testing: players use mixed strategies but observe only realized pure actions, and enforcement is statistical, explicitly managing false alarms and missed detections. We introduce equilibrium notions that safely disregard tail probability histories generated by the monitoring process, and we provide a general test-then-punish strategy that yields two Folk-theorem-type resultsn showing that any feasible and individually rational payoff profile can be supported as an approximate equilibrium for sufficiently patient players. This study is, to our knowledge, the first to leverage the flexibility offered by the hypothesis testing framework to analyze repeated games with imperfect public monitoring.

We also provide two implementations with the corresponding tests. The anytime construction uses sequential inference to provide implementable strategies with uniform Type I error control over the whole horizon and finite-sample detection guarantees, while the batch construction loses anytime Type I error control to handle richer deviation patterns and to obtain a high-probability analogue of subgame perfection.

More broadly, the analysis highlights how repeated-game reasoning changes when interaction is mediated by data and inference. Promising directions for future work include richer information structures--e.g., private signals--heterogeneous agents, adaptive or nonstationary environments among others.

%
\bibliography{sample-bibliography}

\appendix
\section{Statistical Background}
\label{appendix:statistical_background}

\begin{definition}[Filtration]\label{definition:filtration}
A \textit{filtration} is a sequence of $\sigma$-fields $\{\mathcal F_t\}_{t\ge0}$ on a common probability space $(\Omega,\mathcal F,\mathbb P)$ such that
$$
\mathcal F_0 \subseteq \mathcal F_1 \subseteq \cdots \subseteq \cF \eqsp.
$$
The interpretation is that $\mathcal F_t$ represents all information available up to and including time $t$.
\end{definition}

\begin{definition}[Supermartingale]\label{def:supermartingale}
A real-valued stochastic process $(M_t)_{t\ge0}$ adapted to a filtration $\{\cF_t\}_{t \geq 0}$ is called a \textit{supermartingale} if it satisfies
\begin{align*}
\E[|M_t|]<\infty \quad\text{for any } t \geq 0 \quad\text{and}\quad
\E[M_t\mid \cF_{t-1}] \leq M_{t-1}\ \text{a.s.}
\end{align*}
If equality holds, $(M_t)_{t \geq 0}$ is a \textit{martingale}. Supermartingales describe processes whose conditional expectation cannot increase over time.
\end{definition}

\begin{definition}[E-process]\label{definition:eprocess}
An e-process (or \textit{test martingale}) is a nonnegative, $\{\cF_t\}_{t \geq 0}$-adapted process $(\tilde{E}_t)_{t\geq 0}$ satisfying
\begin{align*}
\E[\tilde{E}_t] \leq 1\quad\text{for any }t\geq 0 \eqsp.
\end{align*}
Under the null hypothesis, $(\tilde{E}_t)_{t \geq 0}$ is typically a nonnegative supermartingale with unit or subunit mean. For any $\gamma \in (0,1)$, the stopping rule ``reject at the first $t$ with $\tilde{E_t} \geq 1/\gamma$'' provides a sequential test controlling the Type-I error at level $\gamma$ uniformly over all stopping times.
\end{definition}

\begin{theorem}[Ville’s inequality, {\cite{doob1939jean}}]\label{theorem:ville_inequality}
Let $(\tilde{E}_t)_{t\ge0}$ be a nonnegative supermartingale with $\E[\tilde{E}_0] \leq 1$. Then, for any threshold $c>0$, we have that
\begin{align*}
\P\left(\sup_{t\geq 0} \tilde{E}_t \geq c\right) \leq 1/c \eqsp,
\end{align*}
which can also be rewritten $\P(\exists \, t \geq 0 \text{ s.t. } \tilde{E}_t \geq c) \leq 1/c$. In particular, setting $c=1/\gamma$ yields the inequality $\P(\sup_{t \geq 0} \tilde{E}_t \geq 1/\gamma)\leq \gamma$. This result ensures anytime-valid Type-I error control for tests based on e-processes.
\end{theorem}

\begin{restatable}{lemma}{lemmachernoffstein}\label{lemma:chernoff_stein}[Chernoff-Stein lemma, see \cite{chernoff1952measure}, \cite{cover1999elements}]
    Let $(A_t)_{t\geq 0}$ be a sequence of i.i.d. random variables with value in a discrete set $\cA$.
    Consider the simple hypotheses
\begin{align*}
\rmH_0 \colon A_t\sim \tilde{w}
\qquad\text{vs}\qquad
\rmH_1 \colon \ A_t\sim q \eqsp,
\end{align*}
where $q(a)>0$ and $\tilde{w}(a)>0$ for all $a\in \cA$, and
$\KL (q\|\tilde{w})=\sum_{a\in \cA} q(a)\log\frac{q(a)}{\tilde{w}(a)}<\infty$.
Let $\phi_n$ denote the test at time $n$, dependent on the history $\hist_n$, and let $\phi_n^{-1}(\{1\}) \subseteq \cH$
be the acceptance region for hypothesis $\rmH_1$ induced by $\phi_t$.
Define the Type I probability error
\begin{align*}
\gamma_n^{(1)} &= \P_{\tilde{w}}(\phi_n = 1) \eqsp, 
\end{align*}
as well as
\begin{align*}
\gamma^{(2)}_n = \min_{\hist_n \subseteq \sigma(A_0, \ldots, A_n) \colon \P_{\tilde{w}}(\phi_n = 1)\leq \gamma^{(1)}} \P_q(\phi_n = 0) \eqsp.
\end{align*}
Then, we have that
\begin{align*}
\lim_{n \to \infty} 1/n \log \gamma^{(2)}_n
= - \KL\left(\tilde{w} \| q\right) \eqsp.
\end{align*}
\end{restatable}

\Cref{lemma:chernoff_stein}
characterizes the best achievable exponential decay rate of the Type II error (missed detection) when the type-I error (false alarm) is constrained to remain under a fixed constant.

\paragraph{Mixture test martingale.} Finally, we highlight the mixture e-processes that could be of interest to extend our setting to unobserved utilities. Fix a symmetric prior density $\pi$, strictly positive and continuous on $\R$ (typically a normal $\mathcal N(0,\tau^2)$) such that $\int_{\R} \pi(\lambda)\,d\lambda=1$. Define the stochastic process $\{\tilde{E}_t\}_{t \geq 0}$ as $\tilde{E}_0=1$ and for any $t\geq 1$
\begin{equation}\label{equation:e-mixture}
\tilde{E}_t = \int_{\mathbb{R}} \exp\,\{ \lambda S_t -\lambda^2\sigma^2\,t/2 \}\, \pi(\lambda)\,d\lambda \eqsp,
\end{equation}
as well as the test
\begin{equation}\label{equation:definition_test_bis}
    \phi_t = \1\{\tilde{E}_t \geq 1/\gamma \} \qquad \text{for any } t \geq 1 \eqsp.
\end{equation}
For a pre-specified level $\gamma \in(0,1)$, one can reject $\rmH_0$ as soon as $\phi_t = 1$.

\emph{Anytime-validity.} Under $\rmH_0$, the process $(\tilde{E}_t)_{t\geq 0}$ in \eqref{equation:e-mixture} is a nonnegative $\{\mathcal F_t\}$-supermartingale with $\E[\tilde{E}_t]\leq 1$ for any $t\geq 0$. Hence, by Ville’s inequality, we have that
\begin{align*}
\P_{\rmH_0}\,\left(\sup_{t\ge0} \tilde{E}_t \geq 1/\gamma\right) \leq \gamma \eqsp.
\end{align*}
Consequently, the stopping rule ``reject at the first $t$ with $\phi_t = 1$'' has an anytime-valid Type-I error smaller than $\gamma$.

\section{Proof of the Folk theorem under perfect monitoring}
\label{section:proofperfect}
\theoremfolk*

\begin{proof}[Proof of \Cref{theorem:folk}]
    Let $t_0 \geq 0$ and $\hist_{t_0}\in\cH$ be the history generated by $\bs_{\bv}$ observed at time $t_0$. We consider two cases. Fix some $i\in[N]$. 

    First, assume there exists $t< t_0$ such that $\bw_t \in\hist_{t_0}$ and $w^j _t \ne w^j _{\bv}$ for some $j\in[N]$, so we are already in the punishment phase. By definition of $\bs_{\bv}$ in \eqref{def:grim}, players only play $\boldsymbol{b}=( b^1, \ldots, b^N)$ from time $t+1$ on. The continuation payoff starting from $t_0$ of player $i$ then reads
    \begin{align*}
        U^i (\bs_{\bv}\,;\,\hist_{t_0})=(1-\beta)\,\sum_{\ell=t_0}^\infty \beta^{\ell-t_0}u^i (b^i , b^{-i})=u^i (b^i , b^{-i})\eqsp.
    \end{align*}Now, consider any deviation $s \in\cS$ for the continuation game from $t_0$. We have
    \begin{align*}
        U^i (s ,s^{-i}_{\bv}\,;\,\hist_{t_0})&=(1-\beta)\,\sum_{\ell=t_0}^\infty \beta^{\ell-t_0}u^i (s(\hist_{\ell-1}), b^{-i})\\
        &\leq (1-\beta)\,\sum_{\ell=t_0}^\infty \beta^{\ell-t_0}u^i (b ^i, b^{-i})=u ^i (b^i , b^{-i})=U^i (\bs_{\bv}\,;\, \hist_{t_0}) \eqsp,
    \end{align*}because $\boldsymbol{b}$ is a Nash equilibrium.
    
    Second, assume that $\bw _t = (w^1_{\bv},\ldots, w^N_{\bv})$ for any $\bw_t \in\hist_{t_0}$. By definition of $\bs_{\bv}$, we have
    $$
U^i (\bs_{\bv}\,;\,\hist_{t_0})=(1-\beta)\,\sum_{s=t_0}^\infty \beta^{s-t_0}u^i (w^i_{\bv}, w^{-i}_{\bv})=v^i \eqsp.
    $$Now, consider a deviation $s\in\cS$ for player $i$, and denote by $T = \min\{t\geq t_0 \colon \: s(\hist_t)\ne w^i_{\bv}\}<\infty$. Player $i$'s continuation payoff reads
    \begin{align*}
        U ^i (s, s_{\bv}^{-i}\,;\, \hist_{t_0})&=(1-\beta)\,\parenthese{\sum_{s=t_0}^{T-1}\beta^{s-t_0} u^i (w^i_{\bv},w^{-i}_{\bv})+\beta^{T-t_0} u^i (s(\hist_T), w^{-i}_{\bv})+\sum_{s=T+1}^{\infty}\beta^{s-t_0} u^i(b^i, b^{-i})}\\
        &\leq (1-\beta)\,\parenthese{\sum_{s=t_0}^{T-1}\beta^{s-t_0} v^i+\beta^{T-t_0} \bar{u}^i+\sum_{s=T+1}^{\infty}\beta^{s-t_0} \underline{u}^i}\\
        &= (1-\beta^{T-t_0})v^i + (1-\beta)\beta^{T-t_0}\,\bar{u}^i + \beta \,\beta^{T-t_0}\,\underline{u}^i 
    \end{align*}And finally, since $\beta\geq (\bar{u}^i-v^i)(\bar{u}^i-\underline{u}^i)^{-1}$,
    $$U^i (s,s^{-i}_{\bv}\,;\,\hist_{t_0})\leq \bar{u}^i - \beta^{T-t_0}((1-\beta)\bar{u}^i-\beta\underline{u}^i + v^i) \leq v^i  = U^i (\bs_{\bv}\,;\,\hist_{t_0})\eqsp.$$This proves that $\bs_{\bv}$ is indeed an SPNE. 
\end{proof}
\section{Proofs in the anytime setting}

\lemmasmalldeviation*

\begin{proof}
    Let $i\in[N]$ $s^i \in\mathsf{B}^i(\bw_{\bv},\varepsilon)$ be a small deviation for player $i$. Since for any $t>\tau_{\bpsi}$ and $w^i_t \in\Delta(\cA^i)$, $u^i (w^i_t, w^{-i}_{\bv})\leq u^i (b^i, b^{-i})=\underline{u}^i$, we have
\begin{align*}
    U^i (s^i , s^{-i}_{\bv})-v^i & \leq (1-\beta)\,\EE{\sum_{t=0}^{\tau_{\bpsi}} \beta^t u^i (s^i (\hist_t), w^{-i}_v)+\sum_{t=\tau_{\bpsi}+1}^{\infty}\beta^t \underline{u}^i}  - v^i \\
     &=(1-\beta)\,\left(\EE{\sum_{t=0}^{\tau_{\bpsi}} \beta^t u^i (s^i (\hist_t), w^{-i}_v)+\sum_{t=\tau_{\bpsi}+1}^{\infty}\beta^t \underline{u}^i} \right)-(1-\beta)\,\EE{\sum_{t=0}^\infty \beta^t\,v^i }\\
    &= (1-\beta)\, \left(\EE{\sum_{t=0}^{\tau_{\bpsi}} \beta^t u^i (s^i (\hist_t),w^{-i}_{\bv})}-\EE{\sum_{t=0}^{\tau_{\bpsi}} \beta^t v^i} \right) \\
    &\quad+(1-\beta)\,\underbrace{\parenthese{\EE{\sum_{t=\tau_{\bpsi}+1}^\infty \beta^t \underline{u}^i} -\EE{\sum_{t=\tau_{\bpsi}+1}^\infty \beta^t v^i }}}_{\leq 0} \eqsp,
    \end{align*}
and since $v^i = u^i (\bw_{\bv})$
\begin{align*}
    U^i (s^i , s^{-i}_{\bv})-v^i &\leq (1-\beta)\,\EE{\sum_{t=0}^{\tau_{\bpsi}} \beta^t \parenthese{u^i (s^i (\hist_t), w^{-i}_{\bv})-u^i (w^i_v, w^{-i}_v)}} \eqsp.
\end{align*}
Now, define $\zeta^i = (u^i (a_1, w^{-i}_{\bv}),\ldots, (a_K, w^{-i}_{\bv}))\in\R^K$, so $u^i (w, \tilde{w}^{-i})=\ps{w}{\zeta^i}$ for any $w\in\Delta(\cA)$. We can rewrite the previous expression as
\begin{align*}
    U^i (s^i , s^{-i}_{\bv})-v^i &=(1-\beta)\,\sum_{t=0}^{\tau_{\bpsi}} \beta^t \ps{s^i (\hist_t)-w^i_{\bv}}{\zeta_i}\\
    &\leq (1-\beta)\,\sum_{t=0}^{\tau_{\bpsi}} \beta^t \lVert \zeta_i \rVert_\infty \lVert s^i (\hist_t)-w^i _{\bv}\rVert_1 \\
    &\leq (1-\beta)\,\varepsilon\,\sum_{t=0}^{\tau_{\bpsi}}\beta^t \\
    & \leq \varepsilon\eqsp,
\end{align*}
where we used $\lVert \zeta^i\rVert_\infty \leq 1$ and $\lVert s^i (\hist_t)-w^i_{\bv}\rVert_1\leq\varepsilon$, which is a consequence of $s^i \in\mathsf{B}^i(w_{\bv}, \varepsilon)$.
\end{proof}

\folkanytime*

\begin{proof}[Proof of \Cref{theorem:folk_anytime_main}]
We first prove $(i)$. Let $i\in[N]$. With $\tau = \min_{i\in[N]}\tau_{\bphi^i}$, and we have
\begin{align*}
    U^i (\bs_{\bv})&=(1-\beta)\,\E^{\bs_{\bv}}\parentheseDeux{\2{\tau_{\bpsi}<\infty}\parenthese{\sum_{t=0}^{\tau_{\bpsi}}\beta^{t }v^i + \sum_{t=\tau_{\bpsi}+1}^\infty\beta^{t } \underline{u}^i}+\2{\tau_{\bpsi}=\infty}\sum_{t=0}^\infty \beta^{t }v^i}\\
    &\geq (1-\beta)\,\E^{\bs_{\bv}}\parentheseDeux{\2{\tau_{\bpsi}=\infty}(1-\beta)^{-1}v^i}=\P^{\bs_{\bv}}(\tau_{\bpsi}=\infty)\,v^i\eqsp.
\end{align*}
Now, by a union bound and \Cref{def:lowlevel}
\begin{align*}
\P^{\bs_{\bv}} (\tau_{\bpsi}=\infty)&=1-\P^{\bs_{\bv}} (\{\tau_{\bpsi^1}<\infty\}\cup\ldots\cup\{\tau_{\bpsi^N}<\infty\})\\
&\geq 1 - (\P^\bs_{\bv} (\tau_{\bpsi^1}<\infty)+\ldots+\P^{\bs_{\bv}} (\tau_{\bpsi^N}<\infty)) \\
& \geq 1-N\frac{\gamma}{N} \\
& = 1 - \gamma\eqsp,
\end{align*}
hence we obtain
\begin{equation}\label{eq:anytimelowerbound}
    U^i (\bs _{\bv})\geq(1-\gamma)v^i\eqsp.
\end{equation} Regarding the upper bound, we note that by assumption, $v^i\geq\underline{u}^i$, and therefore
\begin{align*}
  U^i (\bs_{\bv})&=(1-\beta)\,\E^{\bs_{\bv}}\parentheseDeux{\2{\tau_{\bpsi}<\infty}\parenthese{\sum_{t=0}^{\tau_{\bpsi}}\beta^{t }v^i + \sum_{t=\tau_{\bpsi}+1}^\infty\beta^{t } \underline{u}^i}+\2{\tau_{\bpsi}=\infty}\sum_{t=0}^\infty \beta^{t }v^i}\\
        &\leq(1-\beta)\,\E^{\bs_{\bv}}\parentheseDeux{\2{\tau_{\bpsi}<\infty}\parenthese{\sum_{t=0}^{\tau_{\bpsi}}\beta^{t}v^i + \sum_{t=\tau_{\bpsi}+1}^\infty\beta^{t} v^i}+\2{\tau_{\bpsi}=\infty}\sum_{t=0}^\infty \beta^{t}v^i}\\
        &= v^i.
\end{align*}
We now focus on point (ii). Let $\varepsilon >0$ and $s\in\mathfrak{S}$ be a deviation for player $i$. We consider two cases. First, assume that $s \in\mathfrak{S}(\bw_{\bv}\,,\,\varepsilon)$. For any player $i$, we have 
\begin{align}
U^i (s,s^{-i}_{\bv})
    &=(1-\beta)\,\E^{(s,s^{-i}_{\bv})}\parentheseDeux{\2{\tau_{\bpsi}=\infty}\sum_{t=0}^\infty \beta^{t}u^i (A^1_t,\ldots,A^N_t)}\nonumber\\
    & \quad + (1-\beta)\,\E^{(s,s^{-i}_{\bv})}\parentheseDeux{\2{\tau_{\bpsi}<\infty}\parenthese{\sum_{t=0}^{\tau_{\bpsi}} \beta^{t}u^i (A^1_t,\ldots,A^N_t)+\sum_{t=\tau_{\bpsi}}^\infty \beta^{t}u^i (A^1_t,\ldots,A^N_t)}}\nonumber\\
    &\leq \P^{(s,s^{-i}_{\bv})}(\tau_{\bpsi}=\infty)\,\bar{u}^i +\E^{(s,s^{-i}_{\bv})} [\2{\tau_{\bpsi}<\infty}((1-\beta^{\tau_{\bpsi}+1})\bar{u}^i + \beta^{\tau_{\bpsi}+1}\underline{u}^i) ]\nonumber\eqsp.
    \end{align}
By \Cref{def:lowlevel}, $\E^{(s,s^{-i}_{\bv})}[\tau_{\bpsi}]<\infty$ so in particular $\P^{(s,s^{-i}_{\bv})}(\tau_{\bpsi}=\infty)=0$. Moreover, $\2{\tau_{\bpsi}<\infty}\leq1$, so we obtain
\begin{align}
    U^i (s,s^{-i}_{\bv})&\leq \E^{(s,s^{-i}_{\bv})} [(1-\beta^{\tau_{\bpsi}+1})\bar{u}^i + \beta^{\tau_{\bpsi}+1}\underline{u}^i]= \bar{u}^i-\E^{(s,s^{-i}_{\bv})} [\beta^{\tau_{\bpsi}+ 1}(\bar{u}^i - \underline{u}^i)]\nonumber
\end{align}
Since $\tau\mapsto -\beta^{\tau}(\bar{u}^i - \underline{u}^i)$ is concave, we have by Jensen inequality that
\begin{align}
    U^i (s,s^{-i}_{\bv}) &\leq \bar{u}^i - \beta^{\E^{(s,s^{-i}_{\bv})} [\tau_{\bpsi}]}(\bar{u}^i - \underline{u}^i)\leq \bar{u}^i - \beta^{\tau^{(\varepsilon)}}(\bar{u}^i - \underline{u}^i) \nonumber \eqsp,
\end{align}
where we denote $\tau^{(\varepsilon)}=\sup_{\bs\in\mathfrak{S}(\bw_{\bv},\varepsilon)}\E^{\bs}[\tau_{\bpsi}]$. Finally,
\begin{align*}U^i (s,s^{-i}_{\bv}) \leq (1-\gamma)v^i \leq U^i (\bs _{\bv})\leq U^i (\bs_{\bv})+(\gamma+\varepsilon)\eqsp,
\end{align*}where we used our assumption $\beta^{\tau^{(\varepsilon)}}\geq (\bar{u}^i - (1-\gamma)v^i)(\bar{u}^i- \underline{u}^i)^{-1}$ in the first inequality, \Cref{eq:anytimelowerbound} in the second inequality, and $\gamma+\varepsilon\geq 0$ in the last inequality.

Now, assume that $\bs\in\mathrm{B}^1 (\bw_{\bv},\varepsilon)\times\ldots\times\mathrm{B}^N (\bw_{\bv},\varepsilon)$. By \Cref{lemma:smalldeviation} and \eqref{eq:anytimelowerbound}, we have for any player $i\in[N]$ that
\begin{align}\label{eq:gainsmalldeviation}
U^i (s, s^{-i}_{\bv})&\leq v^i + \varepsilon = (1-\gamma)v^i + (\gamma\,v^i+\varepsilon) \leq U^i (\bs _{\bv}) + (\gamma+\varepsilon)\eqsp.
\end{align}
Thus, $\bs_{\bv}$ is a $(\gamma+ \varepsilon)$-NE which concludes the proof.
\end{proof}

\begin{lemma}[Lower bound on $\kappa_{\bphi}$]\label{lemma:barkappa}
Suppose that \Cref{eq:batchtypeierror} holds for some $(p_L,q_L)\in(0,1)^2$.
Fix $t\geq 0$ and a public history $\hist_t\in\mathsf{H}$.
Then there exists a sequence $(U_k)_{k\geq k_t}$ of i.i.d.\ $\mathrm{Uniform}([0,1])$ random variables,
independent of the processes $(\hat w^i_k)_{k\geq k_t}$ for all $i\in[N]$,
such that the following holds.

Define, for $k\geq k_t$,
$$
\bar\varphi_k = \2{U_k \leq p_L}\eqsp,
\qquad
\bar\kappa = \inf\{k\geq k_t \colon \bar\varphi_k = 1\},
$$
and recall that $\kappa = \min\{\kappa_{\bphi^1},\ldots,\kappa_{\bphi^N}\}$. Then, conditional on $\hist_t$
\begin{enumerate}
    \item $\kappa_{\bphi} \geq \bar\kappa$ almost surely, and $\kappa_{\bphi}$ is independent of $\bar\kappa$;
    \item $\bar\kappa \sim \mathrm{Geometric}(p_L)$.
\end{enumerate}
\end{lemma}

\begin{proof}
Fix $t\geq 0$ and $\hist_t\in\mathsf{H}$.
For each player $i\in[N]$ and each $k\geq k_t$, define
$$
\phi^i_k = \2{\|\hat w^i_k - w^i_{\bv}\|_1 \geq \delta} \eqsp.
$$
By construction, conditional on $\hist_t$, $\phi^i_k$ is a Bernoulli random variable with parameter
$$
p_k = \P\left(\|\hat w^i_k - w^i_{\bv}\|_1 \geq \delta \,\middle|\, \hist_t\right).
$$
By \Cref{eq:batchtypeierror}, we have $p_k \leq p_L$ for all $k\geq k_t$. We now construct a coupling.
Let $(U_k)_{k\geq k_t}$ be a sequence of i.i.d.\ $\mathrm{Uniform}([0,1])$ random variables,
independent of $(\hat w^i_k)_{k\geq k_t}$ for all $i\in[N]$.
Define
\begin{align*}
\varphi^i_k = \2{U_k \leq p_k} \; ,
\qquad
\bar\varphi_k = \1\{U_k \leq p_L\} \eqsp.
\end{align*}
Then, conditional on $\hist_t$, we have
\begin{align*}
\phi^i_k \overset{\mathrm{d}}{=} \varphi^i_k \; ,
\qquad
\text{and} \qquad
\varphi^i_k \leq \bar\varphi_k \quad \text{almost surely} \eqsp.
\end{align*}
Define the corresponding stopping times
\begin{align*}
\kappa_{\bphi^i} = \inf\{k\geq k_t \colon \phi^i_k = 1\},
\qquad
\bar\kappa = \inf\{k\geq k_t \colon \bar\varphi_k = 1\} \eqsp.
\end{align*}
By the coupling inequality above, we obtain
\begin{align*}
\kappa_{\bphi^i}
\;\overset{\mathrm{d}}{=}\;
\inf\{k\geq k_t \colon \varphi^i_k = 1\}
\;\geq \; \bar{\kappa}
\qquad \text{almost surely} \eqsp.
\end{align*}
Since this holds for every $i\in[N]$, it follows that
\begin{align*}
\kappa = \min_{i\in[N]} \kappa_{\bphi^i} \geq \bar\kappa
\quad \text{almost surely} \eqsp.
\end{align*}
Moreover, $\bar\kappa$ is a measurable function of $(U_k)_{k\geq k_t}$ only,
whereas $\kappa$ is a measurable function of $(\hat w^i_k)_{k\geq k_t}$.
By construction, these collections are independent, hence $\kappa$ and $\bar\kappa$ are independent.

Finally, since $(\bar\varphi_k)_{k\geq k_t}$ are i.i.d.\ Bernoulli$(p_L)$ random variables,
it follows immediately that
\begin{align*}
\bar\kappa \sim \mathrm{Geometric}(p_L) \eqsp,
\end{align*}
which concludes the proof.
\end{proof}

\anytimefolk*

\begin{proof}[Proof of \Cref{corollary:folk_anytime}]The result directly follows from instantiating \Cref{theorem:folk_anytime_main} with $\gamma=\varepsilon$, and using the bound on $\tau^{(\varepsilon)}$ in \Cref{theorem:probaboundanaytime}.    
\end{proof}

\subsection{Proofs of the test-then-punish strategy with e-processes}

\anytimetestlowlevel*

\begin{proof}This is an immediate consequence of the Ville inequality (\Cref{theorem:ville_inequality}).
\end{proof}
\subsubsection{Type I error.}
Recall from \Cref{eq:hoanytime} that each player $j$ aims to test for any $i\ne j$
\begin{equation*}   
\rmH^i_{0} \colon s^{i}=\sigma^{i}_{\bv} \qquad\text{versus}\qquad \rmH^i_{1} \colon \; s^{i}\ne \sigma^i_{\bv}\eqsp.
\end{equation*}
$\rmH^i_{0}$ being true means that player $i$ sticks to the cooperative strategy $w^{i}_{\bv}$. On the contrary, rejecting $\rmH^i_{0}$ amounts to concluding that the player $i$ does not play the agreed-upon strategy. Also recall that we denote by $\E_0^i$ and $\P_0^i$ the probabilities and expectations under $\rmH_0^i$. We now introduce the filtration $\{\cF_t^i\}_{t\geq 0}$ for any player $i \in \{1, \ldots, N\}$ as
\begin{align*}
    \cF_t^i = \sigma(\{A_s^i\}_{0 \leq s \leq t}) \eqsp.
\end{align*}
By convention, we also define $\cF_{-1}^i = \{\varnothing, \Omega\}$.

\begin{theorem}\label{theorem:anytime} Under the null $\rmH_0^i$ defined in \eqref{eq:hoanytime}, the process $(E_t^i)_{t\ge0}$ defined in \eqref{eq:e-process} is a nonnegative $\{\mathcal F^i_t\}$-martingale with $\E_0^i[E_t^i]=1$ for any $t\geq 0$. Consequently, by Ville's inequality
\begin{align*}
\P_0^i \left(\sup_{t\ge0} E_t^i \,\geq \, 1/\gamma\right) \leq \gamma \eqsp.
\end{align*}
In particular, the stopping rule ``reject $\rmH_0^i$ at the first $t$ such that $E_t^i\ge 1/\gamma$'' controls the Type-I error at level $\gamma$.
\end{theorem}

\begin{proof}[Proof of \Cref{theorem:anytime}]
Since $\widehat w^{i}_{t}$ is $\cF_{t-1}^i$-measurable, using the law of total probability under the events $\{A_t^{i}=a\}_{a \in \cA}$, we get
\begin{align*}
\E_0^i \left[\frac{\widehat w^{i}_{t}(A_t^{i})}{w_{\bv}(A_t^{i})}\Bigm|\cF_{t-1}^i\right]
=\sum_{a\in \cA}\frac{\widehat w^{i}_{t}(a)}{w_{\bv}(a)} \P_0^i \,\left(A_t^{i}=a\mid\mathcal F_{t-1}^i\right)
=\sum_{a\in \cA}\widehat w^{i}_{t}(a)=1 \eqsp,
\end{align*}
where the second equality uses the fact that under $\rmH_0^i$, $\P_0^i(A_t^{i}=a \mid \cF^i_{t-1})=w_{\bv}(a)$. Therefore, $E_{t-1}^i$ being $\cF^i_{t-1}$-measurable, we have that
\begin{align*}
\E_0^i \left[E_t^i\mid \cF_{t-1}^i\right]
= E_{t-1}^i\ \E_0^i \left[\frac{\widehat w^{i}_{t}(A_t^{i})}{w_{\bv}(A_t^{i})}\Bigm|\cF_{t-1}^i\right] = E_{t-1}^i \eqsp,
\end{align*}
showing that $\{E_t^i\}_{t\geq 0}$ is a nonnegative martingale with unit mean (since $E_0^i=1$). Ville's inequality for nonnegative supermartingales then yields
$\P_0^i \big(\sup_{t\geq 0} E_t^i \geq 1/\gamma\big)\leq \gamma$, which implies the claimed anytime-valid Type-I control for the thresholding rule.
\end{proof}

\subsubsection{Expected Stopping Time Upper Bound}

We now formulate a general version of our results on the sequential test, that encompass the specific game setting considered in the paper. Let $(A_t)_{t \geq 0}$ be a sequence of i.i.d. random variables with distribution $\pi$ on $\{1, \ldots, K\}$ under $\P_\pi$. We denote by $\E_\pi$ the associated expectation. For $\pi_0$ a distribution on $\{1, \ldots, K\}$, we aim to test
\begin{equation}\label{equation:definition_null_appendix}
    \rmH_0 \colon \pi = \pi_0 \quad \text{versus} \quad \rmH_1 \colon \pi \ne \pi_0 \eqsp.
\end{equation}
Define the e-process $(E_t)_{t \geq -1}$ with $E_{-1}=1$ and for any $t \geq 0$
\begin{equation}\label{equation:definition_eprocess_appendix}
    E_t = \prod_{s=0}^{t} \frac{\widehat \pi_{s}(A_s)}{\pi_0(A_s)} \eqsp,
\end{equation}
where for any $a \in \{1, \ldots, K\}$
\begin{equation}\label{equation:definition_appendix_hat_pi}
\hat{\pi}_{t}(a) = \frac{N_{t}(a)+1}{t+K} \quad \text{ with } \quad N_{t}(a)=\sum_{s=0}^{t-1} \1\{A_s=a\} \eqsp,
\end{equation}
with the convention that $\hat{\pi}_{-1}(a) = 1/K$ for any $a \in \{1, \ldots, K\}$. We prove in \Cref{theorem:anytime} that $(E_t)_{t \geq -1}$ is an e-process under the null hypothesis $\rmH_0$. For any distribution $\pi$ on $\{1, \ldots, K\}$, we define the support of this distribution as
\begin{align*}
\supp(\pi) = \{a\in \{1, \ldots, K\} \; \text{ such that } \; \pi(a)>0\}\eqsp,
\end{align*}
the minimal value of this distribution as
\begin{align}\label{equation:definition_underline_pi}
\underline{\pi} = \min_{a \in\{1, \ldots, K\}} \pi(a) \eqsp.
\end{align}
Finally let $\cF_t = \sigma(A_1,\dots,A_t)$ and define for any $n \in \nset$
\begin{align*}
    S_{n}^{1} = \sum_{t=0}^{n-1}\Big(\log \widehat \pi_{t-1}(A_t) - \E_\pi[\log \widehat \pi_{t-1}(A_t)\mid\cF_{t-1}]\Big) \eqsp.
\end{align*}
\begin{lemma}\label{lemma:bound_pinelis}
There exist universal constants $A_1$ and $A_2$ that solely depend on $K, \underline{\pi}$ such that we have that for any $\delta >0, n \geq 0$
\begin{align*}
    \P_\pi\Big(|S_n^{1}|/n \geq \delta\Big) \leq
    A_1 \exp\left(-\frac{n \delta^2}{8 A_2 \log(n)} \right) \eqsp.
\end{align*}
\end{lemma}

\begin{proof}[Proof of \Cref{lemma:bound_pinelis}]
Note that $(S_{n}^{1})_{n \geq 0}$ is a martingale with respect to $(\cF_n)_{n \geq 0}$. By \citet[][Theorem 4.1]{pinelis1994optimum}, there exists a universal constant $C_1 >0$ such that we have for any $p \geq 2$
\begin{equation}\label{equation:ejrzlkjf}
\begin{aligned}
\E_\pi[|S_n^{1}|^p]^{1/p} & \leq C_1\Big(p \,  \E_\pi[|\sup_{0\leq t \leq n-1}\Delta_t|^p]^{1/p} + \sqrt{p} \, \E_\pi[\Big(\sum_{t=0}^{n-1} \E_\pi[\Delta_t^2|\cF_{t-1}]\Big)^{p/2}]^{1/p} \Big) \\
    & \leq C_1\Big(p \log(n) + \sqrt{p} \, \E_\pi \, \Big[ \Big(\sum_{t=0}^{n-1}\E_\pi[\Delta_t^2|\cF_{t-1}]\Big)^{p/2}\Big]^{1/p}\Big) \eqsp,
\end{aligned}
\end{equation}
where we denote our martingale increments $\Delta_t = \log \hat{\pi}_{t-1}(A_t) - \E_\pi[\log \hat{\pi}_{t-1}(A_t)\mid\cF_{t-1}]$. For any $t\geq 0, a \in \{1, \ldots, K\}$, we have $N_{t-1}(a)\in\{0,\dots,t-1\}$ almost surely, hence
\begin{align*}
\hat\pi_{t-1}(a)=\frac{N_{t-1}(a)+1}{(t-1)+K}\in\Big[\frac1{(t-1)+K},\frac{t}{(t-1)+K}\Big] \eqsp,
\end{align*}
and hence $\log(1/(K+t-1)\Big) \leq \log \hat\pi_{t-1}(A_t) \leq \log(t/(K+t-1))$ almost surely. The same also holds for $\E_\pi[\log \hat\pi_{t-1}(A_t)\mid \mathcal F_{t-1}]$. Consequently,
\begin{align*}
|\Delta_t|
&= |\log \hat\pi_{t-1}(A_t) - \E_\pi[\log \hat\pi_{t-1}(A_t)\mid \mathcal F_{t-1}]| \leq \log t \eqsp.
\end{align*}
We now consider $N_t(a)$ as a sum of $t$ Bernoulli random variables with parameter $\pi(a)$ under $\P_\pi$ and apply the Hoeffding inequality to obtain that for any $\delta_a >0$
\begin{align*}
 \P_\pi(\{|N_t(a)-t q(a)| > \delta_a \}) \leq 2\exp(-2 \delta_a^2/t) \eqsp,
\end{align*}
and hence
\begin{align*}
\P_\pi(|N_{t-1}(a)-(t-1)q(a)|> (t-1)\pi(a)/2 \text{ for any }a \in \supp(\pi))\leq 2K \exp(-(t-1) \underline{\pi}^2 /2) \eqsp.
\end{align*}
Therefore, with probability at least $1 - 2K \exp(- \underline{\pi}^2(t-1)/2)$, we have simultaneously that for any $a\in\{1, \ldots, K\}$, $N_{t-1}(a) \geq (t-1) \, \pi(a)/2$. On this event, for any $t \geq 0$, we have that
\begin{align*}
1 \geq \hat{\pi}_{t-1}(a) =\frac{N_{t-1}(a)+1}{t-1+K} \geq \frac{(t-1)\pi(a)/2}{t-1+K} \geq \pi(a)/4K \eqsp,
\end{align*}
and hence
\begin{align*}
|\log \hat{\pi}_{t-1}(a)| \leq \max\{0,|\log(\underline{\pi}/4K)|\,\} = |\log(\underline{\pi}/4K)| \eqsp.
\end{align*}
Consequently, still on this event, we have
\begin{align*}
|\Delta_t|^2 =|\log \hat{\pi}_{t-1}(A_t) - \E_\pi[\log \hat{\pi}_{t-1}(A_t)\mid\cF_{t-1}]|^2 \leq 4 |\log(\underline{\pi}/4K)|^2 \eqsp.
\end{align*}
On the opposite event, since $1 \leq N_{t-1}(a)+1 \leq t$ for any $a \in \supp(\pi)$, we have deterministically that $|\log \hat{\pi}_{t-1}(a)| \leq \log t$ and hence $|\Delta_t|^2 \leq 4 \log(t)^2$. Consequently, we have that
\begin{align*}
\E_\pi[|\Delta_t|^2] \leq 8 \log(t)^2 K e^{- \underline{\pi}^2 t/2} + 4 |\log(\underline{\pi}/4K)|^2 \eqsp.
\end{align*}
Since $x^2 \rme^{-c x} = (x \rme^{-cx/2})^2 \leq (1/\rme c)^2$ for any $x\geq 0, c>0$ and $\log t \leq t$ for any $t\geq 1$ and $\log(1)=0$, we have that $8 \log(t)^2 K \rme^{- \underline{\pi}^2 t/2} \leq 8K (2/\rme \underline{\pi}^2)^2 = 32K/\rme^2 \underline{\pi}^4$, and hence
\begin{align*}
    \E_\pi[|\Delta_t|^2] \leq 32K/\rme^2 \underline{\pi}^4 + 4 |\log(\underline{\pi}/4K)|^2 = K C_2^2/\underline{\pi}^4 \eqsp,
\end{align*}
where $C_2>0$ is a universal constant. Back to \Cref{equation:ejrzlkjf}, we have that for any $p \geq 2$
\begin{align}\label{equation:rfjzlef}
    \E_\pi[|S_n^1|^p]^{1/p} \leq C_1\left\{ p \log(n) + \sqrt{p} \Big(\Big(\sum_{t=0}^{n-1} K C_2^2/\underline{\pi}^4 \Big)^{p/2}\Big)^{1/p} \right\} = C_1\left\{p\log(n)+ \sqrt{K} C_2 \sqrt{pn}/\underline{\pi}^2 \right\} \eqsp.
\end{align}
Using the Markov inequality, we can write
\begin{align*}
    \P_\pi\Big(|S_n^1|/n \geq \delta\Big) = \P_\pi\Big(|S_n^1|^p \geq \delta^p n^p\Big) \leq \E_\pi[|S_n^1|^p]/\delta^p n^p \leq \frac{C_1^p\left\{p\log(n)+\sqrt{K} C_2 \sqrt{pn}/\underline{\pi}^2 \right\}^p}{\delta^p n^p} \eqsp,
\end{align*}
which gives using \Cref{equation:rfjzlef}
\begin{align*}
    \P_\pi\Big(|S_n^1|/n \geq \delta\Big) \leq \Big(\frac{C_1 \log(n)p}{\delta n} + \frac{\sqrt{K} C_2 \sqrt{p}}{\delta \underline{\pi}^2\sqrt{n}}\Big)^p \eqsp.
\end{align*}
We define $C = \max\{1, C_1, C_2^2 K/\underline{q}^4\}$ and have
\begin{align*}
\P_\pi\Big(|S_n^1|/n \geq \delta\Big) \leq \Big(\frac{C p \log(n)}{\delta n} + \frac{\sqrt{C} \sqrt{p}}{\delta \sqrt{n}} \Big)^p \eqsp.
\end{align*}
Now for any $n \geq 8C\log(n)/\delta^2$, choosing $p =\delta^2 n /(4C \log n)$ we have that $p \geq 2$ and we obtain
\begin{equation}\label{equation:condition1nlarge}
\begin{aligned}
    & C \log(n)p/\delta n = 1/4 \eqsp, \\
    & \sqrt{C} \sqrt{p}/\sqrt{n} \delta = \frac{\sqrt{C} \delta \sqrt{n}}{2 \delta \sqrt{n} \sqrt{C \log(n)}} \leq 1/4 \eqsp,
\end{aligned}
\end{equation}
which implies that
\begin{align}\label{equation:bound1}
\P_\pi(|S_n^1|/n \geq \delta) \leq (1/2)^p = \exp(-p \log 2 ) = \exp\Big(- \frac{\delta^2 n \log(2)}{4 C \log(n)} \Big) \leq \exp\Big(- \frac{\delta^2 n}{8 C \log(n)} \Big) \eqsp,
\end{align}
for any $n$ such that $n \geq 8C\log(n)/\delta^2$, which completes the proof.
\end{proof}

\begin{remark}
Note that a straightforward application of Azuma-Hoeffding's inequality with increments bounded by $\log(t)$ for any $t\geq 1$ gives a result comparable to that of \Cref{lemma:bound_pinelis} but is less tight. Since $0+\sum_{t=1}^{n-1} \log(t)^2 \leq n \log(n)^2$ where we take $0$ for the first term--with $\epsilon = 4\sqrt{n\log(n)^3}$, it states that for any $\delta>0$
\begin{align*}
    \P_\pi\Big(|S_n^{1}|/n \geq \delta\Big) = \P_\pi\Big(|S_n^{1}| \geq n \delta\Big) \leq \exp\left(-\frac{n^2 \delta^2}{8 n \log(n)^2} \right) = \exp\left(-\frac{n \delta^2}{8 \log(n)^2} \right) \eqsp.
\end{align*}
\end{remark}

\begin{lemma}\label{lemma:bound_KL}
Consider a distribution $\pi$ on $\{1, \ldots, K\}$ and the corresponding Laplace estimator $\hat{\pi}$ as defined in \Cref{equation:definition_appendix_hat_pi}. For any $n$ such that $n \geq 32 C_{\KL}/(\delta^2 \log(8 C_{\KL}/\delta))$, we have that
\begin{align*}
\P_\pi\left(\frac{1}{n} \sum_{t=0}^{n-1} \KL\big(\pi\Vert \hat{\pi}_{t-1} \big)\geq \delta
\right) \leq 4 n \exp\left(- \frac{n \delta}{8 C_{\KL} \log(n)^2} \right) \eqsp.
\end{align*}
\end{lemma}

\begin{proof}[Proof of \Cref{lemma:bound_KL}]
Our goal is to control $\sum_{t=0}^{n-1} \KL(\pi\|\hat{\pi}_{t-1})$ using \citet[Theorem 1]{mourtada2025estimation} that states that for any $t\geq 2$ and any $\delta\in(0,1/4)$, there exists a constant $C_{\KL} > 0$ such that for any $\eta \in (\exp(-t/6), \exp(-2))$
\begin{align}\label{equation:jaouad}
\P_\pi\bigg(
\KL\big(\pi \|\hat{\pi}_{t-1}\big)
> C_{\KL}\,\frac{K + \log(1/\eta)\,\log\log(1/\eta)}{t-1}
\bigg)
\leq 4\eta \eqsp.
\end{align}
Set $x = (t-1)\delta/2C_{\KL}$ and assume that $t \geq \max\{2C_{\KL}/(\delta \exp(-3\delta/C_{\KL})), 4 C_{\KL} K /\delta\}$ as well as $K\geq 10 >e^2$. It implies that $x \geq K = \max\{K,e^2\}$ or equivalently $(t-1)\delta/2C_{\KL} \geq \max\{K,e^2\} = K$, and that $t\geq 8C_{\KL}/\delta$. A consequence is that once we define $\eta = \exp(-x/\log x)$, we have that $\eta \in (\exp(-t/6), \exp(-2))$ and we can apply \Cref{equation:jaouad} with this $\eta$. Now, since $\log(x/\log(x))\leq \log(x)$ here, we have that
\begin{align*}
C_{\KL}\frac{K + \log(1/\eta)\,\log\log(1/\eta)}{t-1}
&= C_{\KL}\frac{K+ x/\log(x) \times \log(x/\log(x))}{t-1} \\
& \leq  C_{\KL}\frac{K+ x}{t-1} \\
& \leq  2 C_{\KL}\frac{x}{t-1} \\
& =  \frac{2 C_{\KL}}{t-1} \frac{(t-1)\delta}{2C_{\KL}} \\
& = \delta \eqsp,
\end{align*}
and thus, we obtain that
\begin{align*}
\P_\pi\big(\KL(\pi\|\hat{\pi}_{t-1})\geq \delta\big)
&\leq \P_\pi\left(\KL(\pi\|\hat{\pi}_{t-1}) > C_{\KL}\frac{K + \log(1/\eta)\,\log\log(1/\eta)}{t-1}\right)\leq 4\eta= 4\exp(-x/\log x) \eqsp.
\end{align*}
Therefore, for any $t \geq t_{\min} = \max\{2C_{\KL} \exp(3\delta/C_{\KL})/\delta, 4 C_{\KL} K /\delta\} = 4K C_{\KL}/\delta$ for $\delta \leq C_{\KL}/3$, we have that
\begin{align}\label{equation:ejrlzejrl}
\P_\pi\big(\KL(\pi\|\hat{\pi}_{t-1}) \geq \delta\big)
\leq 4\exp\left(-\frac{(t-1)\delta/(2C_{\KL})}{\log\big((t-1)\delta/(2C_{\KL})\big)}
\right) \eqsp,
\end{align}
and hence we define
\begin{equation}\label{equation:definition_t_min}
    t_{\min} = 4 K C_{\KL}/\delta \eqsp.
\end{equation}
By definition of the smoothed Laplace estimator, for any $t \in \N, \KL(\pi\|\hat{\pi}_{t-1})\leq \log(t)$, and thus
\begin{align*}
\left\{\frac{1}{n}\sum_{t=0}^{n-1} \KL(\pi\|\hat{\pi}_{t-1}) \geq \delta \right\} & \subseteq \left\{\frac{1}{n}\sum_{t=t_0}^{n-1} \KL(\pi\|\hat{\pi}_{t-1}) \geq \delta - t_0 \log(t_0)/n \right\} \\
& \subseteq \bigcup_{t=t_0}^{n-1} \left\{\KL(\pi\|\hat{\pi}_{t-1}) \geq \delta_t \right\} \eqsp,
\end{align*}
for any sequence $(\delta_t)_{t\in [t_0,n-1]} \in (0,1)^{n-t_0}$ such that $\sum_{t=t_0}^{n-1} \delta_t =n \delta - t_0 \log(t_0)$. Therefore, an union bound gives that
\begin{align*}
\P_\pi\left(\frac{1}{n}\sum_{t=0}^{n-1} \KL(\pi\|\hat{\pi}_{t-1}) \geq \delta
\right) \leq \sum_{t=t_0}^{n -1}\P_\pi\left(\KL(\pi\|\hat{\pi}_{t-1}) \geq \delta_{t-1}\right)
\eqsp.
\end{align*}
From the pointwise bound in \Cref{equation:ejrlzejrl}, for any $t_0 \in [t_{\min}, n]$, we obtain
\begin{align*}
    \P_\pi\left(\frac{1}{n} \sum_{t=0}^{n-1} \KL(\pi\|\hat{\pi}_{t-1})\geq \delta\right) \leq 4 \sum_{t=t_0}^{n-1} \exp\left(-\frac{(t-1) \delta_{t-1}}{4C_{\KL} \log((t-1) \delta_{t-1}/(4C_{\KL}))} \right) \eqsp.
\end{align*}
Define $B(t_0) = (n \delta-t_0\log(t_0))/4 C_{\KL} H(t_0,n)$ with $H(t_0,n) = \sum_{t=t_0}^{n-1} 1/t$ being the harmonic series. Now choose the sequence $(\delta_t)_{t\in [t_0+1,n]}$ such that
\begin{align*}
    t\delta_t/4 C_{\KL} = B(t_0) \; \text{ and hence }\; \sum_{t=t_0}^{n-1} \delta_t = 4 C_{\KL} B(t_0)/\sum_{t=t_0}^{n-1} 1/t = n \delta-t_0\log(t_0) \eqsp.
\end{align*}
i.e. $\delta_t = 4 C_{\KL} B(t_0)/t$. Therefore, we have that
\begin{align*}
    \exp\left(-\frac{(t-1) \delta_{t-1}}{4C_{\KL} \log((t-1) \delta_{t-1}/(4C_{\KL}))} \right) = \exp(- B(t_0) /\log(B(t_0)) \eqsp.
\end{align*}
Consequently, we have
\begin{align*}
    \P_\pi\left(\frac{1}{n} \sum_{t=0}^{n-1}\KL(\pi\|\hat{\pi}_{t-1})\geq \delta\right) \leq 4 (n-t_0) \exp(- B(t_0) /\log(B(t_0)) \leq n \exp(- B(t_0) /\log(B(t_0)) \eqsp.
\end{align*}
We now choose $t_0$ such that $t_0 \log(t_0) = n\delta/2$, that is
\begin{align*}
    \exp(\log(t_0)) \log(t_0) = n \delta/2 \; \text{ which gives } \log(t_0) = W(n\delta/2) \eqsp,
\end{align*}
with $W$ the Lambert function. Using the identity $\exp(W(z)) = z W(z)$, we have that
\begin{align*}
    t_0 = \frac{n \delta}{2 W(n\delta/2)} \eqsp,
\end{align*}
and for $n$ such that
\begin{align*}
n \geq \frac{16 C_{\KL}}{\delta^2 \log(8C_{\KL}/\delta)} \eqsp,
\end{align*}
it ensures that $t_0\geq t_{\min}$. We have $B(t_0) = (n \delta-t_0\log(t_0))/4 C_{\KL} H(t_0,n) =n\delta/(8C_{\KL} H(t_0,n))\geq n\delta/(8C_{\KL} \log(n))$, which gives that
\begin{align*}
    \P_\pi\left(\frac{1}{n} \sum_{t=0}^{n-1} \KL\big(\pi\| \hat{\pi}_{t-1} \big)\geq \delta\right) \leq 4 n \exp\left(- \frac{n\delta}{8C_{\KL} \log(n) \log(n\delta/(8C_{\KL} \log(n))} \right) \eqsp,
\end{align*}
for any $n \geq 16 C_{\KL}/(\delta^2 \log(8C_{\KL}/\delta))$. Consequently, since $\log(n\delta/(8C_{\KL} \log(n)) \leq \log(n)$, we have that
\begin{align*}
    \P_\pi\left(\frac{1}{n} \sum_{t=0}^{n-1} \KL\big(\pi\| \hat{\pi}_{t-1} \big)\geq \delta
    \right) \leq 4 n \exp\left(- \frac{n \delta}{8 C_{\KL} \log(n)^2} \right) \eqsp,
\end{align*}
hence the result.
\end{proof}

\begin{lemma}\label{lemma:bounding_large_term_with_KL}
    For any distributions $\pi_0, \pi$ over $\{1, \ldots, K\}$ such that $\supp(\pi)\subseteq \supp(\pi_0)$, define $S_n^2 = \sum_{t=0}^{n-1}(\E_\pi[\log \hat{\pi}_{t-1}(A_t)\mid\cF_{t-1}] - \log \pi_0(A_t))$. For any $\delta>0$ and $n \geq 3\delta/\KL(\pi\|\pi_0)$, we have that
    \begin{align*}
        \P_\pi(S_n^2 \leq \delta) \leq 4 n \exp\left(- \frac{n \KL(\pi\|\pi_0)}{32 C_{\KL} \log(n)^2} \right) + 2 \exp\left(-\frac{2 n \KL(\pi\|\pi_0)^2}{16 |\log \underline{\pi}_0|} \right) \eqsp.
    \end{align*}
\end{lemma}

\begin{proof}[Proof of \Cref{lemma:bounding_large_term_with_KL}]
    We decompose $S_n^2$ as
\begin{align*}
S_n^2 & = \sum_{t=0}^{n-1}(\E_\pi[\log \hat{\pi}_{t-1}(A_t)\mid\cF_{t-1}] - \E_{\pi}[\log \pi_0]) + \sum_{t=0}^{n-1}(\E_{\pi}[\log \pi_0] - \log \pi_0(A_t)) \\
& \nonumber = \sum_{t=0}^{n-1} (
\E_{\pi}[\log\hat{\pi}_{t-1}(A_t)\mid\cF_{t-1}] -\E_{\pi}[\log \pi]) + \sum_{t=0}^{n-1} (\E_{\pi}[\log \pi]- \E_{\pi}[\log \pi_0]) \\
& \quad + \sum_{t=0}^{n-1}(\E_{\pi}[\log \pi_0] - \log \pi_0(A_t)) 
) \\
& = \sum_{t=0}^{n-1} (
\E_{\pi}[\log\hat{\pi}_{t-1}(A_t)\mid\cF_{t-1}] -\E_{\pi}[\log \pi]) + n\,\KL(\pi\|\pi_0) + \sum_{t=0}^{n-1}(\E_{\pi}[\log \pi_0] - \log \pi_0(A_t)) \eqsp.
\end{align*}
For any $n \in \N^\star$ such that $\KL(\pi\|\pi_0) \geq 2 \delta/n$, we thus have that
\begin{align}
& \nonumber \P_\pi(S_n^2 \leq \delta) \\
& \quad = \P_{\pi}\left(\frac{1}{n} \sum_{t=0}^{n-1} (\E_{\pi}[\log\hat{\pi}_{t-1}(A_t)\mid\cF_{t-1}] -\E_{\pi}[\log \pi]) + \KL(\pi\|\pi_0) + \frac{1}{n} \sum_{t=0}^{n-1}(\E_{\pi}[\log \pi_0] - \log \pi_0(A_t)) \leq \delta / n \right) \\
& \nonumber \quad = \P_{\pi}\left(\frac{1}{n} \sum_{t=0}^{n-1} (
\E_{\pi}[\log\hat{\pi}_{t-1}(A_t)\mid\cF_{t-1}] -\E_{\pi}[\log \pi]) + \frac{1}{n} \sum_{t=0}^{n-1}(\E_{\pi}[\log \pi_0] - \log \pi_0(A_t)) \leq \delta/n - \KL(\pi\|\pi_0) \right) \\
& \nonumber \quad \leq \P_{\pi}\left(\Big|\frac{1}{n} \sum_{t=0}^{n-1} (
\E_{\pi}[\log\hat{\pi}_{t-1}(A_t)\mid\cF_{t-1}] -\E_{\pi}[\log \pi]) + \frac{1}{n} \sum_{t=0}^{n-1}(\E_{\pi}[\log \pi_0] - \log \pi_0(A_t)) \Big| \geq \KL(\pi\|\pi_0) - \delta /n \right) \\
& \nonumber \quad \leq \P_{\pi}\left(\Big| \frac{1}{n} \sum_{t=0}^{n-1} (
\E_{\pi}[\log\hat{\pi}_{t-1}(A_t)\mid\cF_{t-1}] -\E_{\pi}[\log \pi]) + \frac{1}{n} \sum_{t=0}^{n-1}(\E_{\pi}[\log \pi_0] - \log \pi_0(A_t)) \Big| \geq \KL(\pi\|\pi_0)/2\right) \\
& \label{equation:bound_iid_at_blabla_erel} \quad \leq \P_{\pi}\left(\Big|\frac{1}{n} \sum_{t=0}^{n-1} (
\E_{\pi}[\log\hat{\pi}_{t-1}(A_t)\mid\cF_{t-1}] -\E_{\pi}[\log \pi])\Big| \geq \KL(\pi\|\pi_0)/4\right) \\
& \label{equation:bound_iid_at_blabla}\quad \quad + \P_{\pi}\left(\Big|\frac{1}{n} \sum_{t=0}^{n-1}(\E_{\pi}[\log \pi_0] - \log \pi_0(A_t)) \Big| \geq \KL(\pi\|\pi_0)/4 \right)
\end{align}
\Cref{lemma:bound_KL} applied to $\pi$ gives an upper bound on \Cref{equation:bound_iid_at_blabla_erel}, following
\begin{align}
\P_{\pi}\left(\Big|\frac{1}{n}\sum_{t=0}^{n-1} (
\E_{\pi}[\log\hat{\pi}_{t-1}(A_t)\mid\cF_{t-1}] -\E_{\pi}[\log \pi])\Big| \geq \KL(\pi\|\pi_0)/4 \right) & \nonumber = \P_{\pi}\left(\frac{1}{n}\sum_{t=0}^{n-1} \KL(\hat{\pi}_t \Vert \pi) \geq \KL(\pi\|\pi_0)/4 \right)  \\
& \label{equation:ejkzelr} \leq 4 n \exp\left(- \frac{n \KL(\pi\|\pi_0)}{32 C_{\KL} \log(n)^2} \right) \eqsp.
\end{align}
$(A_t)_{t \leq n}$ in the right-hand side of \Cref{equation:bound_iid_at_blabla} are i.i.d. draws from $\pi$. Consequently, the right-hand is a sum of i.i.d. random variables. Under $\pi$ which satisfies $\supp(\pi)\subseteq \supp(\pi_0)$, we have $|\log \pi_0(A_t)|\leq \max\{|\log \underline \pi_0|,\,|\log 1|\} = |\log \underline{\pi_0}|$ and the summands $\E_{\pi}[\log \pi_0] - \log \pi_0(A_t)$ satisfy $|\E_{\pi}[\log \pi_0] - \log \pi_0(A_t)| \leq 2 |\log \underline \pi_0|$. Also note that $\E_\pi[\E_{\pi}[\log \pi_0] - \log \pi_0(A_t)] = 0$. Therefore, for any $\delta >0$, Hoeffding inequality applied to \Cref{equation:bound_iid_at_blabla} yields
\begin{align}\label{equation:efjqelrfkj}
\P_{\pi}\left(\Big|\sum_{t=0}^{n-1}(\E_{\pi}[\log \pi_0] - \log \pi_0(A_t))\Big| \geq n \KL(\pi\|\pi_0)/2\right) \leq 2 \exp\left(-\frac{2 n \KL(\pi\|\pi_0)^2}{16 |\log \underline{\pi}_0|} \right) \eqsp.
\end{align}
Consequently, grouping togetether the decomposition of $S_n^2$, \Cref{equation:ejkzelr} and \Cref{equation:efjqelrfkj}, we have that for any $n \geq 2\delta/\KL(\pi \Vert \pi_0)$
\begin{align*}
    \P_\pi(S_n^2 \leq \delta) \leq 4 n \exp\left(- \frac{n \KL(\pi\|\pi_0)}{32 C_{\KL} \log(n)^2} \right) + 2 \exp\left(-\frac{2 n \KL(\pi\|\pi_0)^2}{16 |\log \underline{\pi}_0|} \right) \eqsp.
\end{align*}
\end{proof}


Define the stopping time $\tau^\infty = \inf\{t \geq 1 \colon A_t \notin \supp(\pi_0)\}$. It is straightforward to observe that $\tau^\infty \geq \tau_{\gamma} +1$ since on the event $\{\tau^\infty <\infty\}$ we have, by construction of the plug–in e–process and since by \Cref{equation:definition_eprocess_appendix}
\begin{align*}
E_{\tau^\infty +1} = \prod_{s=0}^{\tau^\infty+1} \frac{\hat{\pi}_{s-1}(A_s)}{\pi_0(A_s)} = \frac{\hat{\pi}_{\tau^\infty}(A_{\tau^\infty})}{\pi_0(A_{\tau^\infty})}\prod_{s=0}^{\tau^\infty} \frac{\hat{\pi}_{s-1}(A_s)}{\pi_0(A_s)} = +\infty\eqsp,
\end{align*}
since $A_{\tau^\infty} \notin \supp(\pi_0)$, hence $\pi_0(A_{\tau^\infty}) = 0$ while $\hat{\pi}_{\tau^\infty}(A_{\tau^\infty}) >0$. Consequently, $S_{\tau^\infty+1} = +\infty \geq 1/\gamma$, which gives that $\tau^\infty \geq \tau_{\gamma} +1$. A second observation is that if $\supp(\pi) \subseteq \supp(\pi_0)$, then $\tau^{\infty} = +\infty$ almost surely because in that case, $(A_t)_{t\geq 0}$ almost never reach $\{1, \ldots, K\} \setminus \supp(\pi_0)$. Also, we see that $\tau^\infty$ follows a geometric distribution with parameter $\sum_{a \in \{1, \ldots, K\} \setminus \supp(\pi_0)} \pi(a)$, 
since $(A_t)_{t \geq 1}$ are i.i.d. and at each step $t$, fall in $\supp(\pi_0)$ with probability $\sum_{a \in \supp(\pi_0)} \pi(a)$ and fall in $\{1, \ldots, K\}\setminus \supp(\pi_0)$ with probability $\sum_{a \in \{1, \ldots, K\} \setminus \supp(\pi_0)} \pi(a)$. Therefore, in the case where $\supp(\pi)\not\subseteq \supp(\pi_0)$, we get
\begin{align*}
\E_\pi[\tau^\infty] = \frac{1}{\sum_{a \in \{1, \ldots, K\} \setminus \supp(\pi_0)} \pi(a)} < \infty \eqsp.
\end{align*}

\begin{theorem}\label{theorem:bigtheoremboundmartingale}
    Let $\pi_0$ be a discrete distribution on $\{1, \ldots, K\}$ and $(A_t)_{t \geq 0}$ be i.i.d. draws under $\pi$. We consider the null hypothesis $\rmH_0$ defined in \Cref{equation:definition_null_appendix}, as well as the e-process $(E_t)_{t \geq 0}$ defined in \Cref{equation:definition_eprocess_appendix}. For a Type I error $\gamma \in (0,1)$, there there exist universal constants $C_1, C_2>0$ such that for any $n$ satisfying $n \geq 9\delta/2\KL(\pi\|\pi_0)$
\begin{align*}
    \P_\pi(S_n/n \leq \delta) & \leq C_1 \pi(\supp(\pi_0))^{n+1} \left\{n \exp\left(- \frac{n \KL(\pi\|\pi_0)}{C_2 \log(n)^2} \right) + \exp\left(-\frac{ n \KL(\pi\|\pi_0)^2}{8 |\log \underline{\pi}_0|} \right) \right\} \\
    & \quad + C_1 \exp\left(-\frac{n \KL(\pi\Vert \pi_0)^2
    }{C_2 \log(n)} \right) \eqsp.  
\end{align*}
\end{theorem}

\begin{proof}[Proof of \Cref{theorem:bigtheoremboundmartingale}]
First note that the case where $\supp(\pi_0)\cap \supp(\pi) = \varnothing$ gives a stopping time equal to $1$ almost surely. Hence, we consider that this extreme case does not hold for the remaining of the proof. Recall that for any $t\geq 0, a \in \{1, \ldots, K\}$, we have
\begin{align*}
&\hat{\pi}_{t-1}(a)= \frac{N_{t-1}(a)+1}{(t-1)+K} \qquad \eqsp, \qquad
E_t = \prod_{s=1}^t \frac{\hat{\pi}_{s-1}\,\left(A_s\right)}{\pi_0\,\left(A_s\right)} \qquad \eqsp, \qquad \cF_t = \sigma(A_1,\dots,A_t) \eqsp,
\end{align*}
and $\hat{\pi}_{-1}(a)=1/K$. Define
\begin{align*}
& \tau_{\gamma} = \inf\{t\geq 0 \colon \,E_t\geq 1/\gamma\} \eqsp, \\
& X_t = \log \hat{\pi}_{t-1}(A_t) - \log \pi_0(A_t) \qquad , \qquad
S_t = \sum_{s=0}^{t-1} X_s = \log E_t \eqsp.
\end{align*}
Then, for any integer $n \geq 0$
\begin{align}
\P_\pi(\tau_{\gamma} \geq n) & \nonumber = \P_\pi(\text{for any } t < n, E_t\leq 1/\gamma) \\
& \nonumber = \P_\pi\Big(\max_{1\leq k\leq n} S_k < \log(1/\gamma)\Big) \\
& \label{equation:makroermkae} \leq \P_\pi\big(S_n < \log(1/\gamma)\big) \eqsp,
\end{align}
and we now decompose $S_n$ following
\begin{align*}
    S_n & = \sum_{t=0}^{n-1}(\log \hat{\pi}_{t-1}(A_t) - \log \pi_0(A_t)) \\
& = \underbrace{\sum_{t=0}^{n-1}(\log \hat{\pi}_{t-1}(A_t) - \E_\pi[\log \hat{\pi}_{t-1}(A_t)\mid\cF_{t-1}])}_{S_n^1} + \underbrace{\sum_{t=0}^{n-1} (\E_\pi[\log \hat{\pi}_{t-1}(A)\mid\cF_{t-1}] -\log \pi_0(A_t))}_{S_n^2} \eqsp.
\end{align*}
For any $n \in \N, \delta >0$, we have that
\begin{align*}
    \P_\pi(S_n \leq \delta) & = \P_\pi(S_n/n\leq \delta/n) \\
    & \leq \P_\pi(|S_n^1|/n \geq \delta/2 n, S_n/n\leq \delta/n) + \P_\pi(|S_n^1|/n \leq \delta/2 n, S_n/n\leq \delta/n) \\
    & \leq \P_\pi(|S_n^1|/n \geq \delta/2 n) + \P_\pi(S_n^2/n \leq 3\delta/2n) \eqsp.
\end{align*}
Observe that by definition of $\tau^\infty$, we have $R_{\tau^\infty+1}^2 = + \infty$, consequently for any $\delta >0$, we have
\begin{align*}
    \P_\pi(S^2_n \leq 3\delta/2) & = \P_\pi(\{\tau^\infty > n\} \cap \{S^2_n < 3\delta/2\}) \\
    &  = \P_\pi(\{\tau^\infty > n\}) \P_\pi(S^2_n < 3\delta/2 |\{\tau^\infty > n\}) 
    \eqsp.
\end{align*}
Conditionally on $\{\tau^\infty >n\}$, we have that $(A_t)_{t \leq n}$ follows an i.i.d. distribution $\tilde{\pi}$ which is a rescaling of $\pi$ on the support of $\pi_0$. Formally, following Bayes' theorem, we define $\tilde{\pi}$ for any $a \in \{1, \ldots, K\}$ as
\begin{align*}
    \tilde{\pi}(a) = \frac{\pi(a)\1\{a \in \supp(\pi_0)\}}{\sum_{a' \in \supp(\pi_0)} \pi(a')} \eqsp.
\end{align*}
Since $\tau^\infty$ follows a geometric distribution with parameter $\sum_{a \in \supp(\pi)/\supp(\pi_0)} \pi(a)$, we have
\begin{align*}
    \P_\pi(\tau^\infty >n) = (1 - \sum_{a \in \supp(\pi)/\supp(\pi_0)} \pi(a))^{n+1} = \pi(\supp(\pi_0))^{n+1} \eqsp.
\end{align*}
For any $F$ that is $\cF_n$-measurable, we have that
\begin{align}\label{equation:separation_pi_tilde}
    \P_\pi(F(A_1, \ldots, A_n) \mid \tau^\infty >n) = \P_{\tilde{\pi}}(F(A_1, \ldots, A_n)) \eqsp,
\end{align}
and therefore, we can write
\begin{align*}
    \P_\pi(S_n^2/n \leq 3\delta/2n) \leq \P_\pi(\tau^\infty > n)\P_{\tilde{\pi}}(S_n^2/n \leq 3\delta/2n) \eqsp.
\end{align*}
Since $\supp(\tilde{\pi}) \subseteq \supp(\pi)$, we can use \Cref{lemma:bounding_large_term_with_KL} to control the right-hand side and obtain that for any $n$ such that $n \geq 9\delta/2\KL(\pi\|\pi_0)$, we have
\begin{align*}
    \P_\pi(S_n^2/n \leq 3\delta/2n) \leq \P_\pi(\tau^\infty > n)\left\{ 4 n \exp\left(- \frac{n \KL(\pi\|\pi_0)}{32 C_{\KL} \log(n)^2} \right) + 2 \exp\left(-\frac{2 n \KL(\pi\|\pi_0)^2}{16 |\log \underline{\pi}_0|} \right) \right\} \eqsp.
\end{align*}
\Cref{lemma:bound_pinelis} ensures that there exist $A_1, A_2 >0$ such that
\begin{align}\label{equation:ekrjlzk}
    \P_\pi \Big(|S_n^{1}|/n \geq \delta\Big) \leq
    A_1 \exp\left(-\frac{n \delta^2}{32 A_2 \log(n)} \right) \eqsp,
\end{align}
and consequently for any $n \geq 9\delta/2\KL(\pi\|\pi_0)$, we have
\begin{align*}
    \P_\pi \Big(|S_n^{1}|/n \geq \delta/n\Big) & \leq \P_\pi \Big(|S_n^{1}|/n \geq \frac{2 n \KL(\pi\Vert \pi_0)}{9} \frac{1}{n} \Big) \\
    & = \P_\pi \Big(|S_n^{1}|/n \geq \frac{2 \KL(\pi\Vert \pi_0)}{9} \Big) \\
    & \leq A_1 \exp\left(-\frac{n \KL(\pi\Vert \pi_0)^2
    }{648 A_2 \log(n)} \right) \eqsp,
\end{align*}
and consequently, regrouping the terms, we obtain
\begin{align*}
    \P_\pi(S_n/n \leq \delta) & \leq \pi(\supp(\pi_0))^{n+1} \left\{ 4 n \exp\left(- \frac{n \KL(\pi\|\pi_0)}{32 C_{\KL} \log(n)^2} \right) + 2 \exp\left(-\frac{n \KL(\pi\|\pi_0)^2}{8|\log \underline{\pi}_0|} \right) \right\} \\
    & \quad + A_1 \exp\left(-\frac{n \KL(\pi\Vert \pi_0)^2
    }{648 A_2 \log(n)} \right) \eqsp.
\end{align*}
Consequently, there exist universal constants $C_1, C_2 >0$ such that for any $n$ such that $n \geq 9\delta/2\KL(\pi\|\pi_0)$
\begin{align*}
    \P_\pi(S_n/n \leq \delta) & \leq C_1 \pi(\supp(\pi_0))^{n+1} \left\{n \exp\left(- \frac{n \KL(\pi\|\pi_0)}{C_2 \log(n)^2} \right) + \exp\left(-\frac{ n \KL(\pi\|\pi_0)^2}{8 |\log \underline{\pi}_0|} \right) \right\} \\
    & \quad + C_1 \exp\left(-\frac{n \KL(\pi\Vert \pi_0)^2
    }{C_2 \log(n)} \right) \eqsp.  
\end{align*}

\end{proof}

\begin{corollary}\label{corollary:concentration_with_KL}
    Considering the stopping time $\tau_{\gamma}$ for any $\gamma\in (0,1)$, there exist $C_1, C_2>0$ such that for any $n \in \N$
    \begin{align*}
   \P_\pi(\tau_{\gamma} \geq n) & \leq C_1\Biggl\{ n\exp\left(-n \left[\frac{2 \TV(\pi\|\pi_0)^2}{C_2 \log(n)^2}+|\log \pi(\supp(\pi_0)) |\right]\right) + \exp\left(-n\frac{4 \TV(\pi\Vert \pi_0)^4
    }{C_2 \log(n)} \right) \\
    & \quad  + \exp\left(-n \left[\frac{ \TV(\pi\|\pi_0)^4}{2 |\log \underline{\pi}_0|} +|\log \pi(\supp(\pi_0))|\right]\right) \Biggr\}+\1\{n\leq 5\log(1/\gamma)/\KL(\pi\Vert \pi_0)\} \eqsp.
\end{align*}
\end{corollary}

\begin{proof}[Proof of \Cref{corollary:concentration_with_KL}]
    Starting from \Cref{theorem:bigtheoremboundmartingale}, we can use Pinsker's inequality to we obtain
\begin{align*}
     \P_\pi(S_n/n \leq \delta) & \leq C_1 \pi(\supp(\pi_0))^{n+1} \left\{n \exp\left(- \frac{2 n \TV(\pi\|\pi_0)^2}{C_2 \log(n)^2} \right) + \exp\left(-\frac{4 n \TV(\pi\|\pi_0)^4}{8 |\log \underline{\pi}_0|} \right) \right\} \\
    & \quad + C_1 \exp\left(-\frac{4 n \TV(\pi\Vert \pi_0)^4
    }{C_2 \log(n)} \right) \\
    & \leq C_1 \exp(n \log(\pi(\supp(\pi_0)))) \left\{n \exp\left(- \frac{2 n \TV(\pi\|\pi_0)^2}{C_2 \log(n)^2} \right) + \exp\left(-\frac{4 n \TV(\pi\|\pi_0)^4}{8 |\log \underline{\pi}_0|} \right)   \right\} \\
    & \quad + C_1 \exp\left(-\frac{4 n \TV(\pi\Vert \pi_0)^4
    }{C_2 \log(n)} \right)
    \eqsp,
\end{align*}
which becomes, since $\log(\pi(\supp(\pi_0))) < 0$
\begin{align*}
    \P_\pi(S_n/n \leq \delta) & \leq C_1\Biggl\{ n\exp\left(-n \left[\frac{2 \TV(\pi\|\pi_0)^2}{C_2 \log(n)^2}+|\log \pi(\supp(\pi_0)) |\right]\right) + \exp\left(-n\frac{4 \TV(\pi\Vert \pi_0)^4
    }{C_2 \log(n)} \right) \\
    & \quad  + \exp\left(-n \left[\frac{ \TV(\pi\|\pi_0)^4}{2 |\log \underline{\pi}_0|} +|\log \pi(\supp(\pi_0))|\right]\right) \Biggr\} \eqsp.
\end{align*}
Therefore, for any $n$ such that $n \geq 5\log(1/\gamma)/\KL(\pi\Vert \pi_0)$, we can use \Cref{equation:makroermkae} to obtain
\begin{align*}
    \P_\pi(\tau_{\gamma} \geq n) & \leq \P_\pi(S_n < \log(1/\gamma)) \\
    & \leq C_1\Biggl\{ n\exp\left(-n \left[\frac{2 \TV(\pi\|\pi_0)^2}{C_2 \log(n)^2}+|\log \pi(\supp(\pi_0)) |\right]\right) + \exp\left(-n\frac{4 \TV(\pi\Vert \pi_0)^4
    }{C_2 \log(n)} \right) \\
    & \quad  + \exp\left(-n \left[\frac{ \TV(\pi\|\pi_0)^4}{2 |\log \underline{\pi}_0|} +|\log \pi(\supp(\pi_0))|\right]\right) \Biggr\} \eqsp,
\end{align*}
and we can conclude that there exist universal constant $C_1, C_2>0$ such that
\begin{align*}
\P_\pi(\tau_{\gamma} \geq n) & \leq C_1\Biggl\{ n\exp\left(-n \left[\frac{2 \TV(\pi\|\pi_0)^2}{C_2 \log(n)^2}+|\log \pi(\supp(\pi_0)) |\right]\right) + \exp\left(-n\frac{4 \TV(\pi\Vert \pi_0)^4
    }{C_2 \log(n)} \right) \\
    & \quad  + \exp\left(-n \left[\frac{ \TV(\pi\|\pi_0)^4}{2 |\log \underline{\pi}_0|} +|\log \pi(\supp(\pi_0))|\right]\right) \Biggr\}+\1\{n\leq 5\log(1/\gamma)/\KL(\pi\Vert \pi_0)\} \eqsp,
\end{align*}
hence the result.
\end{proof}

\probaboundanaytime*

\begin{proof}[Proof of \Cref{theorem:probaboundanaytime}]
    Let $\epsilon >0$ and $\bs \in \mathfrak{S}(\bw_{\bv}, \epsilon)$, by definition, there exists $i \in [N]$ playing a strategy $\pi^i$ such that $\|\tilde{w}^i - w^i_{\bv}\|_1 \geq 2 \epsilon$. Consider such a player $i$ as well as the associated process $(E_t^i)_{t\geq 0}$. By definition, the actions $(A_t^i)_{t\geq 0}\sim \tilde{w}^i$ drawn by player $i$ are i.i.d., which ensures by \Cref{theorem:anytime} that $(E_t^i)_{t\geq 0}$ is an e-process and matches the expression given in \Cref{equation:definition_eprocess_appendix}. Therefore, by \Cref{corollary:concentration_with_KL}, we have that the expected stopping time $\ttau_{\bpsi^i}$ satisfies for any $n \geq 2$
    \begin{align*}
\P_\pi(\ttau_{\bpsi^i} \geq n) & \leq C_1\Biggl\{ n\exp\left(-n \left[\frac{2 \TV(\tilde{w}^i\|w_{\bv})^2}{C_2 \log(n)^2}+|\log \tilde{w}^i(\supp(w_{\bv})) |\right]\right) + \exp\left(-n\frac{4 \TV(\tilde{w}^i\Vert w_{\bv})^4
    }{C_2 \log(n)} \right) \\
    & \quad  + \exp\left(-n \left[\frac{ \TV(\tilde{w}^i\|\w_{\bv})^4}{2 |\log \underline{w_{\bv}}|} +|\log \tilde{w}^i(\supp(w_{\bv}))|\right]\right) \Biggr\}+\1\{n\leq 5\log(1/\gamma)/\KL(\tilde{w}^i\Vert w_{\bv})\} \eqsp.
\end{align*}
Using Pinsker inequality to bound the $\KL$ in the rightmost term as well as the fact that $\TV(\tilde{w}^i - w^i_{\bv}) = 2 \|\tilde{w}^i - w^i_{\bv}\|_1 \geq \epsilon$, it gives the expression we ned for $(\zeta^{\epsilon}_t)_{t \geq 0}$. The fact that $\zeta^{\epsilon}$ upper bounds a probability gives the result for $t=0$ and $t=1$.
\end{proof}

\begin{restatable}{proposition}{propositionexpectedstoppingtime}\label{proposition:expected_stopping_time}
Considering the stopping time $\tau_{\gamma}$ for any $\gamma\in (0,1), \eta \in (0,1)$, there exists a universal constant $C_\eta>0$ such that
\begin{align*}
\E_\pi[\tau_{\gamma}] & \leq \frac{5\log(1/\gamma)}{\KL(\pi\Vert \pi_0)} + C_\eta \frac{1+ |\log \underline{\pi}_0|^{1/(1-\eta)}}{\TV(\pi\Vert \pi_0)^{4/(1-\eta)}} \eqsp.
\end{align*}
\end{restatable}

\begin{proof}[Proof of \Cref{proposition:expected_stopping_time}]
For any $\gamma \in (0,1)$, we have that
\begin{align*}
    \E_\pi[\tau_{\gamma}] = \sum_{n \geq 0} \P_\pi(\tau_{\gamma}\geq n) \eqsp,
\end{align*}
and recall from \Cref{corollary:concentration_with_KL} that
\begin{align*}
   \P_\pi(\tau_{\gamma} \geq n) & \leq C_1\Biggl\{ n\exp\left(-n \left[\frac{2 \TV(\pi\|\pi_0)^2}{C_2 \log(n)^2}+|\log \pi(\supp(\pi_0)) |\right]\right) + \exp\left(-n\frac{4 \TV(\pi\Vert \pi_0)^4
    }{C_2 \log(n)} \right) \\
    & \quad  + \exp\left(-n \left[\frac{ \TV(\pi\|\pi_0)^4}{2 |\log \underline{\pi}_0|} +|\log \pi(\supp(\pi_0))|\right]\right) \Biggr\}+\1\{n\geq 5\log(1/\gamma)/\KL(\pi\Vert \pi_0)\} \eqsp.
\end{align*}
Consequently, we have
\begin{align*}
    \E_\pi[\tau_{\gamma}] & \leq 5\log(1/\gamma)/\KL(\pi\Vert \pi_0) + \sum_{n \geq 0} C_1\Biggl\{ n\exp\left(-n \left[\frac{2 \TV(\pi\|\pi_0)^2}{C_2 \log(n)^2}+|\log \pi(\supp(\pi_0)) |\right]\right) \\
    & \quad + \exp\left(-n\frac{4 \TV(\pi\Vert \pi_0)^4
    }{C_2 \log(n)} \right) + \exp\left(-n \left[\frac{ \TV(\pi\|\pi_0)^4}{2 |\log \underline{\pi}_0|} +|\log \pi(\supp(\pi_0))|\right]\right) \Biggr\} \eqsp,
\end{align*}
where the first term is due to the indicator from the upper bound in \Cref{corollary:concentration_with_KL}. We can write
\begin{align*}
    \E_\pi[\tau_{\gamma}] & \leq 5\log(1/\gamma)/\KL(\pi\Vert \pi_0) + C_1 \sum_{n \geq 0} \exp\left(-n\frac{4 \TV(\pi\Vert \pi_0)^4
    }{C_2 \log(n)} \right) \\
    & \quad 
    + C_1 \min \left\{ \sum_{n \geq 0} \Biggl\{ n\exp\left(-n \frac{2 \TV(\pi\|\pi_0)^2}{C_2 \log(n)^2}\right) + n \exp\left(-n \frac{ \TV(\pi\|\pi_0)^4}{2 |\log \underline{\pi}_0|} \right), 2 \exp(-n|\log \pi(\supp(\pi_0)) |) \right\} \eqsp.
\end{align*}
For any $\eta \in (0,1)$, there exists $N_\eta >0$ such that for any $n \geq N_\eta$, we have
\begin{align*}
    & \exp\left(-n\frac{4 \TV(\pi\Vert \pi_0)^4
    }{C_2 \log(n)} \right) \leq \exp\left(-n^{1-\eta}\frac{4 \TV(\pi\Vert \pi_0)^4
    }{C_2} \right)  \eqsp, \\
    & n\exp\left(-n \frac{2 \TV(\pi\|\pi_0)^2}{C_2 \log(n)^2}\right) \leq \exp\left(-n^{1-\eta} \frac{2 \TV(\pi\|\pi_0)^2}{C_2}\right) \eqsp, \\
    & n \exp\left(-n \frac{ \TV(\pi\|\pi_0)^4}{2 |\log \underline{\pi}_0|} \right) \leq \exp\left(-n^{1-\eta} \frac{ \TV(\pi\|\pi_0)^4}{2 |\log \underline{\pi}_0|} \right)  \eqsp,
\end{align*}
and we also have the existence of a constant $C_\eta>0$ such that for any $n \in \N$--since we can absorb the constant $C_1$ into $C_\eta$
\begin{align*}
    & \exp\left(-n\frac{4 \TV(\pi\Vert \pi_0)^4
    }{C_2 \log(n)} \right) \leq C_\eta \exp\left(-n^{1-\eta}\frac{4 \TV(\pi\Vert \pi_0)^4
    }{C_2} \right)  \eqsp, \\
    & n\exp\left(-n \frac{2 \TV(\pi\|\pi_0)^2}{C_2 \log(n)^2}\right) \leq C_\eta \exp\left(-n^{1-\eta} \frac{2 \TV(\pi\|\pi_0)^2}{C_2}\right) \eqsp, \\
    & n \exp\left(-n \frac{ \TV(\pi\|\pi_0)^4}{2 |\log \underline{\pi}_0|} \right) \leq C_\eta \exp\left(-n^{1-\eta} \frac{ \TV(\pi\|\pi_0)^4}{2 |\log \underline{\pi}_0|} \right)  \eqsp.
\end{align*}
For any $c>0$, we have that
\begin{align*}
    \sum_{n \geq 0}\rme^{-cn^{1-\eta}} \leq 1+\int_{x\geq 0} \rme^{-cx^{1-\eta}} \rmd x = 1 + \frac{1}{1-\eta}c^{-1/(1-\eta)} \Gamma(1/(1-\eta)) \eqsp,
\end{align*}
where $\Gamma$ stands for the gamma function. We can thus apply this inequality to each of the terms in the sum and obtain
\begin{align*}
    & \E_\pi[\tau_{\gamma}] \leq 5\log(1/\gamma)/\KL(\pi\Vert \pi_0) + C_\eta \sum_{n \geq 0} \exp\left(-n^{1-\eta}\frac{4 \TV(\pi\Vert \pi_0)^4
    }{C_2} \right) \\
    & \quad 
    + C_\eta \min \left\{ \sum_{n \geq 0} \Biggl\{ \exp\left(-n^{1-\eta} \frac{2 \TV(\pi\|\pi_0)^2}{C_2}\right) + \exp\left(-n^{1-\eta} \frac{ \TV(\pi\|\pi_0)^4}{2 |\log \underline{\pi}_0|} \right) , 2 \exp(-n|\log \pi(\supp(\pi_0)) |) \right\} \\
    & \quad \leq 3+ 5\log(1/\gamma)/\KL(\pi\Vert \pi_0) + \frac{C_\eta \Gamma(1/(1-\eta))}{1-\eta} \left\{\frac{4 \TV(\pi\Vert \pi_0)^4
    }{C_2}\right\}^{-1/(1-\eta)} \\
    & \quad + \min \left\{\frac{C_\eta \Gamma(1/(1-\eta))}{1-\eta}\left(\left(\frac{2 \TV(\pi\|\pi_0)^2}{C_2}\right)^{-1/(1-\eta)} + \left(\frac{ \TV(\pi\|\pi_0)^4}{2 |\log \underline{\pi}_0|} \right)^{-1/(1-\eta)}\right), \frac{2}{1-\rme^{-|\log \pi(\supp(\pi_0)) |}}\right\} \\
    & \quad \leq 3+ 5\log(1/\gamma)/\KL(\pi\Vert \pi_0) + \frac{C_\eta \Gamma(1/(1-\eta)) C_2^{1/(1-\eta)}}{4^{1/(1-\eta)}(1-\eta)} \frac{1
    }{\TV(\pi\Vert \pi_0)^{4/(1-\eta)}} \\
    & \quad + \min \left\{\frac{C_\eta \Gamma(1/(1-\eta))}{1-\eta}\left(\frac{C_2^{1/(1-\eta)}}{2^{1/(1-\eta)}} \frac{1}{\TV(\pi\|\pi_0)^{2/(1-\eta)}} + \frac{2^{1/(1-\eta)}|\log \underline{\pi}_0|^{1/(1-\eta)}}{\TV(\pi\Vert \pi_0)^{4/(1-\eta)}} \right), \frac{2}{1-\rme^{-|\log \pi(\supp(\pi_0)) |}}\right\} \eqsp.
\end{align*}
We now group all the terms depending on $\eta$ and universal constants independent of the problem together. Also, since we use the convention $\TV(p\Vert q) = 1/2 \int |p(x)-q(x)|dx$, we have that $\TV(\pi\Vert \pi_0)\leq 1$, hence $1/\TV(\pi\Vert \pi_0)^2\leq 1/\TV(\pi\Vert \pi_0)^4$. We also have that $\pi(\supp(\pi_0))\leq 1$, hence $-|\log \pi(\supp(\pi_0))| = \log \pi(\supp(\pi_0))$, which simplifies the right-hand side term. Therefore, we have the existence of a universal constant $C_\eta>0$ such that (we can absorb also the first term $3$ and the multipliers independent of the problem in $C_\eta$ since $\TV(\pi\Vert \pi_0)\leq 1$) we have
\begin{align*}
\E_\pi[\tau_{\gamma}] & \leq \frac{5\log(1/\gamma)}{\KL(\pi\Vert \pi_0)} + \frac{C_\eta}{\TV(\pi\Vert \pi_0)^{4/(1-\eta)}} + C_\eta \min\left\{\frac{1+ |\log \underline{\pi}_0|^{1/(1-\eta)}}{\TV(\pi\Vert \pi_0)^{4/(1-\eta)}}, \frac{1}{1-\pi(\supp(\pi_0))} \right\}
\end{align*}
We can finally simplify the minimum since one of the terms in the expression is already added before. Consequently, we finally have that for any $\eta \in (0,1)$ the existence of a universal constant $C_\eta$ such that
\begin{align*}
\E_\pi[\tau_{\gamma}] & \leq \frac{5\log(1/\gamma)}{\KL(\pi\Vert \pi_0)} + C_\eta \frac{1+ |\log \underline{\pi}_0|^{1/(1-\eta)}}{\TV(\pi\Vert \pi_0)^{4/(1-\eta)}} \eqsp.
\end{align*}
\end{proof}

\corogamestoppingtime*

\begin{proof}[Proof of \Cref{corollary:corogamestoppingtime}]
  Stopping time with \Cref{proposition:anytime_low_level}. 
  Let $\epsilon >0, s \in \mathfrak{S}(\bw_{\bv}, \epsilon)$. As we show in the proof of \Cref{theorem:probaboundanaytime}, we have the existence of a player $i \in [N]$ with a strategy $\pi^i$ such that $\TV(\tilde{w}^i - w^i_{\bv}) = 2 \|\tilde{w}^i - w^i_{\bv}\|_1 \geq \epsilon$. The e-process $(E_t^i)_{t\geq 0}$ satisfying all the required properties, we can apply \Cref{proposition:expected_stopping_time} and obtain that for any $\eta \in (0,1)$, there exists a universal constant $C_\eta>0$ such that
  \begin{align*}
\E^{\bs}\pi[\tau^{(\epsilon)}] & \leq \frac{5\log(1/\gamma)}{\KL(\tilde{w}^i\Vert w^i_{\bv})} + C_\eta \frac{1+ |\log \underline{w^i_{\bv}}|^{1/(1-\eta)}}{\TV(\pi\Vert w^i_{\bv})^{4/(1-\eta)}} \leq \frac{10 \log(1/\gamma)}{\TV(\tilde{w}^i\Vert w^i_{\bv})^2} + C_\eta \frac{1+ |\log \underline{w^i_{\bv}}|^{1/(1-\eta)}}{\TV(\pi\Vert w^i_{\bv})^{4/(1-\eta)}} \eqsp.
\end{align*}
Using the fact that $\TV(\pi\Vert w^i_{\bv})\geq \epsilon$ and choosing $\eta = 1/5$, we obtain the result.
\end{proof}

\paragraph{A discussion on the optimality of our test.} For our test $\phi_T^i = \1\left\{\,E_T^i \geq 1 / \gamma^{(1),i}\right\}$ based on samples $A_1^{i},\dots,A_T^{i}$, we can define
\begin{align*}
\gamma^{(1),i}_T = \P_{w_{\bv}}(\phi^i_T=1) ,\qquad
\gamma^{(2),i}_T = \P_{q}(\phi_T^i = 0) \eqsp,
\end{align*}
where $q \ne w_{\bv}$. $\gamma^{(2),i}_T$ corresponds to the Type-II error probability (accept $\rmH_0^i$ under the alternative $q$). Note that an asymptotic treatment of the expected stopping time upper bound from \Cref{theorem:anytime} gives that under $\rmH_1^i$, $\log(E_T^i)/T \xrightarrow[T\to\infty]{a.s.} \mathrm{KL}\,\big(q\|w_{\bv}\big)$. Hence for every $\epsilon>0$ there exists $T_\epsilon$ such that for all $T\geq T_\epsilon$
$$
\P\,\Big(\log E_T^i \ge T\big(\mathrm{KL}(q\|w_{\bv})-\epsilon\big)\Big)\ \underset{T \to \infty}{\longrightarrow} 1 \eqsp.
$$
Choose $T$ large enough so that $\log(1/\gamma)\le T\big(\mathrm{KL}(q\|w_{\bv})-\epsilon\big)$, we also have
\begin{align*}
\beta_T =  \P(\phi_T^i=0)
= \P\Big(\log E_T^i < \log(1/\gamma)\Big) \leq \P\,\Big(\log E_T^i < T\big(\mathrm{KL}-\epsilon\big)\Big)\ \xrightarrow[T\to\infty]{}\ 0 \eqsp,
\end{align*}
which implies that
\begin{align*}
\liminf_{T\to\infty} -\frac{\log \beta_T}{T} \geq \KL\big(q\|w_{\bv}\big)-\epsilon \eqsp.
\end{align*}
Since $\epsilon>0$ is arbitrary, we have
\begin{align*}
\liminf_{T\to\infty} -\frac{\log \beta_T}{T} \geq \KL \big(q\|w_{\bv}\big) \eqsp.
\end{align*}
By Chernoff-Stein lemma (see \Cref{lemma:chernoff_stein}) that controls the decay of the Type-II error probability, no level-$\gamma$ test can achieve a larger exponent, hence
\begin{align*}
\lim_{T\to\infty} -\frac{1}{T}\log \beta_T = \mathrm{KL}\,\big(q\|w_{\bv}\big),
\end{align*}
showing that the e-process threshold test is asymptotically optimal.

\subsection{Proofs in the batch setting}

\folkbatchtesting*

\begin{proof}[Proof of \Cref{theorem:batch_folk}]We start by proving point (2). Let $t \geq 0$. In what follows, we write $k_t = \lfloor t / L \rfloor$, so $k_t L \leq t < (k_t + 1)L$. By \Cref{eq:batchtypeierror}, there exists $\mathsf{H}_t\subset\cH_t$, with $\P^{\bbs_{\bv}}(\hist_t \in\mathsf{H})\geq 1 - q_L$ such that for any $\P^{\bbs_{\bv}}(\kappa_{\bphi}\geq k+1 \,|\,\kappa_{\bphi}\geq k\,,\,\hist_t)\geq 1-p_L$. For the rest of the proof, we pick a $\hist_t \in\mathsf{H}_t$.

First, assume $\mathsf{H}_t \subset\{\kappa_{\bphi}<k_t\}$, so we already are in the punishment phase. By definition of $\bbs_{\bv}$ in \eqref{def:grim}, players only play $\boldsymbol{b}=( b^1,\ldots, b^N)$ from time $t$ on. On the one hand, sticking to $\bars^i_{\bv}$ for player $i$ yields
\begin{align*}
    U^i (\bbs_{\bv}\,;\,\hist_t)=(1-\beta)\,\sum_{\ell=t}^\infty \beta^{\ell-t}u^i (b^i , b^{-i})=u^i (b^i , b^{-i})=\underline{u}^i\eqsp.
\end{align*}
On the other hand, deviating to some $s\in\cS$ yields a payoff
    \begin{align*}
        U^i (s, \bars^{-i}_{\bv}\,;\,\hist_t)&=(1-\beta)\,\E^{(s,\bars^{-i}_{\bv})}\parentheseDeux{\sum_{\ell=t}^\infty \beta^{\ell-t}u^i (w^i_{\ell}, b^{-i})\,|\,\hist_t}\\
        &\leq (1-\beta)\,\sum_{\ell=t}^\infty \beta^{\ell-t}u^i (b ^i, b^{-i}) \\
        & =u ^i (b^i , b^{-i}) \\
        & =\underline{u}^i \\
        & = U^i (\bbs_{\bv}\,;\, \hist_{t}) \eqsp.
    \end{align*}
    The last inequality stems from the fact that $(b^i , b^{-i})$ is Nash equilibrium. Therefore, the HP-SPNE condition is satisfied.
   
    Now, assume $\mathsf{H}_t\subset \{k_t \leq \kappa_{\bphi}\}$ so players have not entered the punishment phase yet. Our proof proceeds in three steps: (i) we lower bound $U^i (\bbs_{\bv}\,;\,\hist_{t})$ by a term that depends on $p_L$, (ii) upper bound $U^i (s,\bars^{-i}_{\bv}\,;\,\hist_{t})$ by another term that depends on $\Delta_L$, and (iii) show that under the conditions of the theorem, the lower bound is greater than upper bound, up to an approximation term.

 We first lower bound $U^i (\bbs_{\bv}\,;\,\hist_{t})$. We have
\begin{align}
    U^i(\bbs_{\bv}\,;\,\hist_{t}) & \nonumber = (1-\beta)\,\E^{\bbs_{\bv}}\parentheseDeux{\2{\kappa_{\bphi}=k_{t}}\sum_{\ell=t}^\infty \beta^{\ell-t}u^i (w^i_{\ell}, w^{-i}_{\ell})+\2{\kappa_{\bphi}>k_{t}}\sum_{\ell=t}^\infty \beta^{\ell-t}u^i (w^i_{\ell}, w^{-i}_{\ell})\,|\,\hist_t}\nonumber\\
    &\geq (1-\beta)\,\E^{\bbs_{\bv}}\parentheseDeux{\2{\kappa_{\bphi}>k_{t}}\sum_{\ell=t}^\infty \beta^{\ell-t}u^i (w^i_{\ell}, w^{-i}_{\ell})\,|\,\hist_t}\nonumber \eqsp.
    \end{align}
Since $t < (k_{t}+1) L$ and $\beta^{t-\ell}\geq\beta^{t- k_t L}$ (because $t \geq k_t L$) and payoffs are positive, we have
\begin{align}
U^i(\bbs_{\bv}\,;\,\hist_t)&\geq (1-\beta)\E^{\bbs_{\bv}}\parentheseDeux{\2{\kappa_{\bphi}>k_{t}}\sum_{\ell=(k_{t} + 1)L}^{\kappa_{\bphi} L -1 }\beta^{\ell-k_{t} L}u^i (w^i_{\ell},w^{-i}_{\ell})\,|\,\hist_t}\label{eq:lastlinealign} \eqsp.
\end{align}
By \Cref{lemma:barkappa}, we may enlarge the probability space and construct 
$\bar\kappa$ from an i.i.d. sequence $(U_k)_{k\geq k_t}$ independent of the play induced by $\bbs_{\bv}$.
We denote expectation on this product space by $\widetilde{\E}^{\bbs_{\bv}}$. By \Cref{lemma:barkappa}, $\bar{\kappa}\leq\kappa_{\bphi}$ almost surely (conditional on $\hist_t$), and in particular $\2{\bar{\kappa}>k_t}\leq\2{\kappa_{\bphi}>k_t}$. Consequently, we have
\begin{align}
U^i(\bbs_{\bv}\,;\,\hist_t)&\geq  \, (1-\beta)\,\widetilde{\E}^{\bbs_{\bv}}\left[\2{\bar{\kappa}>k_t}\sum_{\ell=(k_t+1) L}^{L \bar{\kappa}-1} \beta^{\ell-k_t L} u^i(w^i_{\ell},w^{-i}_{\ell})\,|\,\hist_t\right]\nonumber\\
&=  (1-\beta)\,\widetilde{\E}^{\bbs_{\bv}}\parentheseDeux{\2{\bar{\kappa}>k_t}\,\sum_{k=k_t+1}^{\bar{\kappa}-1}\beta^{(k-k_t)L}\sum_{\ell=0}^{L-1} \beta^\ell u^i(w^i_{kL+\ell},w^{-i}_{kL+\ell})\,|\,\hist_t}\nonumber\\
    &= (1-\beta) \,\widetilde{\E}^{\bbs_{\bv}}\parentheseDeux{\2{\bar{\kappa}>k_t}\,{\sum_{k=1}^{\bar{\kappa}-k_t-1}\beta^{kL}\sum_{\ell=0}^{L-1} \beta^\ell v^i}} \nonumber\\
    &= (1-\beta)\,\widetilde{\E}^{\bbs_{\bv}}\parentheseDeux{\2{\bar{\kappa}>k_t}\,\sum_{k=1}^{\bar{\kappa}-k_t-1}\beta^{kL}\frac{1-\beta^L }{1-\beta}v^i}\nonumber\\
    &= \beta \, \E\parentheseDeux{\2{\bar{\kappa}>k_t}\,\beta^L {\frac{1-\beta^{(\bar{\kappa}-k_t-1)L}}{1-\beta^L}(1-\beta^L)\,v^i}}\nonumber\\
&=\beta^{L}\,v^i\,\EE{\2{\bar{\kappa}>k_t}(1-\beta^{L(\bar{\kappa}-k_t -1)})} \eqsp.
\end{align}
By \Cref{lemma:barkappa}, we know that $\bar{\kappa}\sim\mathrm{Geometric}(p_L)$. Consequently, we can compute the expectation term explicitly and obtain
\begin{align}
U^i(\bbs_{\bv}\,;\,\hist_t)&\geq \beta^{L}\,v^i\sum_{k\geq k_t + 1}(1-p_L)^{k-k_t-1}p_L (1-\beta^{L(k-k_t -1)})\nonumber\\
        &= \beta^{L}\,v^i \parenthese{1-p_L \sum_{k\geq0}[(1-p_L)\beta^L]^k}\nonumber\\
        &=\beta^{L}\,v^i\,\parenthese{1-p_L (1-(1-p_L)\beta^L)^{-1}}\nonumber\\
    &  \geq \beta^{L} \,(v^{i} - \,(1-\beta^L)^{-1}\,p_L )\eqsp,\label{equation:lowerboundcoop}
 \end{align}where we used $p_L\in(0,1)$, $v^i\in[0,1]$  and $\beta^L\leq 1$ in the last line.

    We now work on an upper bound for $U^i (s, \bars^{-i}_{\bv}\,;\,\hist_{t})$, where $s\in\cS^i$ is any deviation for player $i$. In what follows, we write $u^i(A^i,w^{-i})=\sum_{a^{-i}}u^i(A^i, A^{-i})w^{-i}[a^{-i}]$ for any $A^i\in\cA^i$. We have
\begin{align}
    U(s,\bars^{-i}_{\bv}\,;\,\hist_{t})&=\E^{(s, \bars^{-i}_{\bv})}\parentheseDeux{\sum_{\ell=t}^{\infty}\beta^{\ell-t}u^i(A^i_\ell , w^{-i}_{\ell})\,|\,\hist_t}.\nonumber \eqsp.
\end{align}
Since $\ell\geq t\geq k_t L$, starting the sum at $k_t L$ results in an upper bound. Moreover, given that $t<(k_t +1)L$,  we have $\beta^{\ell-t}\leq\beta^{\ell-(k_t+1) L }$ for any $\ell\geq t$. Thus
\begin{align}
    U(s,\bars^{-i}_{\bv};\,\hist_{t})&\leq(1-\beta)\E^{(s, \bars^{-i}_{\bv})}\parentheseDeux{\sum_{\ell=k_tL }^{\infty}\beta^{\ell-(k_t +1)L}u^i(A^i_\ell, w^{-i}_{\ell})\,|\,\hist_t}\nonumber\\
    &= (1-\beta)\beta^{-(k_t+1)L}\,\E^{(s,\bars^{-i}_{\bv})}\left[\2{\kappa_{\bphi}>k_t}\parenthese{\sum_{k=k_t}^{\kappa_{\bphi}-1} \sum_{\ell=kL}^{(k+1)L-1} \beta^{\ell} u^i (A^i_\ell, w^{-i}_{\bv}) }\,|\,\hist_t\right] \nonumber\\
        & \quad +(1-\beta)\,\beta^{-(k_t+1)L}\,\E^{(s,\bars^{-i}_{\bv})}\parentheseDeux{\sum_{\ell=\kappa_{\bphi} L }^{(\kappa_{\bphi}+1)L-1}\beta^{\ell} u^i (A^i_{\ell},w^{-i}_{\bv})\,|\,\hist_t}\nonumber\\
    & \quad +(1-\beta)\,\beta^{-(k_t+1)L}\,\E^{(s,\bars^{-i}_{\bv})}\parentheseDeux{\sum_{\ell\geq(\kappa_{\bphi}+1)L}\beta^{\ell} u^i (A^i_\ell, b^{-i})\,|\,\hist_t}\nonumber\\
    & \leq \underbrace{(1-\beta)\beta^{-(k_t+1)L}\,\E^{(s,\bars^{-i}_{\bv})}\parentheseDeux{\2{\kappa_{\bphi}>k_t}\sum_{k=k_t}^{\kappa_{\bphi}-1}\sum_{\ell=kL}^{(k+1)L-1} \beta^\ell u^i (A^i_{\ell},w^{-i}_{\bv})\,|\,\hist_t}}_{(\star\star)}\nonumber\\
    & \quad +(1-\beta^L)\,\EE{\beta^{(\kappa_{\bphi}-k_t-1) L }\,|\,\hist_t}\bar{u}^i+\EE{\beta^{(\kappa_{\bphi}-k_t)L}\,|\,\hist_t}\underline{u}^i\label{eq:withstar}\end{align}

    We now bound $(\star\star)$. Since $\{k<\kappa_{\bphi}\}\subset\ \{\phi^i_k =0\}$, we have $\2{k<\kappa_{\bphi}}=\2{k<\kappa_{\bphi}}\2{\phi^i_k = 0}$. It then follows that
    \begin{align*}
        &\E^{(s,\bars^{-i}_{\bv})}\parentheseDeux{\2{\kappa_{\bphi}>k_t}\sum_{k=k_t}^{\kappa_{\bphi}-1}\sum_{\ell=kL}^{(k+1)L-1}\beta^\ell \,u^i(A^i_\ell, w^{-i}_{\bv})\,|\,\hist_t}\\
        &\quad= \E^{(s,\bars^{-i}_{\bv})}\parentheseDeux{\2{\kappa_{\bphi}>k_t}\sum_{k=k_t}^\infty \2{k<\kappa_{\bphi}}\sum_{\ell=kL}^{(k+1)L-1}\beta^\ell u^i(A^i_\ell, w^{-i}_{\bv})\,|\,\hist_t}\\
        &\quad\leq\E^{(s,\bars^{-i}_{\bv})}\parentheseDeux{\2{\kappa_{\bphi}>k_t}\sum_{k=k_t}^\infty \2{k<\kappa_{\bphi}}\2{\phi^i_k = 0}\sum_{\ell=kL}^{(k+1)L-1}\beta^\ell u^i(A^i_\ell, w^{-i}_{\bv})\,|\,\hist_t}\\
        &\quad\leq\E^{(s,\bars^{-i}_{\bv})}\parentheseDeux{\2{\kappa_{\bphi}>k_t}\sum_{k=k_t}^\infty \2{k<\kappa_{\bphi}}\sum_{\ell=kL}^{(k+1)L-1}\beta^\ell (v^i+ \Delta_L)\,|\,\hist_t}\qquad\text{by  }\Cref{eq:utilitygain}\\
        &\quad\leq\E^{(s,\bars^{-i}_{\bv})}\parentheseDeux{\2{\kappa_{\phi}>k_t}\sum_{k=k_t}^{\kappa_{\bphi}-1}\beta^{kL}\frac{1-\beta^L}{1-\beta}(v^i+\Delta_L)\,|\,\hist_t}\\
        &\leq\beta^{k_t L}\frac{1-\E^{(s,\bars^{-i}_{\bv})}[\beta^{(\kappa_{\bphi}- k_t -1)L}\,|\,\hist_t]}{1-\beta^L}\frac{1-\beta^L}{1-\beta}(v^i +\Delta_L)\eqsp,
    \end{align*}
    and therefore
    $$
(\star\star)\leq \beta^{-L}(1-\E^{(s,\bars^{-i}_{\bv})}[\2{\kappa_{\bphi}>k_t}\beta^{(\kappa_{\bphi}-k_t-1)L}\,|\,\hist_t])(v^i+\Delta^L) \eqsp.
    $$
Thus, plugging this back in \eqref{eq:withstar} gives that
    \begin{align}
    U(s,\bars^{-i}_{\bv}\,;\,\hist_{t})&  \leq \beta^{-L}\,(v^i + \Delta_L) + \,\E^{(s, \bars^{-i}_{\bv})}\parentheseDeux{\beta^{(\kappa_{\bphi}-k_t-1)L}\,|\,\hist_t}\parenthese{(1-\beta^L)\bar{u}^i + \beta^L \underline{u}^i -\beta^{-L}(v^i + \Delta_L)}\nonumber\\
    &\leq \beta^{-L}\,(v^i + \Delta_L) + \E^{(s,\bars^{-i}_{\bv})}\parentheseDeux{\beta^{(\kappa_{\bphi}-k_t-1)L}\,|\,\hist_t}\parenthese{(1-\beta^L)\bar{u}^i + \beta^L \underline{u}^i -(v^i + \Delta_L)} \eqsp.
\end{align}
Since we assume that $\beta^L \geq (\bar{u}^i - v^i -  \Delta_L)(\bar{u}^i - \underline{u^i})^{-1}$, the second term is negative and we are left with
\begin{align}
    U(s,\bars^{-i}_{\bv}\,;\,\hist_{t})&\leq \beta^{-L}(v^i + \Delta_L)\eqsp.\label{equation:upperboundeviation2}
\end{align}
Finally, by \eqref{equation:lowerboundcoop} and \eqref{equation:upperboundeviation2}
\begin{align*}
U^i(s, \bars^{-i}_{\bv}\,;\,\hist_{t})&\leq \beta^{-L}(v^i +\Delta_L) \\
&= \beta^{L}v^i - \beta^{L}\frac{p_L}{1-\beta^L}+\beta^{-L}(1-\beta^{2L})v^i +\beta^{-L}\Delta_L + \beta^{L}\frac{p_L}{1-\beta^L}\\
&\leq U^i(\bbs_{\bv}\,;\,\hist_{t})+\beta^{-L}(1-\beta^{2L})v^i +\beta^{-L}\Delta_L + \beta^{L}\frac{p_L}{1-\beta^L}\\
&\leq U^i(\bbs_{\bv}\,;\,\hist_{t})+\beta^{-L}(1-\beta^{2L}+\Delta_L ) + \frac{p_L}{1-\beta^L}
\end{align*}Since  $\hist_{t}\in\mathsf{H}_t$, with $\P^{\bbs_{\bv}}(\hist_t\in\mathsf{H}_t)\geq 1 -q_L$, this proves that $\bbs_{\bv}$ is a ($\beta^{-L}(1-\beta^{2L}+\Delta_L) + p_L(1-\beta^L)^{-1}\,,\,q_L)$)-HP-SPNE.

It remains to prove point (1). By the exact same lines of computation as those leading to \eqref{equation:lowerboundcoop} starting from $t=0$, we obtain $U^i (\bbs_{\bv}\,;\, \hist_t)\geq \beta^{L}(1-p_L(1-\beta^L)^{-1})v^i$.  Regarding the upper bound, we note that by assumption, $v^i\geq\underline{u}^i$, and therefore
\begin{align*}
        U^i (\bbs_{\bv})&=(1-\beta)\,\E^{\bbs_{\bv}}\parentheseDeux{\sum_{\ell=0}^{(\kappa_{\bphi}+1)L-1}\beta^{\ell}v^i+\sum_{\ell=(\kappa_{\bphi}+1)L}^{\infty}\beta^{\ell}\underline{u}^i}\\
        &\leq (1-\beta)\,\E^{\bbs_{\bv}}\parentheseDeux{\sum_{\ell=0}^{(\kappa_{\bphi}+1)L-1}\beta^{\ell}v^i+\sum_{\ell=(\kappa_{\bphi}+1)L}^{\infty}\beta^{\ell}v^i}\\
        &= v^i\eqsp.
\end{align*}This concludes the proof. 

\end{proof}

In the following lemma, for $i\in[N]$ we denote by $U^i_k = u^i (a^i_k , w^{-i}_{\bv})=\sum_{a^{-i}}u^i (a^i_k, a^{-i})w^{-i}[a^{-i}]$ the utility associated to action $a^i_k \in\cA^i$ (when opponents play $w^{-i}_{\bv}$) for player $i$. For ease of notation, we consider that $U^i_1\geq\ldots\geq U^i_K$ and identify an action $a^i_k$ with its index $k$. For any action indexes $\bA=(A_0, \ldots, A_{L-1})\in[K]^L$, we define
$$
V^i (\bA)=\sum_{s=0}^{L-1}\beta^s \, U^i_{A_s}\eqsp.
$$
Likewise, for any $\sigma\in\Sigma(\bA)$, where $\Sigma(\bA)$ denotes the set of permutations over $\bA$, we define $$V^i _L (\sigma, \bA)=\sum_{s=0}^{L-1}\beta^{s}U^i_{\sigma(A_s)}\eqsp.$$
\begin{restatable}{lemma}{lemmaworstinstance}\label{lemma:worst_instance}
 Consider any even $L \in \N$ and sequence $\bA = (A^i_0,\dots,A^i_{L-1})\in[K]^L$. we have that
\begin{align*}
    \max_{\sigma \in \Sigma(\bA)} V^i(\sigma ,\bA)  &\leq V^i ( \bA) + \frac{(1-\beta^L)^2}{1-\beta} \eqsp.
\end{align*}
\end{restatable}

\begin{proof}[Proof of \Cref{lemma:worst_instance}]
We fix a sequence $\bA$ of length $L$ with entries drawn from $[K]$. For ease of presentation, we drop the dependency on $i$ the following proof. 

First, assume $L$ is even. Since $\beta \in (0,1), (\beta^t)_{t \in \{0,\ldots,L-1\}}$ is decreasing. Hence, the ordering that maximizes the utility is to sort the $U_{A_t}$ nonincreasing order, and we define $\sigma^\downarrow$ the associated permutation. Similarly, the ordering that minmizes the utility is to sort the $U_{A_t}$ nondecreasing order, and we define $\sigma^\uparrow$ the associated permutation. Thus
\begin{align*}
\max_{\sigma\in\Sigma(\bA)}V^i(\sigma, \bA) - V^i (\bA)&\leq \max_{\sigma,\sigma'\in\Sigma(\bA)} V^i (\sigma, \bA)-V^i(\sigma', \bA)\\
& =\sum_{t=0}^{L-1} \beta^t\, U_{\sigma^{\downarrow}(A_t)} - \sum_{t=0}^{L-1} \beta^t\, U_{\sigma^{\uparrow}(A_t)} \\
& =\sum_{t=0}^{L-1} \beta^t (U_{\sigma^\downarrow(A_t)} -U_{\sigma^\uparrow(A_t)}) \eqsp.
\end{align*}
Since $\sigma^\uparrow(A_{t}) = \sigma^\downarrow(A_{L+1-t})$, it can be written as
\begin{align*}
\max_{\sigma,\sigma'\in\Sigma(\bA)} \big(V^i(\sigma,\bA)-V^i(\sigma', \bA)\big) &= \sum_{t=0}^{L/2} (\beta^t - \beta^{L+1-t})\, (U_{\sigma^{\downarrow}(A_t)}) -U_{\sigma^{\downarrow}(A_{L+1-t})}) \eqsp,
\end{align*}
where we paired the symmetric indices. By definition, for any $t \in [L/2]$, we have that $U_{\sigma^\downarrow(A_t)} -U_{\sigma^\downarrow(A_{L+1-t})} \leq U_1 - U_K\leq 1$, and thus
\begin{align*}
\max_{\sigma,\sigma'\in\Sigma(\bA)} \big(V^i(\sigma,\bA)-V^i( \sigma',\bA)\big) &  \leq  \sum_{t=1}^{L/2} (\beta^t - \beta^{L+1-t})  = (1- \beta^{L/2}) \sum_{t=0}^{L/2}\beta^t \\
    & = (1- \beta^{L/2}) (1-\beta^{L/2})/(1-\beta) \\
    & \leq \frac{(1- \beta^{L/2})^2}{1-\beta} \leq \frac{(1-\beta^L)^2}{1-\beta} \eqsp.
\end{align*}Note that the proof directly extends to the case where $L$ is odd by counting twice in the sum the term related to the action $A_{(L-1)/2}$ played mid-batch.
\end{proof}

\requirementsbatch*

\begin{proof}
    We start by proving point (1).
    Let $t \in\mathbb{N}$, $L>0$ and denote $k_t = \lfloor t / L\rfloor$. We aim to prove that there exists $\mathsf{H}_t\subset\cH_{t}$ such that
\begin{align*}
\P^{\bbs_{\bv}}(\hist_t\in\mathsf{H}_t)\geq 1 -2KN\exp\parenthese{-2\frac{L\delta^2}{K^2}}\eqsp,
\end{align*}
and for any $\hist_{t}\in\mathsf{H}_t$ and $k\geq k_t$,
\begin{align*}
\P^{\bbs_{\bv}}(\kappa_{\bphi}\geq k+1\,|\,\kappa_{\bphi}\geq k\,,\,\hist_t)\geq 1 -2KN\exp\parenthese{-2\frac{L\delta^2}{K^2}}\eqsp.
\end{align*}
In what follows, we define $t_0 = t - k_t L$. We also define the vectors $w^{i\leftarrow}_{t_0}\in[0,1]^{t_0}$ and $w^{i\rightarrow}_{t_0}\in[0,1]^{L-t_{0}}$ satisfying for any $\ell\in[K]$ and $i\in[N]$
$$
w^{i\leftarrow} _{t_0 ,\ell} = t_0 ^{-1}\sum_{s=0}^{t_0-1}\2{A^i_{k_t L + s} =a_\ell}\quad\text{and}\quad w^{i\rightarrow} _{t_{0},\ell} = (L-t_0)^{-1}\sum_{s=t_0}^{L-1}\2{A^i _{k_t L + s} = a_\ell}\eqsp,
$$
which are respectively the empirical frequency of action $a_k$ from the start of the batch $k_t$ to $t_0 - 1$, and from $t_0$ to the end of the batch $k_t$ for player $i$. Fix $q_L \in(0,1)$ and define $B\subset\cA^{\N}$ as follows
$$B=B_0 \times B_1 \times \ldots\times B_{k_{t}L-1}\times \bar{B}\ldots\eqsp,$$
where $B_k = \cA$ for any $t\in\N \setminus\{k_{t} L,\ldots,k_{t}L+t_0-1\}$ and
\begin{equation}
    \label{eq:firstpartnormone}
\bar{B}=\defEns{ (A^i_{k_{t}L+s})_{\substack{1\leq i \leq N \\ 0\leq s\leq t_0 -1}} \colon \: \text{for any }\ell\in[K]\text{ and }i\in[N] \, , \,\abs{w^{i\leftarrow}_{t_0,\ell}-w^i_{\bv}[a^i_{\ell}]}\leq \sqrt{\frac{\ln(2KN/q_{L})}{2t_0}}}\eqsp,
\end{equation}$\bar{B}$ is the set  of actions from round $k_t L$ to round $k_t L+t_{0}-1$ such that  the empirical frequency of each action $\ell\in[K]$ for each player $i\in[N]$ over  $\{k_t L ,\ldots,k_t L + t_0\} $ is not too far  from $w^i_{\bv}[a^i_{\ell}]$. Observe that applying Hoeffding's inequality to the independent and bounded variables $(\2{A^i_{k_t L +s}=a_\ell})_{s\in\{0,\ldots,t_{0}-1\}}$ for each $(\ell, i)\in[K]\times[N]$ along with a union bound over all pairs $(\ell, i)$ gives $\P^{\bbs_{\bv}} (\bar{B})\geq 1 -q_L$, and it follows by independence of draws from $\bw_{\bv}$ that $$\P^{\bbs_{\bv}}(B)=\P(B_0)\times \ldots\times P(B_{k_{t}L-1})\times \P(\bar{B}) = 1\,.\,\P(\bar{B})\geq 1-q_L\eqsp.$$ We  show that the following set of histories satisfies the condition we are aiming for
\begin{equation}
    \label{eq:historyset}
    \mathsf{H}_t=\{\kappa_{\bphi}<k_t\}\cup \{\,\{\kappa_{\bphi}\geq k_t\}\cap B \,\}\eqsp,
\end{equation}where we recall that $\kappa_{\bphi}=\min_{i\in[N]}\kappa_{\bphi^i}$. To see why $\P^{\bbs_{\bv}}(\mathsf{H}_t)\geq 1 -q_L$, observe that  $\mathsf{H}_t^{c}=\{\{\kappa_{\bphi}\geq k_{t}\}\cap\{\kappa_{\bphi}<k_{t}\}\}\cup \{\{\kappa_{\bphi}\geq k_{t}\}\cap B\}$ and it follows that
\begin{align*}
\P^{\bbs_{\bv}}(\mathsf{H}^{c}_t)&=\P^{\bbs_{\bv}}(\{\{\kappa_{\bphi}< k_{t}\}\cap\{\kappa_{\bphi}\geq k_{t}\}\}\cup \{\{\kappa_{\bphi}\geq k_{t}\}\cap B^c\})\\
&\leq 0+ \P^{\bbs_{\bv}}(\{\kappa_{\bphi}\geq k_{t}\}\,\cap\,B^c)\leq q_L\eqsp,
\end{align*}where the last inequality stems from the fact that $\P^{\bbs_{\bv}}(B )\geq 1 - q_L$.

Let $\hist_t \in\mathsf{H}_t$. We now prove that with $q_L=p_L=2KN\exp(-L\delta^2/K^2)$, $$\P^{\bbs_{\bv}}(\kappa_{\bphi}\geq k +1\,|\,\kappa_{\bphi}\geq k\,,\,\hist_t)\geq 1-p_L\eqsp,$$ for any $k\geq k_t$. If $\hist_t \subset \{\kappa_{\bphi}< k_t\}$ then the statement is trivially true since $\{\kappa_{\bphi}\geq k \}\cap\hist_t =\emptyset$. 

We therefore focus on the case $\hist_t \subset\{k_t \leq \kappa_{\bphi}\}\cap B$. Let $k\geq 0$ such that $k_t\leq k \leq\kappa_{\bphi}$. First, assume $k=k_t$. For any $i\in[N]$, we have
\begin{align}
\normone{\hat{w}^i_{k_{t}} - w^i_{\bv}}&=\sum_{\ell=1}^{K}\abs{\frac{1}{L}\sum_{s=0}^{L-1} \2{A^i _s = a_\ell}-w^i_{\bv}[a_{\ell}]}\nonumber\\
&=\sum_{\ell=1}^{K}\abs{\frac{t_0}{L}w^{i\leftarrow}_{t_{0},\ell}+\frac{L-t_0}{L}w^{i\rightarrow}_{t_{0},\ell}-\parenthese{\frac{t_0}{L}w^i_{\bv}[a_{\ell}]+\frac{L-t_0}{L}w^i_{\bv}[a_{\ell}]}}\nonumber\\
&\leq \sum_{\ell=1}^{K}\parenthese{\frac{t_0}{L}\abs{w^{i\leftarrow}_{t_{0},\ell}-w^i_{\bv}[a_{\ell}]}+\frac{L-t_0}{L}\abs{w^{i\rightarrow}_{t_{0},\ell}-w^i_{\bv}[a_{\ell}]}}\nonumber \eqsp.
\end{align}
Since $\hist_t \in B$
\begin{align}
\normone{\hat{w}^i_{k_{t}} - w^i_{\bv}} & \leq \sum_{\ell=1}^{K}\parenthese{\frac{t_0}{L}\sqrt{\frac{\ln(2KN/q)}{2t_0}}+\frac{L-t_0}{L}\abs{w^{i\rightarrow}_{t_{0},\ell}-w^i _{\bv}[a_{\ell}]}}\nonumber \eqsp.
\end{align}
Thus, applying Hoeffding inequality  to $|w^{i\rightarrow}_{t_{0},\ell}-w^i_{\bv}[a_{\ell}]|$ for any $\ell\in[K]$ and $i\in[K]$, along with an union bound on all pairs $(\ell, i)$ gives
\begin{align}
\P^{\bbs_{\bv}}\parenthese{\bigcap_{i\in[N]}\defEns{\lVert\hat{w}^i_{k_{t}} -w^i_{\bv}\rVert_1 \leq \sum_{k=1}^{K}\parenthese{\frac{t_0}{L}\sqrt{\frac{\ln(2KN/q)}{2t_0}}+\frac{L-t_0}{L}\sqrt{\frac{\ln(2KN/q)}{2(L-t_0)}}}}\,\left.|\right.\kappa_{\bphi}\geq k_t\,,\,\hist_t}&\geq 1- q_L\eqsp, \nonumber 
\end{align}
Each summand in the above sum is maximized when $t_0=L/2$, so we also have
\begin{align}
\P^{\bbs_{\bv}}\parenthese{\bigcap_{i\in[N]}\defEns{\lVert\hat{w}^i_{k_t} -w^i_{\bv}\rVert \leq K\sqrt{\frac{\ln(2KN/q)}{L}}}\,|\,\kappa_{\bphi}\geq k_{t}\,,\,\hist_t}&\geq 1- q_L\eqsp,\nonumber
\end{align}
Then, picking $q_L = 2KN\exp\parenthese{-2\frac{L\delta^2}{K^2}}$ yields
\begin{align}
\P^{\bbs_{\bv}}\parenthese{\bigcap_{i\in[N]}\defEns{\normone{\hat{w}^i_{k_t}-w^i_{\bv} }\leq\delta}\,|\,\kappa_{\bphi}\geq k_t\,,\,\hist_t}\geq 1 - 2KN\exp\parenthese{-\frac{L\delta^2}{K^2}}\eqsp.\label{eq:intersection}
\end{align}
Since conditionally on $\{\kappa_{\bphi}\geq k_t\}$, $\{\kappa_{\bphi}\geq k_t +1\}=\cap_{i\in[N]}\{\phi^i_{k_{t}}=0\}=\cap_{i\in[N]}\defEns{\normone{\hat{w}^i_{k_t}-w^i_{\bv} }\leq\delta}$, we have by \eqref{eq:intersection}
\begin{align*}
\P^{\bbs_{\bv}}\parenthese{\kappa_{\bphi}\geq k_t +1\,|\,\kappa_{\bphi}\geq k_t \,,\,\hist_t}&=\P^{\bbs_{\bv}}\parenthese{\bigcap_{i\in[N]}\{\phi^i_{k_t}=0\}\,|\,\kappa_{\bphi}\geq k\,,\,\hist_t}\\
&=\P^{\bbs_{\bv}}\parenthese{\bigcap_{i\in[N]}\defEns{\normone{\hat{w}^i_{k_t}-w^i_{\bv} }\leq\delta}\,|\,\kappa_{\bphi}\geq k_t\,,\,\hist_t}\geq 1-q_L\eqsp,
\end{align*}which establishes the result.

Now, assume $k>k_t$. Applying Hoeffding's inequality to the bounded, independent random variables $\{\1\{A^{i} _s=a_\ell\}\}_{s\in\cB_k}$ for any $\ell\in[K]$,

$$
\P^{\bbs_{\bv}}\parenthese{\normone{\hat{w}^i_{k,\ell} - w^i_{\bv}[a^i_{\ell}]}\geq \sqrt{\frac{\ln(2KN/q_L)}{2t_0}}\,|\,\hist_t\,,\,\kappa_{\bphi}\geq k}\leq \frac{q_L}{KN}\eqsp,
$$where we removed the conditioning on $\hist_t$ because $(A^i _s)_{s\in\cB_k}$ are independent from $\hist_t$ (because $k>k_t$). From there, a union bound on all pairs $(i,\ell)$ and the same reasoning as before yields the desired result. 

 We now prove point 2.  For ease of notation, we drop the $i$ index. In what fellows, we define
    $$
\xi^i = \parenthese{\sum_{a^{-i}}u^i(a_1,a^{-i})w^{-i}_{\bv}[a^{-i}],\ldots,\sum_{a^{-i}}u^i(a_K,a^{-i})w^{-i}_{\bv}[a^{-i}]}\eqsp,
    $$so for any $w\in\Delta(\cA^i)$, $u^i (w,w^{-i}_{\bv})=\langle w,\xi^i\rangle.$We also introduce the following notations
    \begin{align*}
        Z=\sum_{t=0}^{L-1}\beta^t
        &\quad\text{and}\quad \hat{w}^{(\beta)}_{k} = Z^{-1}\parenthese{\sum_{t=0}^{L-1}\beta^t \2{A_{kL+t} = a_1},\ldots,\sum_{t=0}^{L-1}\beta^t \2{A_{kL+t} = a_K}}\eqsp.
    \end{align*}
    Note that $w^{(\beta)}_{k}$ is a  probability measure over $\{a_1,\ldots,a_K\}$. Finally, for any permutation $\sigma \colon \{1,\ldots,L\}\rightarrow\{1,\ldots,L\}$ over actions taken in batch $k$, we denot by
$$
\hat{w}^i_{k,\sigma} = L^{-1}\parenthese{\sum_{s=0}^{L-1} \2{A^i _{kL + \sigma(1)}=a_1},\ldots, \sum_{s=0}^{L-1} \2{A^i _{kL + \sigma(L)}=a_K}}\eqsp.
$$Observe that $\hat{w}^i _k = \hat{w}^i _{k,\sigma}$ and therefore $\2{\normone{\hat{w}^i _{k,\sigma}-\tilde{w}^i}\leq \delta}=\2{\normone{\hat{w}^i _{k}-\tilde{w}^i}\leq \delta}$. In other words, player $i$ can shuffle actions within any batch to increase its payoff without being detected. Observe that
    \begin{align}
    \2{\normone{\hat{w}^i_k - &w^i_{\bv}}\leq\delta}\parenthese{\sum_{t=kL}^{(k+1)L-1}\beta^t u^i (A^i_{t}, w^{-i}_{\bv})-\sum_{t=kL}^{(k+1)L-1}\beta^t u^i (w^i_{\bv}, w^{-i}_{\bv})}
        \nonumber\\
        &=\2{\normone{\hat{w}^i_k - w^i_{\bv}}\leq\delta}\,\beta^{kL}\parenthese{\sum_{t=0}^{L-1}\beta^t u^i (A^i_{kL+t}, w^{-i}_{\bv})-\sum_{t=0}^{L-1}\beta^t u^i (w^i_{\bv}, w^{-i}_{\bv})}
        \nonumber\\
        &\leq \max_{\sigma}\defEns{\2{\normone{\hat{w}^i_{\sigma,k} - w^i_{\bv}}\leq\delta}\,\beta^{kL}\parenthese{\sum_{t=0}^{L-1}\beta^t u^i (A^i_{kL+\sigma(t)}, w^{-i}_{\bv})-\sum_{t=0}^{L-1}\beta^t u^i (w^i_{\bv}, w^{-i}_{\bv})}}\nonumber\\
        &= \2{\normone{\hat{w}^i_{k} - w^i_{\bv}}\leq\delta}\,\beta^{kL}\parenthese{\max_{\sigma}\defEns{\sum_{t=0}^{L-1}\beta^t u^i (A^i_{kL+\sigma(t)}, w^{-i}_{\bv})}-\sum_{t=0}^{L-1}\beta^t u^i (w^i_{\bv}, w^{-i}_{\bv})}\nonumber\\
        &\leq \2{\normone{\hat{w}^i_{k} - w^i_{\bv}}\leq\delta}\,\beta^{kL}\parenthese{\sum_{t=0}^{L-1}\beta^t u^i (A^i_{kL+t}, w^{-i}_{\bv})+\frac{(1-\beta^L)^2}{1-\beta}-\sum_{t=0}^{L-1}\beta^t u^i (w^i_{\bv}, w^{-i}_{\bv})}\quad \text{by \Cref{lemma:worst_instance}}\nonumber\\
        &=\2{\normone{\hat{w}^i_k - w^i_{\bv}}\leq\delta}\beta^{kL}\,Z\,\parenthese{\langle \hat{w}^{(\beta)}_{k} - w^i_{\bv}\,,\,\xi^i\rangle+Z^{-1}\frac{(1-\beta^L)^2}{1-\beta}}\nonumber \\
        &\leq \2{\normone{\hat{w}^i_k - w^i_{\bv}}\leq\delta}\beta^{kL}\,Z\parenthese{\normone{\hat{w}^{(\beta)}_{k}-w^{i}_{\bv}}\lVert{\xi^i}\rVert_{\infty}+(1-\beta^L)}\nonumber\\
        &\leq \2{\normone{\hat{w}^i_k - w^i_{\bv}}\leq\delta}\beta^{kL}\,Z\parenthese{\normone{\hat{w}^{(\beta)}_{k}-w^{i}_{\bv}}+(1-\beta^L)}\qquad\text{(because $u^i(\ba)\leq 1$ for any $\ba\in\cA$)}\nonumber\\       &\leq\2{\normone{\hat{w}^i_{k}-w^{i}_{\bv}}\leq\delta}Z\,\beta^{kL}\,\parenthese{\normone{\hat{w}^{(\beta)}_{k}-\hat{w}}+\normone{\hat{w}-w^i_{\bv}}+(1-\beta^L)}\nonumber\\
       &\leq Z\,\beta^{kL}\,\parenthese{\sum_{t=0}^{L-1} \abs{\frac{\beta^t}{Z}-\frac{1}{L}}+\delta +(1-\beta^L)}\label{eq:normhattrue}\eqsp.
    \end{align}
    We now bound $\sum_{t=0}^{L-1}\left|\frac{\beta^t}{Z} - \frac{1}{L}\right|$. For each $t\in\{0,\ldots,L-1\}$, applying the triangle inequality gives
$$
\left|\frac{\beta^t}{Z} - \frac{1}{L}\right|
\leq \left|\frac{\beta^t}{Z} - \frac{\beta^t}{L}\right|
   + \left|\frac{\beta^t}{L} - \frac{1}{L}\right|
= \beta^t\left|\frac{1}{Z}-\frac{1}{L}\right|
  + \frac{1-\beta^t}{L}\eqsp.
$$
Then, summing over $t=0,\dots,L-1$ yields
$$
\sum_{t=0}^{L-1}\left|\frac{\beta^t}{Z} - \frac{1}{L}\right|
\leq \left|\frac{1}{Z}-\frac{1}{L}\right|\sum_{t=0}^{L-1}\beta^t
   + \frac{1}{L}\sum_{t=0}^{L-1}(1-\beta^t)\leq \left|1-\frac{Z}{L}\right| + \frac{L-Z}{L}\eqsp,
$$
Because $\beta^t \leq 1$ for all $t$, we have $Z \leq L$, and therefore $
\left|1-\frac{Z}{L}\right| = 1-\frac{Z}{L} = \frac{L-Z}{L}\eqsp.
$ Hence
$$
\sum_{t=0}^{L-1}\left|\frac{\beta^t}{Z} - \frac{1}{L}\right|
\leq 2\,\frac{L-Z}{L}\eqsp.
$$
Finally, observe that
$$
L-Z = \sum_{t=0}^{L-1}(1-\beta^t)
\leq \sum_{t=0}^{L-1}(1-\beta^{L})
= L(1-\beta^{L})\eqsp,
$$
where we used $\beta^t \geq \beta^L$ for $0\leq t \leq L-1$. Therefore
\begin{align*}
\sum_{t=0}^{L-1}\left|\frac{\beta^t}{Z} - \frac{1}{L}\right|
\leq 2(1-\beta^{L})\eqsp.
\end{align*}
Plugging this in \eqref{eq:normhattrue} yields
\begin{align}
        \2{\normone{\hat{w}^i_k - &w^i_{\bv}}\leq\delta}\parenthese{\sum_{t=kL}^{(k+1)L-1}\beta^t u^i (A^i_{t}, w^{-i}_{\bv})-\sum_{t=kL}^{(k+1)L-1}\beta^t u^i (w^i_{\bv}, w^{-i}_{\bv})}
        \nonumber\\
        &\leq Z\,\beta^{kL}\parenthese{2(1-\beta^L)+\delta+\frac{(1-\beta^L)^2}{1-\beta}}\nonumber\\
        &=\sum_{t=kL}^{(k+1)L-1}\beta^t (2(1-\beta^L)+\delta+(1-\beta^L))\eqsp,
\end{align}
which establishes the result.   
\end{proof}

\corollaryfolkbatch*

\begin{proof}[Proof of \Cref{corollary:folk_batch}]
Fix $\varepsilon\in(0,1]$. By assumption and with our parameters, we have
\begin{equation}
\label{eq:valeursdeltal_alt_v3}
1-\frac{\varepsilon}{16}\leq \beta^L\leq1-\frac{\varepsilon}{32} \; ,
\qquad
\delta=\frac{\varepsilon}{16} \; ,
\qquad
L=\frac{256K^2}{\varepsilon^2}\ln\left(\frac{128K}{\varepsilon^2}\right) \eqsp.
\end{equation}
Recall that
\begin{align*}
\Delta_L=\delta+3(1-\beta^L) \; ,\qquad
p_L=2KN\exp\left(-\frac{L\delta^2}{K^2}\right) \eqsp.
\end{align*}
We first bound $1-\beta^{2L}+\Delta_L$. Using $1-\beta^{2L}=(1-\beta^L)(1+\beta^L)\leq 2(1-\beta^L)$, we get
\begin{equation}
\label{eq:firstterm_bound1_v3}
1-\beta^{2L}+\Delta_L
\leq 2(1-\beta^L)+\delta+3(1-\beta^L)
=\delta+5(1-\beta^L)\eqsp.
\end{equation}
Moreover, from \eqref{eq:valeursdeltal_alt_v3} we have $1-\beta^L\leq \varepsilon/16$ and $\delta=\varepsilon/16$, hence
\begin{equation}
\label{eq:firstterm_bound2_v3}
1-\beta^{2L}+\Delta_L
\leq \frac{\varepsilon}{16}+\frac{5\varepsilon}{16}
=\frac{3\varepsilon}{8} \eqsp.
\end{equation}
We now bound the prefactor $\beta^{-L}$. Since $\beta^L\geq 1-\varepsilon/16$, we immediately obtain
\begin{equation}
\label{eq:beta_prefactor_v3}
\beta^{-L}\leq \frac{1}{1-\varepsilon/16}\leq \frac{16}{15} \eqsp,
\end{equation}
where the last inequality uses $\varepsilon\leq 1$. Combining \eqref{eq:firstterm_bound2_v3} and \eqref{eq:beta_prefactor_v3} yields
\begin{equation}
\label{eq:easy_part_alt_v3}
\beta^{-L}\bigl(1-\beta^{2L}+\Delta_L\bigr)
\leq \frac{16}{15}\cdot\frac{3\varepsilon}{8}
=\frac{2\varepsilon}{5} \eqsp.
\end{equation}
We now bound $p_L/(1-\beta^L)$: writing $s=2KN\exp(-L\delta^2/K^2)$, and using the fact that $1-\beta^L \geq \varepsilon/32$ gives
\begin{equation}
\label{eq:denom_lower_alt2_v3}
\frac{p_L}{1-\beta^L}=\frac{s}{1-\beta^L}\leq \frac{32s}{\varepsilon}\eqsp.
\end{equation}
Now compute $s$ using our choice of $L$ and $\delta$
\begin{align*}
s=2KN\exp \left(-\frac{L\delta^2}{K^2}\right)
=
2KN\exp \left(-\frac{L}{K^2}\cdot\frac{\varepsilon^2}{256}\right)
=
2KN\exp \left(-\ln \left(\frac{128KN}{\varepsilon^2}\right)\right)
=
\frac{\varepsilon^2}{64} \eqsp.
\end{align*}
Plugging this into \eqref{eq:denom_lower_alt2_v3} yields
\begin{equation}
\label{eq:hard_part_alt_v3}
\frac{p_L}{1-\beta^L}
\leq
\frac{32}{\varepsilon}\cdot\frac{\varepsilon^2}{64} = \frac{\varepsilon}{2}\eqsp.
\end{equation}
Thus, by \eqref{eq:easy_part_alt_v3} and \eqref{eq:hard_part_alt_v3}, we obtain
\begin{align*}
\beta^{-L}(1-\beta^{2L}+\Delta_L)+\frac{p_L}{1-\beta^L}
\leq
\frac{2\varepsilon}{5}+\frac{\varepsilon}{2}
=
\frac{9\varepsilon}{10}
\leq \varepsilon\eqsp.
\end{align*}
Finally, observe that $p_L =\varepsilon^2 / 64\leq \varepsilon$. All in all, we obtain that $\bbs_{\bv}$ is a $(\varepsilon, \varepsilon)-$SPNE.

For point (ii), we seek a lower bound on
\begin{align*}
\beta^{L}\left(1-\frac{p_L}{1-\beta^L}\right)v^i\eqsp.
\end{align*}
From our previous computations, $p_L/(1-\beta^L)\leq \varepsilon/2$, so we have
\begin{align*}
    \beta^{L}\parenthese{1-\frac{p_L}{1-\beta^L}}\,v^i &\geq \parenthese{1-\frac{\varepsilon}{16}} \parenthese{1-\frac{\varepsilon}{2}}\,v^i\\
    &=\parenthese{1-\frac{\varepsilon}{2}-\frac{\varepsilon}{16}+\frac{\varepsilon^2}{32}}v^i\\
    &\geq\parenthese{1-\frac{\varepsilon}{2}-\frac{\varepsilon}{16}}v^i\\
    &= \parenthese{1-\frac{9\varepsilon}{16}}v^i \\
    & \geq (1-\varepsilon)v^i\eqsp.
\end{align*}
\end{proof}

\nocontroltypei*

\begin{proof}
In what follows, we let $$(j,r)\in\operatorname{argmin}\limits_{i\in[N],\ell\in[K]}\,\defEns{w^i_{\bv,\ell},\;w^i_{\bv,\ell}>0}\eqsp.$$Observe that, denoting $\lVert w^j_{\bv}\rVert_0 = \sum_{\ell=1}^K \2{w^j_{\bv, \ell}\ne 0}$,
\begin{equation}
    \label{eq:lowerboundwjr}
    w^j _{\bv,r} \leq 1/\lVert w^j_{\bv}\rVert_0\eqsp,
\end{equation}
since otherwise $\sum_{\ell\in[K]}w^j_{\bv,\ell}\geq\lVert w^j_{\bv}\rVert_0\frac{1}{\lVert w^j_{\bv}\rVert_0}>1$. Moreover,
\begin{align}
    \P^{\bbs_{\bv}}(\kappa_{\bphi}<\infty)=\P^{\bbs_{\bv}}\parenthese{\bigcup_{i\in[N]}\bigcup_{k\geq 0}\defEns{\normone{\hat{w}^i_k -w^i_{\bv}}>\delta}}\geq\P^{\bbs_{\bv}}\parenthese{\bigcup_{k\geq0}\defEns{\normone{\hat{w}^j _k - w^j_{\bv}}>\delta}}\label{eq:bigcup}\eqsp.
\end{align}
Finally, note that for any $k\geq0$, $\cap_{t\in\cB_k}\{A^j_t = a_r\}\subset\{\normone{\hat{w}^j_k - w^j_{\bv}}>\delta\}$ since under this event
\begin{align*}
\normone{\hat{w}^j _k - w^j_{\bv}}=\sum_{\ell=1}^{K}\abs{\hat{w}^j_{k,\ell}-w^j _{\bv, \ell}}=1-w^j_{\bv,r}\geq 1 - 1/\lVert w^j_{\bv}\rVert_0 >\delta\eqsp,
\end{align*}
where the two last inequalities follows from \eqref{eq:lowerboundwjr} and our assumption respectively. Connecting this observation with \eqref{eq:bigcup} then gives
    \begin{align*}
        \P^{\bbs_{\bv}}(\kappa<\infty)\geq\P^{\bbs_{\bv}}\parenthese{\bigcup_{k\geq0}\defEns{\normone{\hat{w}^j_{k}-w^j _{\bv}}>\delta}}&\geq\P^{\bbs_{\bv}}\parenthese{\bigcup_{k\geq 0}\bigcap_{t\in\cB_k}\defEns{A^j_t = a_r}}\\
        &= 1-\P^{\bbs_{\bv}}\parenthese{\bigcap_{k\geq0}\bigcup_{t\in\cB_k}\{A^j_t \ne a_r\}}\\
        &= 1 - \prod_{k=0}^\infty (1-(w^j_{\bv,r})^L) =1\eqsp,
    \end{align*}where we used the independence of draws under $\bbs_{\bv}$ in the last step.
\end{proof}

\end{document}